%% file: arxiv_2026_main.tex
\documentclass[11pt,a4paper]{article}
\usepackage{booktabs} 
\usepackage{natbib} 
\usepackage[margin=1in]{geometry}
\input{macros}

\usepackage{amsmath,amsthm}

\newtheorem{observation}[theorem]{Observation}

\usepackage[margin=1in]{geometry}
\usepackage{xcolor}
\usepackage{hyperref} 
\hypersetup{
	colorlinks,
	linkcolor={blue!70!black},
	citecolor={green!40!black},
	urlcolor={blue!70!black}
}
\usepackage[normalem]{ulem}
\usepackage{lscape}
\usepackage{tikz}
\usepackage{subcaption} 
\usepackage{wrapfig}
\usepackage{multirow}
\usepackage{float}
\usepackage{bbm}
\graphicspath{{arxiv-figs/}}
\newcommand{\F}{\bar{F}}

\newcommand{\invb}{\frac{1}{\beta}}

\newcommand{\invba}{\frac{1}{\beta^\alpha}}
\newcommand{\ba}{\beta^\alpha}

\def\tilde{\widetilde}
\def\hat{\widehat}
\def\epsilon{\varepsilon}

\newcommand{\sbr}[1]{\left[#1\right]}

\newcommand{\mcL}{\mathcal{L}}
\newcommand{\dta}{\delta}

\newcommand{\set}[1]{\left\{#1\right\}}
\newcommand{\piecewise}[1]{\left\{\begin{array}{ll}#1\end{array}\right.}
\newcommand{\indc}{\mathbbm{1}}

\newcommand{\maxz}[1]{\left[#1\right]^+}

\DeclareMathOperator{\distributed}{\sim}
\DeclareMathOperator{\Pareto}{Pareto}

\DeclareMathOperator{\disp}{disp}

\newcommand{\theemptyset}{\emptyset}
\newcommand{\mcS}{\mathcal{S}}
\newcommand{\mcT}{\mathcal{T}}

\newcommand{\mbX}{\mathbb{B}}
\newcommand{\mbb}[1]{\mbX\left(#1\right)}

\makeatletter

\newenvironment{mproof}[1]{%

\smallskip 

\noindent \emph{Proof of #1}
}
{\smallskip}

\makeatother

\newcommand{\ecref}[1]{Appendix~\ref*{#1}}

\title{Reducing the Filtering Effect in Public School Admissions:\\ A Bias-aware Analysis for Targeted Interventions} 

\usepackage{authblk}

\author{Yuri Faenza \thanks{yf2414@columbia.edu}}
\author{Swati Gupta \thanks{swatig@mit.edu,}}
\author{Aapeli Vuorinen \thanks{aapeli.vuorinen@columbia.edu}}
\author{Xuan Zhang \thanks{xz2569@columbia.edu}}
\affil{}
\date{}

\begin{document}

\maketitle

\begin{abstract}

\noindent 
{\bf Problem definition:} Traditionally, New York City's top 8 public schools have selected candidates solely based on their scores in the Specialized High School Admissions Test (SHSAT). These scores are known to be impacted by socioeconomic status of students and test preparation received in middle schools, leading to a massive filtering effect in the education pipeline. The classical mechanisms for assigning students to schools do not naturally address problems like school segregation and class diversity, which have worsened over the years. The scientific community, including policymakers, have reacted by incorporating group-specific quotas and proportionality constraints, with mixed results. The problem of finding effective and fair methods for broadening access to top-notch education is still unsolved.

\noindent 
{\bf Methodology/results:} We take an operations approach to the problem different from most established literature, with the goal of increasing opportunities for students with high economic needs. Using data from the Department of Education (DOE) in New York City, we show that there is a shift in the distribution of scores obtained by students that the DOE classifies as ``disadvantaged'' (following criteria mostly based on economic factors). We model this shift as a ``bias" that results from an underestimation of the true potential of disadvantaged students. 
We analyze the impact this bias has on an assortative matching market. We show that centrally planned interventions can significantly reduce the impact of bias through scholarships or training, when they target the segment of disadvantaged students with average performance. 

\noindent 
{\bf  Managerial implications:}  To make these interventions incentive compatible and individually-fair, we propose a randomization-based policy for allocation of training resources to students, which is heavily targeted towards average performers. Our results challenge existing notions of scholarships in the current education system. We believe that these insights can guide policymakers in answering a critical question: how should one allocate limited funding across schools and students to maximally help disadvantaged students.
\end{abstract}


\maketitle

\section{Introduction}

The disparity in opportunities plays a major role in access to education at different levels of the educational pipeline~\citep{quinn2017implicit}. It is known that outcomes of middle school admissions dictate high school admissions, which in turn impact pathways to higher education \citep{corcoran2018pathways, boschma2016concentration}. Selection starts much earlier however, with gifted and talented programs screening students as young as 4 years old. In particular, our work is motivated by high school admissions in New York City (NYC). NYC has an extensive public school system with current enrollment of over one million students. Every year roughly 80,000 students wish to join one of its 700 high school programs. The most sought after public schools are the so-called Specialized High Schools (SHSs) which, by law, select candidates solely based on their score on the Specialized High School Admissions Test (SHSAT) \citep{DOE}. Such scores are known to be impacted by socioeconomic status of students 
and test preparation received in middle schools~\citep{corcoran2018pathways,Eliza2}.  The results is a massive filtering effect in high school admissions: 50\% (resp. 80\%) of students admitted to the SHSs come from only 5\% (resp. 15\%) of middle schools \citep{corcoran2018pathways}.

The goal of this work is to investigate data-driven interventions at the middle school level to reduce this filtering effect. An extensive literature has focused on doing so by proposing changes to admissions policies themselves (see Section~\ref{sec:literature-review}). That approach however, has multiple downsides. For one, simply ``fixing'' the admissions process to boost under-represented students does not fundamentally prepare those students to perform well once admitted. Another downside of such admissions policies is that they are seen as unfair by many, and there are significant political and legislative hurdles to implementing criteria that take the disparate backgrounds of students into account during the admissions process. In particular, the simple strategy of rescaling scores to remove distributional differences is not legally feasible. For example, in 2003, an attempt by the University of Michigan to add 12 points for ``diversity'' on a 150 point scale in an effort to promote admissions of underrepresented ethnic minorities was met with a lawsuit, which was ultimately decided {\it not} in favor of the university~\citep{gratz}. There is further resistance to removing the SHSAT scores from the admissions criteria: a 2019 plan supported by the then mayor of New York City to eliminate the SHSAT---criticized by some because of the unequal access to test preparation---failed to gain enough support, and was not approved by the New York State Senate~\citep{shapirowang}.

In this work, we take a completely different operations perspective. We focus on centralized pre-admission interventions, a fundamentally meritocratic approach that does not involve an unfair or legally dubious change in the admissions criteria.
We introduce a matching model of schools and students where some students (that we call \emph{disadvantaged}, following a classification used by the Department of Education of New York City\footnote{We remark that our theoretical analysis only assumes that there are two groups of students: advantaged and disadvantaged, where the true potential of disadvantaged is not fully observed. For the purposes of this work, we consider the definition from the Department of Education of NY City on which students should be considered in the ``disadvantaged" group. However, alternately other definitions could be considered as well.}) are not evaluated at their true potential, but at a strictly lower level. We then investigate both theoretically and empirically the impact of such differences in treatment, and propose and evaluate interventions to counter them. These interventions are in the form of \emph{vouchers} targeted at certain disadvantaged students, affording them access to supplemental instruction: thereby providing them real support in order to perform closer to their innate ability. Our main contribution is two-fold: (a) characterizing the impact of the presence of a distributional shift when evaluating of students, both theoretically and empirically, and (b) construction of randomized policies for voucher allocation that are individually fair, incentive compatible and (by carefully targeting certain disadvantaged students) can substantially reduce the mistreatment they experience, as measured by various metrics. We next present the setup, intermediate results, and experiments leading to our main contributions.

\subsection{Contributions}\label{sec:contribs}

In order to present our mathematical model, we first introduce the characteristics and mechanism for SHSs admissions in NYC. SHSs admit students uniquely based on their scores on the highly competitive SHSAT. The NYC Department of Education (DOE) acknowledges that there is a disparity in students' abilities to prepare for the test, and so classifies some students as disadvantaged. This classification uses criteria such as their household income and the middle school they attended, which together constitute a proxy for socioeconomic status~\citep{doe-disadvantaged}.
Following the DOE's definition, we divide students who take the SHSAT into two groups: non-disadvantaged ($G_1$) and disadvantaged ($G_2$). We find that the distribution of SHSAT scores of the two groups in the NYC DOE\footnote{We thank the NYC DOE for providing us this data under a non-disclosure agreement.} data (Figure \ref{fig:dist-shsat-1}) exhibits a significant distributional shift, but matches closely (as measured by Wasserstein distance) if the scores of disadvantaged students are adjusted by a multiplicative factor of $\frac{1}{\beta} \approx \frac{1}{0.88} \approx 1.13$ (Figure \ref{fig:dist-shsat-2}), or an additive factor of $\gamma=49$ points (Figure~\ref{fig:dist-shsat-3}).

\begin{figure}[h]
    \centering
    \begin{subfigure}{.32\textwidth}
        \centering
        \includegraphics[width=\textwidth]{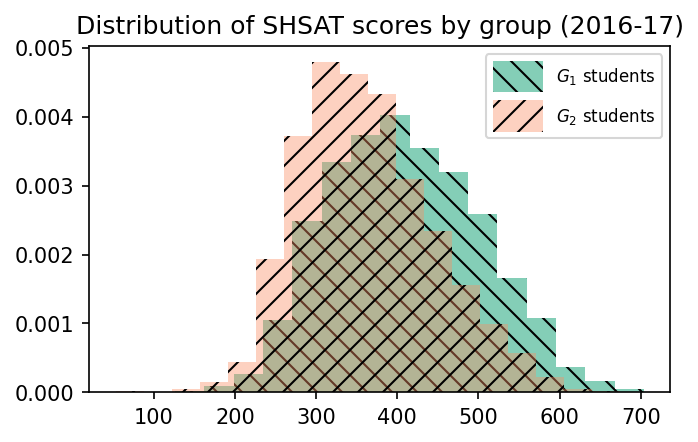}
        \caption{raw SHSAT scores} \label{fig:dist-shsat-1}
    \end{subfigure}
    \vspace{.02\textwidth}
    \begin{subfigure}{.32\textwidth}
        \centering
        \includegraphics[width=\textwidth]{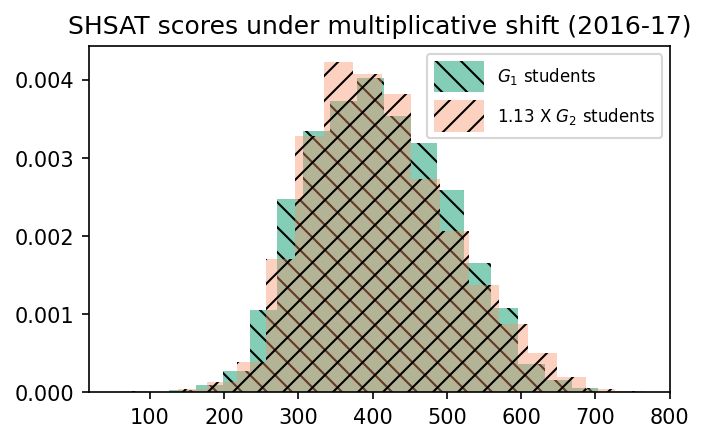}
        \caption{multiplicative shift} \label{fig:dist-shsat-2}
    \end{subfigure}
    \vspace{.02\textwidth}
    \begin{subfigure}{.32\textwidth}
        \centering
        \includegraphics[width=\textwidth]{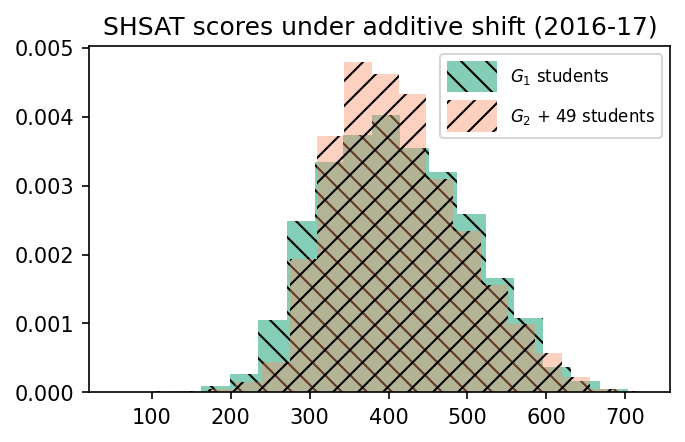}
        \caption{additive shift} \label{fig:dist-shsat-3}
    \end{subfigure}
    \caption[Distribution of SHSAT scores for students in groups $G_1$ and $G_2$ for the 2016--17 academic year.]{Distribution of SHSAT scores for students in groups $G_1$ and $G_2$ for the 2016--17 academic year. Scores between the groups align closely after multiplicative shift of $\beta\approx0.88$, or under an additive shift of $49$ points. 
    } \label{fig:dist-shsat}
\end{figure}

While the exact mechanism that causes distributional shifts in performance of student groups and its causal factors are unknown and debated in the literature\footnote{Various mechanisms for the existence of this disparity have been debated and studied, and it is difficult to establish causality. We appreciate that this is a contentious issue. In our work, we assume that there can be a distributional shift due a large number of factors, and we study the impact of these shifts on the stable matching in the school-student markets.}, the consensus stands that performance gaps between socioeconomic groups stem not from differences in innate ability, but from disadvantages that hinder students' potential~\citep{considine2002influence}. It is therefore natural to postulate that the distribution of innate ability ought to coincide when students are partitioned into groups based on their socioeconomic status.

Motivated by these observations and the literature on performance gaps between socioeconomics groups, we consider the following model. 
The \emph{true potential} (unobserved innate ability) $Z$ of a student is sampled from the Pareto distribution\footnote{The choice of the Pareto distribution to model potentials is inspired by a body of empirical work, see for example~\cite{clauset2009power} on the achievements of individuals in many professions. See Figure \ref{fig:combined-figures}(a) for the fit of the Pareto distribution to the distribution of the postulated true potentials.}, while the \emph{perceived potential} (observed performance) $\hat{Z}$ is equal to $Z$ for $G_1$ students and to $\beta Z$ for $G_2$ students, where $\beta\in (0,1)$ is some \emph{bias factor}. We focus our analysis on this multiplicative bias model due to its computational tractability and fit to the SHSAT score distributions of admitted students (see Figure \ref{fig:dist-shsat}). In~\ecref{sec:additive-bias}, we study the robustness of our qualitative findings under other models of bias such as an additive model where the perceived potential of disadvantaged students is $\beta+Z$.

\begin{figure}[h]
    \centering
    \begin{subfigure}[t]{0.48\textwidth}
        \centering
        \includegraphics[width=\textwidth]{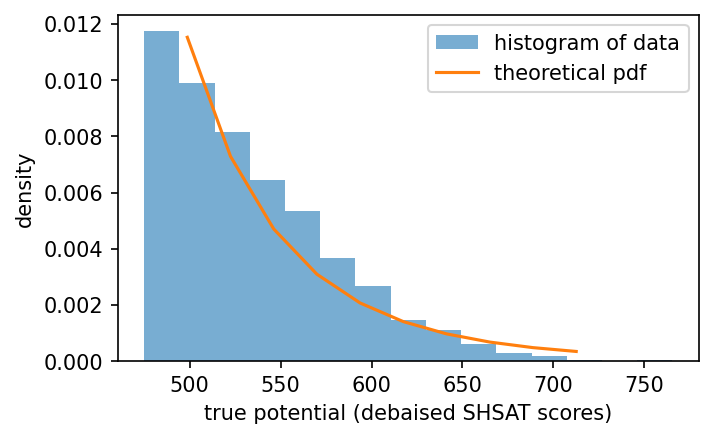}
        \caption[Distribution of estimated true potentials of students.]{\footnotesize Distribution of estimated true potentials of students who score high enough to receive an offer from a SHS under the multiplicative model. The best fitting Pareto distribution has parameter $\alpha=8.9$.}
        \label{fig:pareto-fitting-shsat}
    \end{subfigure}%
    \hfill
    \begin{subfigure}[t]{0.48\textwidth}
        \centering
        \includegraphics[width=\textwidth]{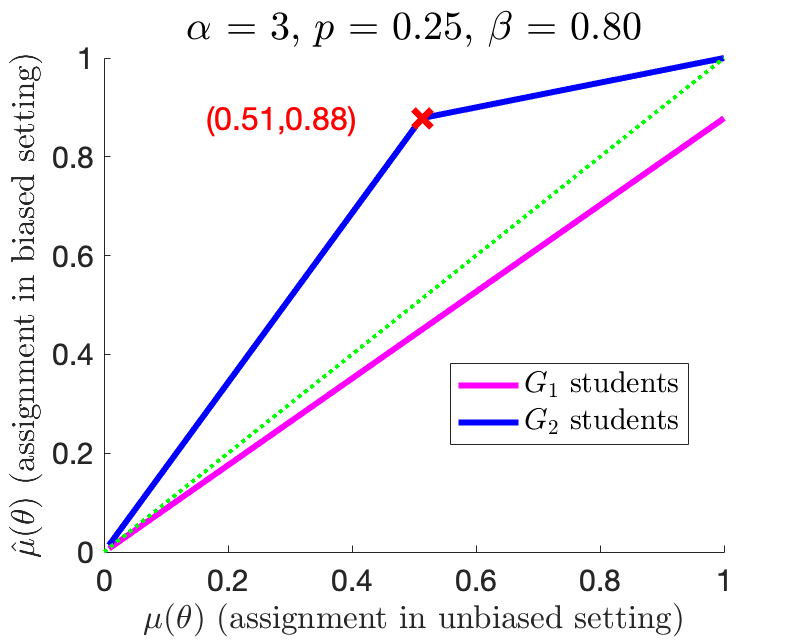}
        \caption[Schools students are matched to under $\hat\mu$ and $\mu$.]{\footnotesize Schools students are matched to under $\hat\mu$ and $\mu$. The green dotted line has slope one, representing the school a student would be placed in if \emph{there were no bias in the system}.}
        \label{fig:should-vs-actual-students}
    \end{subfigure}
    \caption{Comparison of distribution fitting (left) and matching outcomes under bias (right).}
    \label{fig:combined-figures}
\end{figure}

In our model, schools rank students based on their perceived potentials $\hat{Z}$. To be able to find tractable policies, we assume in our theoretical analysis that all students share the same ranking of schools (for example based on US News Rankings). This assumption abstracts out considerations that may be important for students such as proximity of a school~\citep{burgess2015parents}, limits on the length of preference lists~\citep{calsamiglia2010constrained}, or strong preferences of students for certain high schools\footnote{For instance in the 2016--17 SHSAT cohort, 56\% of students indicated Stuyvesant or Brooklyn Tech as their first preference, with 76\% naming at least one of the two in their top two preferences.}. We later show that our qualitative results are robust to relaxations of our stylized model (up to an extent) even when we use stable matchings with heterogeneous preferences in the NYC DOE data.

We then study the impact of interventions that we call \emph{vouchers}, in the form of additional resources (e.g., tutoring, test prep, or a scholarship) allocated to certain disadvantaged students. In our model, such a voucher enables a disadvantaged student to unlock their innate ability and perform at their true potential. We therefore refer to this process as \emph{debiasing}. This model of intervention by additional resources is motivated by real world examples where supplemental instruction such as tutoring has been shown to be effective in closing the performance gap, see for example~\cite{yue2018supplemental}.

However, such interventions are costly, and so they cannot be offered to everyone. This leads to a natural resource allocation question. We investigate the impact that bias has on both disadvantaged and non-disadvantaged students on two opposing measures of aggregate mistreatment, then quantify the effect that these interventions have on the mistreatment, and identify the best populations they should be targeted to. We discuss qualitative results in the real-world context of New York City Specialized High Schools by first estimating the parameters of our model from the data, and then evaluating the effect of various interventions on the actual stable matching outcomes. Our key findings are as follows:

\begin{enumerate}
\item {\bf Asymmetric impact:} We observe that under reasonable assumptions on the parameters (such as 
the number of disadvantaged students being smaller than the advantaged students), the impact of bias on $G_2$ (disadvantaged) students is much bigger than the slight advantage that $G_1$ students obtain (see Figure \ref{fig:combined-figures}(b) and Figure~\ref{fig:doe-disp-voucher}).

\item {\bf Deterministic Centralized Interventions:}
We next study the impact of deterministic voucher allocation. We define the \emph{mistreatment} of a student as the positive increase in rank of the school they get matched to under bias compared to the unbiased setting\footnote{If the rank of the school decreases, we set the mistreatment of the student to 0.}. We study mechanisms to allocate vouchers to reduce aggregate measures of mistreatment, which we interpret as measures of (group) fairness. We first show that under two such very different measures, maximum benefit is achieved by providing vouchers to average-performing (rather than top-performing) disadvantaged students, assuming that the abilities of students are Pareto-distributed. These findings challenge existing scholarship/aid allocation mechanisms, addressing one of the key questions facing policy makers on how to distribute resources, if they care about minimizing maximum or total mistreatment of the student population.

\item {\bf Incentive Compatible Voucher Distribution:} We next observe that the deterministic allocation of vouchers to average performers creates an incentive for some top students to underperform. We show that the only deterministic policy that is incentive compatible distributes vouchers to top students. However, this policy has a small impact on reducing aggregate mistreatment. We therefore discuss two classes of mechanisms for randomly allocating vouchers that are incentive compatible. In particular, one of them (that we term \emph{Proportional to Mistreatment}) guarantees that the maximum expected mistreatment is lower than under any deterministic policy. The other class of random policies we study is incentive compatible under general potential distributions. These policies have the additional benefit of being \emph{individually fair} (in the sense of a Lipschitz condition) so that the probability of receiving a voucher for students with similar potentials is similar. Both the randomized policies we propose heavily target average performers to help minimize mistreatment in expectation.

\item {\bf Alternate Models of Bias:} We next study alternate models to capture the differences in distribution of performance between non-disadvantaged and disadvantaged students, given the broader literature on discrimination \citep{boschma2016concentration, lang2012racial}. 
We quantify the impact of applying our model to contexts where bias arises due to a more sophisticated process, and discuss experimental results quantifying efficiency loss from such model misspecification. We then extend some of our results from the case of uniform multiplicative bias to the case of uniform additive bias. We show that our main takeaways continue to hold under an additive model or when the simple multiplicative model is applied under moderate levels of model misspecification.

\item {\bf Experimental validation:} We validate our theoretical results on admissions data to the SHSs for the academic year 2016--17. Although our model assumes that students have homogeneous preferences, we compute the stable matching using real SHSAT scores and reported heterogeneous preferences of students. We find that our key theoretical takeaways are still valid under relaxation of our model assumptions: for instance, the shape of students' mistreatment resembles the prediction of our model (including the fact that average performers are the most mistreated) and that our voucher distribution program improves the mistreatment across the board. We further show that the best ranges of students to give vouchers to obtained via our theory, are qualitatively similar to the best empirically found ranges under the real data with heterogeneous preferences. This leads to our policy insights, which we discuss next.

\item {\bf Policy Insights:} Motivated by the goal of maximizing the impact of limited resources, in this work we propose that additional training and resources should be offered to average performers rather than top performers. At a high level, the two assumptions that lead to this result are (1) the concentration of students who perform around the average, compared to a much smaller cohort who make up the top performers, and (2) that given enough opportunities and support, the performance of the two cohorts of students ought to be indistinguishable. 
The key phenomenon that arises due to these assumptions is that a small deviation in an average disadvantaged student's perceived performance leads to a significant change in the rank of the school they are matched to.

Note that (1) is a key characteristic of many common distributions, including the Pareto distribution investigated in this paper. (2) is supported by education policy research and discrimination literature which shows that additional resources can positively impact low achieving student groups. 
Our work complements this line of work through mathematical analysis in order to optimally target limited resources. 
We propose two ways of doing this: either through a deterministic allocation of resources, or through a randomized allocation that assigns a higher probability of receiving a voucher for students that experience a ``higher mistreatment'' in the current system. The latter allocation policies (called {\it PropM} policies) are incentive compatible, and recommended when the distribution is known to be close to  Pareto. However, in the case that truly nothing is known about the distribution of student potentials (so, in particular, condition (1) cannot be assumed), we show that the only policies guaranteed to be incentive compatible are those where the probability of receiving vouchers increases with performance. We remark that randomization  is not a unusual feature in public high school admissions, even within the context of NYC DOE. For instance, in certain high schools offers are made in order of random lottery numbers within any priority group (zoned students, siblings, etc.). 
\end{enumerate}

\smallskip

The rest of the paper is organized as follows. In Section \ref{sec:market}, we formally introduce our mathematical model for a continuous matching market with bias. In Section~\ref{sec:analysis} we analyze the effects of bias on both disadvantaged and non-disadvantaged students, introducing the key concepts of \emph{displacement} and \emph{mistreatment}. We then consider deterministic policies for reducing the mistreatment via a centralized approach in Section~\ref{sec:voucher}, quantifying its impact on students and discussing various notions of (group) fairness. In particular, we present two theorems that quantify the optimal deterministic debiasing sets under different measures of fairness. In Section~\ref{sec:incentive compatible}, we show that such deterministic policies fail to be \emph{incentive compatible} and \emph{individually fair}, introduce the randomized assignment of vouchers to satisfy these fairness conditions, and show that such randomized policies achieve a lower maximum mistreatment than their deterministic counterparts. In Section~\ref{sec:experiments} we apply our policies to the real-world dataset of SHSs admissions for the 2016--17 cohort. A discussion of our results appears in Section~\ref{sec:discussion}. Several auxiliary and additional results are presented in the Electronic Companion. In particular, in~\ecref{sec:additive-bias}  we discuss variations of models of bias, quantifying the impact of applying our stylized model to contexts with more intricate forms of bias, and extending our results on deterministic debiasing to an additive model. 

\subsection{Related Work}\label{sec:literature-review}

The most common way to model admissions to schools is with a two-sided market, consisting of schools and students respectively, where each agent has an ordered preference over agents acceptable to them on the other side of the market. This model has been used to match doctors to hospitals by the National Residency Matching Program since the 1960s, and it has since gained widespread notoriety when~\citet{abdulkadirouglu2005new} used it to reform the admissions process for New York City public high schools in 2003. Since then, admission decisions in NYC have been centralized and are (essentially) governed by the classical Gale-Shapley Deferred Acceptance algorithm~\citep{gale1962college}. The simplicity of the algorithm, as well as the drastic improvement in the quality of the matching it provides when compared to the pre-2003 method have led to academic and public acclaim, and spurred applications in many other systems. However, this mechanism does not naturally address problems like school segregation and class diversity, which have worsened and become more and more of a concern in recent years~\citep{kamada2024fair,kucsera2014new, Eliza}. The scientific community and policy-makers have reacted in various ways such as by modifying the mathematical model to incorporate group-specific quotas or proportionality constraints~\citep{biro2010college, nguyen2019stable, tomoeda2018finding}. However, there is evidence that adding such constraints may even hurt the very students they were meant to help~\citep{backes2012affirmative,fershtman2021soft,hafalir2013effective}, and the question of legal challenges abounds.

There is a long line of work on affirmative action policies in theory and in practice~\citep{abdulkadirouglu2005college,chade2014student,chan2003does,hafalir2013effective,quinn2017implicit}; and alternatives such as the ``top 10\%'' admissions criteria implemented in Texas \citep{TexasTop10_2024}. Substitute mechanisms such as the top 10\% criteria deviate significantly from current practice bar esoteric implementations, and it is unclear whether such criteria improve the status quo or worsen it. For instance, ~\cite{long2004race} found a significant negative impact on the admissions rates of disadvantaged groups if affirmative action policies for college admissions were replaced by top $x$\% rules. We take a completely different approach to improving the outcomes for disadvantaged students by voucher distribution, which will naturally help with the downstream impacts in the education pipeline towards economic opportunities~\citep{kannan2019downstream, coate1993will}.

Our work complements an increasing body of work on test-taking itself, in particular policies that make tests optional or allow applicants to take tests multiple times. Test-optional policies are typically studied in the context of college admissions (see e.g.~\cite{liu2021test,dessein2023test}). It is unclear if these can be adapted to the public high school admissions due to significant differences in admissions dynamics and resistance in past cases due to public scrutiny\footnote{Unlike college admissions, almost all public school admissions processes on the other hand are deeply scrutinized by the public, and choose simple mechanisms with high explainability, up to the extreme case of SHS admissions being dictated by law.}. 
A recent study in~\cite{niu2022best} shows the impact of being able to retake SAT exams and that reporting all the scores leads to more equitable outcomes as well as a more accurate signal for colleges\footnote{The DOE allows students to take the SHSAT both in their 8th grade and in their 9th grade with slightly different tests, but only about 6\% of test takers are in 9th grade~\citep{shsatpipeline}.}. 
Further, \cite{garg2021standardized} focus on the design of a fair admissions process by identifying conditions where standardized tests should be dropped, while our paper mostly focuses on pre-admissions policies. We finally remark that our proposed interventions are politically and practically palpable as they do not require changing the admissions criteria (changes to the admissions criteria can pose a significant hurdle to implementation~\citep{CityJournal2022}). In New York for instance, this would itself require changing state law (Hecht-Calandra Act).

Any admissions policy is susceptible to manipulation by applicants. Recent work by \cite{hu2019disparate} has considered strategic behavior of students in a classification setting, where each student can expend some bounded amount of resources to improve their test-score performance and convert a ``reject" decision to an ``accept" decision. The school can provide subsidies to students to reveal their true potential. They further identify cases where providing a subsidy can make the group receiving the subsidy worse-off. Though our work considers a completely different model, we also find theoretical conditions under which voucher distribution can in fact worsen some fairness metrics over the disadvantaged groups, and investigate strategic behavior of students, see Section~\ref{sec:voucher} and Section~\ref{sec:randomized-vouchers}.

Various selection problems have been investigated in models with a multiplicative bias introduced by \cite{kleinberg2018selection} (including~\cite{celis2020interventions, celis2021effect, salem2023secretary, emelianov2020fair}), but to the best of our knowledge, this paper is the first to investigate it in the role of school admissions. There exists some literature such as~\cite{hastings2009heterogeneous,laverde2020unequal} on understanding the impact of family backgrounds on student preferences, but this is orthogonal to the questions we study here. Our work complements existing work in the education and policy literature which shows additional resources can positively impact 
low achieving student groups~\citep{dee2011impact,greenwald1996effect}, in particular,~\cite{yue2018supplemental} studies the impact of supplemental instruction on disadvantaged students.

Lastly, our work is related to the modeling of bias and discrimination itself, which is an active area of research in economics~\citep{boschma2016concentration,lang2012racial,fryer2011importance}. The DOE is required by law to use only the SHSAT score to decide admissions to the SHSs, and so does not directly discriminate against any student based on group membership, since it does not take this information into account in the decision. Instead, the discrepancies in performance that arise in this setting may be caused by unequal opportunities available to students at earlier stages in the educational pipeline. However, our analysis only relies on modeling the distributional shift in observed student performance, and not the reasons due to which these differences exist. Our models bear resemblance to what is known as \emph{taste-based discrimination} in the literature: models where agents hold uniform distaste towards members of some group. In the modern world, many guardrails exist to disbar and discourage direct taste-based discrimination in consequential decisions such as in school admissions (e.g.~based on ``protected attributes''). Therefore, while taste-based discrimination models make some attempt to describe why bias or prejudice come to be, more modern work commonly takes bias to be primarily of the \emph{statistical discrimination} kind~\citep{phelps1972statistical}. While these models do not explain why initial bias has arisen, they posit that the processes that beget discrimination can be explained via economically rational behavior of individual agents~\citep{arrow1973theory}. In particular, such phenomena may arise even when agents do not directly imbue distaste towards some group, but due to distributional differences between groups. A common feature in such models is that the decision maker relies on some signal that is noisier and hence a less accurate predictor of the true characteristic of interest (such as innate ability) for discriminated populations~\citep{aigner1977statistical}.  We refer the reader to the survey by \cite{fang2011theories} for more details on this literature. As part of this work, we fit various models to the distributional differences in performance drawing from both taste-based discrimination and statistical discrimination models (with evaluations of the disadvantaged group considered to be the noisier signal), and gauge the impact of these differences computationally. Assuming that the innate ability of the two groups of students should be similar, but the perceived potential measured via the standardized tests is biased due to some bias process, we show that additional resources should be targeted towards average performers. The qualitative takeaways from our work are consistent under both models -- taste-based and statistical discrimination, up to a reasonable amount of noise (see~\ecref{sec:additive-bias} and~\ecref{appx:misspecified}).

\section{A continuous matching market}\label{sec:market}

We now introduce our stylized matching model for school choice. For tractability, we follow a recent trend in the literature and assume both schools and students to be continuous sets (see Appendix~\ref*{appx-sec:discrete-vs-cont} for a discussion on this choice). We denote the set of students by $\Theta$ and for each student $\theta\in \Theta$, endow them with a \emph{true potential} $Z(\theta)$ sampled from some probability distribution. We interpret this true potential as the innate ability of student $\theta$. For the rest of this section, we assume $Z(\theta)\distributed\Pareto(1,\alpha)$ (note that all students then have true potentials exceeding $1$). This assumption is relaxed in Section~\ref{sec:randomized-vouchers} where we consider randomized voucher programs under different assumptions on the potential distribution. We occasionally identify a student's true potential $Z(\theta)$ with $\theta$, the student itself.

We denote the mass of schools by the unit interval, $[0,1]$, and assume that all students rank schools in the same order with school $0$ the best, and school $1$ the worst. Schools on the other hand rank students uniquely based on their potential. This conveniently lets us identify the matching with the complementary cumulative distribution function (ccdf) of the students' potentials. That is, let $\mu$ denote the matching when schools rank students according to their true potentials. Then, student $\theta$ gets matched to school $\mu(\theta)=1-F(Z(\theta))$ where $F(t)$ is the cumulative distribution function (cdf) of the distribution of student potentials. For convenience, we denote the ccdf by $\F=1-F$, so $\mu(\theta)=\F(Z(\theta))$. In the Pareto case, we get $\mu(\theta)=Z(\theta)^{-\alpha}$.

We consider the student body to be composed of two groups: a proportion $1-p$ of \emph{non-disadvantaged} students $G_1$, and a proportion $p$ of \emph{disadvantaged} students $G_2$. We let $\hat Z(\theta)$ denote the \emph{perceived potential} of student $\theta\in\Theta$. While students in $G_1$ are perceived at their true potential, we assume that the perceived potential of disadvantaged students are biased by a constant multiplicative factor $\beta\in(0,1]$. That is, if $\theta\in G_2$ then $\hat Z(\theta)=\beta Z(\theta)$; otherwise, if $\theta\in G_1$ then $\hat Z(\theta)=Z(\theta)$. We explore alternate models of bias, including additive models where $\hat{Z}(\theta) =  Z + \beta + \epsilon$ (with $\epsilon \sim N(0, \alpha)$ or $\epsilon \sim U[-\alpha, \alpha]$) in~\ecref{sec:additive-bias} and  \ecref{appx:misspecified}.

Let $F_1$ and $F_2$ be the cdfs for the perceived potentials of students in $G_1$ and $G_2$, respectively. Then,
\begin{align*}
F_1(t) = 1- t^{-\alpha}; \qquad F_2(t) = 1- \ba t^{-\alpha}.
\end{align*}
Note that the domain of $F_1$ is $[1,\infty)$, whereas the domain of $F_2$ is $[\beta, \infty)$. We now consider the matching of students to schools under this biased regime (when schools use perceived potentials to rank students), which we denote by $\hat\mu$. For a student $\theta\in\Theta$, $\hat\mu(\theta)$ is equal to the mass of students whose perceived potentials exceed $\hat Z(\theta)$, and one can compute
\begin{equation} \label{eq:match-act}
    \hat\mu(\theta) = \begin{cases}
        (1-p)\F_1(\hat Z(\theta)) + p\F_2(\hat Z(\theta)) & \text{ if } \theta\in G_1, \\[1mm]
        (1-p)\F_1(\hat Z(\theta) \vee 1) + p\F_2(\hat Z(\theta)) & \text{ if } \theta\in { G_2},
    \end{cases}
\end{equation}
where $\vee$ is the maximum operator. 
When $\beta=1$ (no bias),~\eqref{eq:match-act} equals $\mu(\theta)$ for all $\theta\in\Theta$.

Formally, we define a \emph{matching} in this market to be a surjective measurable function $\gamma:\Theta\to[0,1]$, such that the mass of students mapped to a measurable set of schools $S\subseteq [0,1]$ coincides with the standard Lebesgue measure $\nu$ of $S$. That is, any surjective function $\gamma$ from $\Theta$ to $[0,1]$ is a matching if
\begin{align*}
\nu(\gamma^{-1}(S)):= (1-p) \int_{\theta \in \gamma^{-1}(S) \cap G_1} dF_1(\hat Z(\theta)) + p\int_{\theta \in \gamma^{-1}(S)\cap G_2} dF_2(\hat Z(\theta))
\end{align*}
is equal to the standard Lebesgue measure of $S$ for all measurable $S \subseteq [0,1]$. One can easily check that $\mu$ and $\hat \mu$ defined above are matchings.

\begin{example}\label{ex:model-example-maya-lisa}
Suppose $\alpha=3$ so student scores are sampled from $\text{Pareto}(1,3)$. Suppose Maya $\in G_2$ scores $Z($``\text{Maya}''$) = 1.4$, and Lisa $\in G_1$ scores $Z($``\text{Lisa}''$) = 1.3$. In the unbiased setting, Maya gets matched to schoool $\F(Z($``\text{Maya}''$)) = 1-(1-1/(1.4^3)) \approx 0.3644$ while Lisa gets matched to $\F(Z($``\text{Lisa}''$)) = 1-(1-1/(1.3^3)) \approx 0.4552$. On the other hand, in the biased case with $\beta=0.9$, we have $\hat Z($``\text{Maya}''$) = 1.26$, while $\hat Z($``\text{Lisa}''$) = 1.3$.  If $p=0.2$, we have that Maya and Lisa are matched to schools $\hat \mu (\text{``Maya''})= 0.4729$ and $\hat \mu (\text{``Lisa''})= 0.4305$, respectively to a significantly worse (slightly better) school than they used to in the setting without bias. Note that Lisa has a smaller true potential than Maya but is assigned to a better school in the biased setting.
\end{example}

\section{Impact of Bias on Students} \label{sec:analysis}

Our first goal is to understand how much bias affects agents in the market. In particular, we would like to quantify the loss of efficiency for students\footnote{In \ecref{ec-sec:schools}, we take the schools' perspective and show that there is effectively no loss of efficiency for schools under this model, creating little incentive for them to intervene at the individual school level.} when students are matched to schools under $\hat \mu$ instead of $\mu$. Formally, we define the displacement and mistreatment of a student as follows.

\begin{definition}[Displacement and Mistreatment]\label{def:dismis}
Let $\theta\in\Theta$, let $\mu$ be the matching in the absence of bias (using $Z(\theta)$), and let $\gamma$ be some other matching. We define
\begin{enumerate}
\item the \emph{displacement} of $\theta$ under $\gamma$ as $\disp_\gamma(\theta)=\gamma(\theta)-\mu(\theta)$; and
\item the \emph{mistreatment} of $\theta$ under $\gamma$ as $m_\gamma(\theta)=\max(0,\gamma(\theta)-\mu(\theta))$.
\end{enumerate}
\end{definition}
We often drop the subscript when the matching at hand is clear from context.

Note that if $\theta\in G_1$, the displacement under $\hat\mu$ is non-positive, and if $\theta\in G_2$, it is non-negative. The displacement for $\hat\mu$ can be calculated using the formulae for $\mu$ and $\hat\mu$ in~\eqref{eq:match-act}.
\begin{proposition} \label{prop:mistreat-students}
    For any student $\theta \in G_2$,
    the displacement under $\hat\mu$ is given by:
    $$\disp_{\hat\mu}(\theta)=
    \begin{cases}
        \displaystyle (1-p) \left({Z(\theta)} \right)^{-\alpha} \left( \beta^{-\alpha} -1 \right) & \text{ if } Z(\theta)\ge \invb, \\[1mm]
        \displaystyle (1-p) \left(1- \left({Z(\theta)} \right)^{-\alpha} \right) & \text{ if } Z(\theta) \le \invb.
    \end{cases}$$
    For any student $\theta\in G_1$, we have $\disp_{\hat\mu}(\theta)=-p\left(1-\ba\right)\left({Z(\theta)} \right)^{-\alpha}.$ The maximum displacement of $(1-p)(1-\ba)$ is experienced by a $G_2$ student with true potential $1/\beta$. 
\end{proposition}

One can interpret this result as follows. Starting from the top school, $G_1$ students gradually take up more seats than they would without bias, thus gradually pushing $G_2$ students to worse schools. This process stops once all $G_1$ students are matched, and the only students that remain to be matched are $G_2$ students. As a result, in the lowest ranked schools, all students are $G_2$ students. Hence, the displacement of $G_2$ students decreases towards the end. Figure~\ref{fig:combined-figures}(b) gives a pictorial illustration of Proposition~\ref{prop:mistreat-students}. From there, one can clearly see how the most mistreated students are average performers. This intuition will be fundamental in devising policies to counter the effect of bias.

\section{Deterministic Centralized Interventions} \label{sec:voucher}

In this section, we discuss deterministic interventions, where the central planner is able to assign vouchers to some targeted subset of disadvantaged students to \emph{debias} them in order to reduce mistreatment. When a disadvantaged student is chosen for such a voucher (in the form of supplemental instruction, test prep, a scholarship, or similar), we assume that they receive the support they need to realize their innate ability and perform at their true potential by the time their performance is measured (e.g.~when they take the SHSAT). Such interventions are costly, so deciding which set of students to assign these vouchers to becomes a key question. In particular, vouchers ought to be assigned in order to mitigate bias as much as possible subject to some budget constraint. These policies are deterministic, and the decision of whether to assign a voucher to a student (hence, debias them) depend only on the potential of the students. In the next section, we will discuss randomized policies.

\paragraph{Deterministic Debiasing Sets:} Formally, the planner chooses some measurable subset $T\subseteq[1,\infty)$ of disadvantaged students to debias, which we call a \emph{deterministic debiasing set} (DDS). We let $\hat c\in[0,1]$ be the budget of the central planner, then define $\mcT(\hat c)$ to be the set of all DDSs respecting this budget, that is the set of all measurable $T\subseteq[1,\infty)$ with $\int_T\,dF_1\leq\hat c$. Additionally, let $\mcT^c(\hat c)\subseteq\mcT(\hat c)$ be the set of such sets that are also closed and connected (i.e.~they are intervals $[a,b]\subseteq[1,\infty)$ that satisfy $F_1(b)-F_1(a)\leq \hat c$). If $\theta\in T$, then we set $\hat Z(\theta)=Z(\theta)$, so that student performs at their true potential. Let $\mu_T:\Theta\to[0,1]$ be the matching after $G_2$ students whose true potentials lie in $T$ have been debiased\footnote{In this section we assume the debiasing decision is based on \emph{true potentials}. In the case of deterministic one-to-one bias, the distinction is not important, but we later discuss debiasing on perceived potentials when it becomes relevant.}. Write $m_T$ for the mistreatment under $\mu_T$ as in Definition~\ref{def:dismis}. Recall that the mistreatment is the drop in the rank of the school the student is matched to (if this drop is positive): a student $\theta$ has mistreatment equal to $0$ if they are assigned to a school at least as good as $\mu(\theta)$. In the following, we evaluate a voucher distribution by its effect on the mistreatment of $G_2$ students, since only $G_2$ students may experience strictly positive mistreatment. It is easy to see that after the interventions, no student $\theta \in G_1$ will be matched to a school worse than $\mu(\theta)$. This is because our interventions focus on helping (certain) $G_2$ students reveal their true potentials, hence for any $G_1$ student $\theta$, no student with potential lower than $Z(\theta)$ can have a perceived potential higher than $\hat Z(\theta)=Z(\theta)$.

\paragraph{Fairness Considerations:}  Finding a set of students to allocate vouchers to is a resource allocation problem with natural fairness considerations that guide the choice of the measures to be optimized\footnote{To read a more detailed discussion on relevant philosophies of equality and decision-making, we refer the interested reader to the 1979 Tanner Lectures on Human Values (\cite{Sen1979}).}. For the cohort of disadvantaged students, we take the view of finding a distribution of vouchers so that the mistreatment across $G_2$ students is as balanced or equitable as possible. We analyze two representative fairness measures in this regard: (1) the \emph{positive area under the mistreatment curve} over all disadvantaged students, and (2) the \emph{maximum mistreatment} experienced in this cohort. The former is the continuous $L_1$ norm (under the $F_1$ measure) of the mistreatment after voucher allocation, or the \emph{positive area under the curve} (PAUC), which we denote by $\sigma$. The latter is the continuous $L_{\infty}$ norm of mistreatment and we denote it by $mm$. Formally, for a matching $\gamma$, we define
\begin{align}
 \sigma(\gamma)&:=\int_{\Theta} m_\gamma\, dF_1 = \|m_\gamma(\theta)\|_1,\\
 mm(\gamma)&:=\sup_{\theta\in\Theta}m_\gamma(\theta)=\lim_{p \rightarrow \infty}\left(\int_{\Theta} \abs{m_\gamma}^p\, dF_1\right)^{1/p}=\|m_\gamma(\theta)\|_{\infty}.
\end{align}
These notions of fairness have been axiomatically established and are well-studied in the literature. For example, the {\it min-max} notion of fairness has been considered in~\cite{Kumar2000}, whereas the positive area under the curve corresponds to average mistreatment of group $G_2$: it is a group notion of fairness consider in many fairness related studies~\citep{Conitzer2019,Dwork2018}. Since we will show that the optimal interventions at the two extremes $L_1$ and $L_\infty$ target qualitatively similar sets of students, we expect the solution for any other $L_p$ norm to also behave similarly\footnote{An $L_p$ norm on a probability space with $p$ small generally measures the average of a function, whereas a large $p$ measures its ``peakiness'', with $p=\infty$ equaling the essential supremum, and values in between trading off between these properties (for a further discussion on the relationship between $L_p$ spaces, see~\cite{folland1999real}).}, and so restrict our analysis to $L_1$ and $L_{\infty}$.

\paragraph{Optimal Deterministic Strategies:} We now proceed to proving our main results in the deterministic setting, which fully describe the optimal debias intervals in our model. 
We define $\mcT_{mm}(\hat c)=\argmin_{T\in\mcT(\hat c)}mm(\mu_T)$, in other words, $\mcT_{mm}(\hat c)$ is the collection of sets $T$ that minimize $mm(\mu_T)$ among sets with $\int_T\,dF_1\leq\hat c$. The next result gives an exact characterization of $\mcT_{mm}(\hat c)$, assuming\footnote{We refer to Section~\ref{sec:randomized-vouchers} for a discussion on the various technical assumptions on data from Section~\ref{sec:voucher} and Section~\ref{sec:randomized-vouchers}.} $p<1-\ba$.

\begin{figure}[h]
    \centering
    \begin{subfigure}{.48\textwidth}
        \centering
        \includegraphics[width=\textwidth]{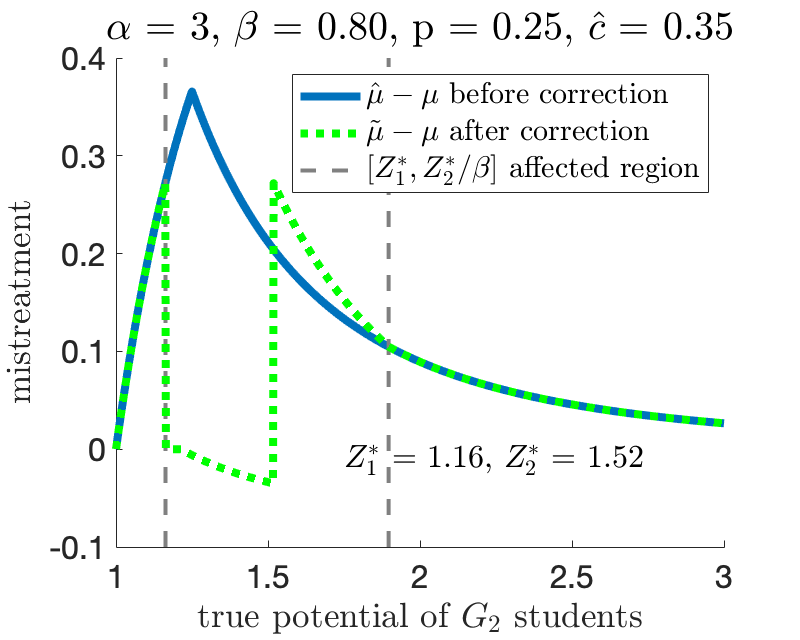}
        \caption{optimal debiasing when $\hat c$ is large}
    \end{subfigure}
    \begin{subfigure}{.48\textwidth}
        \centering
        \includegraphics[width=\textwidth]{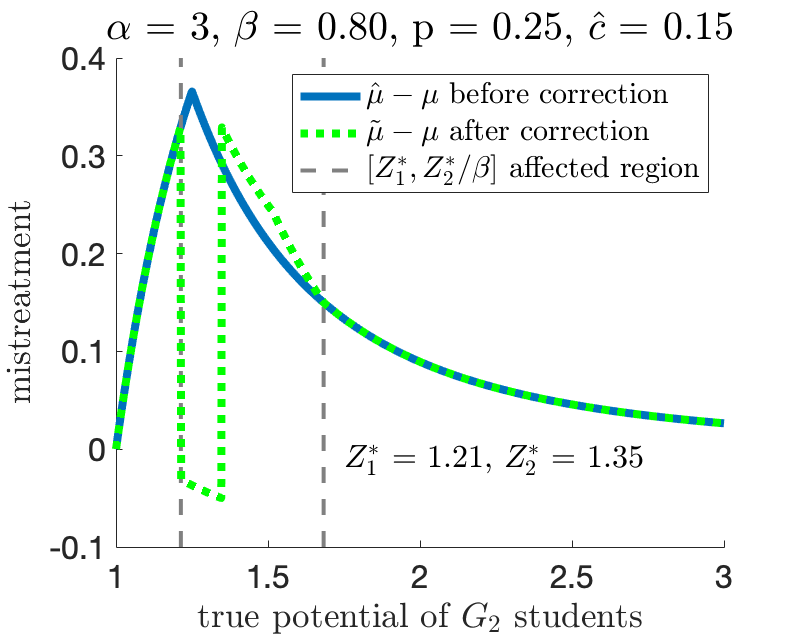}
        \caption{optimal debiasing when $\hat c$ is small}
    \end{subfigure}
    \vspace{.8em}
    \caption{Maximum mistreatment before and after optimal voucher allocation.} \label{fig:best-range}
\end{figure}

\begin{theorem} \label{thm:optimal-range-mm}
    Assume $p<1-\ba$. Then  there exists a set $T=[Z_1^*, Z_2^*] \in \mathcal{T}_{mm}(\hat c)$ such that all other sets in $\mathcal{T}_{mm}(\hat c)$ differ from $T$ on a set of measure zero. If $\hat c\ge \frac{(1-p) (1 -\ba)}{1-p+1- \ba}$, then
    $$\displaystyle Z_1^* = \bigg( \frac{(1-p)+(\invba -1)\hat c}{\invba-p} \bigg)^{ -\frac{1}{\alpha}} \quad \text{and} \quad Z_2^* = \bigg( \frac{(1-p)(1-\hat c)}{\invba-p} \bigg)^{ -\frac{1}{\alpha}},$$ and $mm(\mu_{[Z_1^*, Z_2^*]}) = (1-p)(1-\ba) \frac{1-\hat c}{1-p\ba}$, reduced from  $mm(\hat\mu) = (1-p)(1-\ba)$. Conversely, if $\hat c\le\frac{(1-p) (1-\ba)}{1-p+1- \ba}$, then  $mm(\mu_{[Z_1^*, Z_2^*]}) = (1-p-\hat c)(1-\ba)+p\hat c$, where
    $$\displaystyle Z_1^* = \bigg( \frac{(1-p-\hat c) \ba}{1-p} + \hat c \bigg)^{ -\frac{1}{\alpha}}\quad \text{ and } \quad Z_2^* = \bigg( \frac{(1-p-\hat c) \ba}{1-p} \bigg)^{ -\frac{1}{\alpha}}.$$
\end{theorem}

We include the proof of Theorem~\ref{thm:optimal-range-mm} in \ecref{ec-sec:proofs:or-mm}. Interestingly, our proof also shows that if vouchers are not distributed carefully, one may actually \emph{increase} the maximum mistreatment and, more generally, shows which debiasing sets lead to an improvement over the status quo. A pictorial representation of Theorem~\ref{thm:optimal-range-mm} is given in Figure~\ref{fig:best-range}. The two sub-figures correspond to two choices of $\hat c$.

We next consider the positive area under the curve (PAUC): this is the aggregate amount of mistreatment experienced by all $G_2$ students. We restrict our attention to debiasing $G_2$ students whose potentials are in a connected set---this is a justifiable implementation in practice. This assumption also makes our analysis tractable. In particular, let $\mcT_{auc}^c(\hat c)=\argmin_{T\in\mcT^c(\hat c)}\sigma(\mu_T)$ be the set of $T\in\mcT^c(\hat c)$ that minimize $\sigma(\mu_T)$. The next result gives an explicit description of these sets when assuming, again that $p<1-\ba$ and additionally that $p<0.5$.
\begin{theorem} \label{thm:optimal-range-pauc}
    Assume $p<1-\ba$ and $p<0.5$. Then $\mathcal{T}^c_{auc}(\hat c)$ is made up of a unique set $T=[Z_1^*, Z_2^*]$. If $\hat c\ge\frac{(1-p) (1-\ba)}{2-p-\ba-p\ba +p\beta^{2\alpha}}$, then:
    $$\displaystyle Z_2^* =\bigg(\frac{(1-p)(1-\hat c)}{p\ba+\invba-2p} \bigg)^{-\frac{1}{\alpha}} \quad\text{ and }\quad Z_1^*= \bigg( \frac{(1-p)(1-\hat c)}{p\ba+\invba-2p} + \hat c\bigg)^{-\frac{1}{\alpha}},$$
    and $\sigma(\mu_{[Z_1^*, Z_2^*]}) = \frac{1}{2}(1-p)(1-\ba) \left( \frac{(\invba-p) (1-\hat c)^2}{p\ba+\invba-2p} \right)$, down from $\sigma(\hat\mu) =\frac{1}{2} (1-p)(1-\ba)$. Else $\hat c\le\frac{(1-p) (1-\ba)}{2-p-\ba-p\ba +p\beta^{2\alpha}}$ and $\sigma(\mu_{[Z_1^*, Z_2^*]}) = \frac{1}{2}(1-p)(1-\hat c)^2 - \frac{1}{2} \ba \left( \frac{[(p\ba-1)\hat c+(1-p)]^2}{1-p} + p\hat c^2 \right)$, where
    $$\displaystyle Z_2^* =\bigg(\frac{(p\ba-1)\hat c+(1-p)}{(1-p) \invba} \bigg)^{-\frac{1}{\alpha}} \quad\text{ and }\quad Z_1^*= \bigg( \frac{(p\ba-1)\hat c+(1-p)}{(1-p) \invba}+ \hat c\bigg)^{-\frac{1}{\alpha}}.$$
\end{theorem}
The proof of Theorem~\ref{thm:optimal-range-pauc} is given in~\ecref{ec-sec:proofs:or-pauc}. 

In Table~\ref*{ec-tab:compare-two-measures} of the e-Companion, we compare the optimal ranges of students to debias under the two measures of fairness, with parameters $\alpha=3$, $\beta=0.8$, and $p=1/4$. There is on average a 95\% overlap of the optimal intervals under the two measures. In particular, both measures suggest that vouchers should be given to the average (middle performing) students.

Although these results highlight an important deviation from the current practice of prioritizing top-performing students for scholarships, we highlight in the next section two fundamental problems with such policies.

\section{Incentive Compatible and Individually Fair Voucher Distribution}\label{sec:incentive compatible}

In the previous section, we characterized the \emph{deterministic} intervals to distribute vouchers to in order to minimize either the maximum mistreatment or the positive area under the mistreatment curve for the disadvantaged student group. In this section, we introduce two natural and desirable properties,  individual fairness and incentive compatibility, that the policies developed in Section~\ref{sec:voucher} fail to have. We then show in Section~\ref{sec:randomized-vouchers} how we can shift from a deterministic voucher distribution policy to a randomized one in order to satisfy them. Our results rely on the original assumptions that the potential $F$ of students follows the Pareto distribution, and  the true potential of a student is always revealed when a voucher is assigned to them. In Section~\ref{sec:stoch-effect-vouchers}, we discuss a model where the latter assumption is  relaxed. In Section~\ref{sec:IwP}, we instead relax the former assumption and show that if no hypothesis can be made on $F$, incentive compatibility can only be achieved by distributing vouchers with a probability increasing with the (perceived)  potential.

Let us now discuss \emph{individual fairness}, which requires that similar individuals be treated similarly~\citep{Dwork2012}. While a formal definition of this concept is postponed to Section~\ref{sec:randomized-vouchers}, we observe here that the policies developed in Section~\ref{sec:voucher} fail to be individually fair as individuals close to the boundary of the debiasing interval are treated very differently depending on whether they are inside or outside of it.

Our second property is \emph{incentive compatibility} (see, e.g.,~\cite{roughgarden2010algorithmic}). In general, it requires that no individual can benefit from misrepresenting their features. In our setting, we assume that a student is able to misrepresent themselves as appearing to have lower perceived potential (e.g., intentionally achieving a lower score on a test) in order to be part of the set of students that get allocated vouchers. Recall that a DDS is a measurable set $T\subseteq [1,\infty)$. A DDS  is \emph{incentive compatible} if no student is assigned to a better school if they misreport their performance. Formally, assume that a voucher given to a disadvantaged student with reported perceived potential  $\beta Z(\theta)$, will improve their performance up to $Z(\theta)$\footnote{This assumption is justified by the fact that additional training is usually commensurate with the (perceived) level of a student.}; then, a DDS $T$ is incentive compatible if for each $x \in [1,\infty) \setminus T$ and $x' \in T$ with $x>x'$, we have $\beta x \geq x'$.

\begin{lemma}\label{lem:deterministic-incentive-compatible}
Assume $\beta\in [0,1)$ and let $T\neq \emptyset$ be an incentive compatible DDS. Then $T$ is of the form $\{\theta \in \Theta : Z(\theta)\geq \delta\}$ or $\{\theta \in \Theta : Z(\theta)> \delta\}$ for some $\delta \in [1,\infty)$.
\end{lemma}

We defer the proof to Appendix~\ref*{appx-sec:proof-deterministic-incentive-compatible}. This lemma shows that if we care about incentive compatibility and require that vouchers are distributed deterministically, then the only feasible mechanism is to debias all students that have potential above some cutoff $\delta$ (i.e.~the top students). However, we showed in the last section that such policies are not optimal. To overcome these flaws in deterministic policies, we next turn to randomization.

\subsection{Randomized assignment of vouchers}\label{sec:randomized-vouchers}

In this section (and in its proofs in \ecref{ec-sec:rvp}) we abuse notation for simplicity and identify a student $\theta$ with their true potential $Z(\theta)$.

A \emph{Randomized Voucher Program} (RVP) is a measurable function $\rho:\Theta\to[0,1]$ that gives, for each $\theta \in \Theta$, the probability that a $G_2$ student with true potential $\theta$ is assigned a voucher. Observe that if $\rho(\theta)\in\set{0,1}$ for all $\theta\in\Theta$, then $\rho^{-1}(1)$ is a measurable set and therefore also a deterministic debiasing set (DDS) as in the definition in Section~\ref{sec:voucher}; likewise, given a DDS ${T}$ we can construct the RVP $\rho_{{T}}(\theta)=\indc_{\{\theta\in{T}\}}$ that coincides with a given DDS.

The main class of RVPs investigated in this section are those we call \emph{Proportional-to-Mistreatment} (PropM), denoted by $\rho_m$ and defined as
\begin{equation}\label{eq:PTM}
\rho_m(\theta):=\frac{2\hat c}{(1-\beta^{\alpha})(1-p)}m_{\hat\mu}(\theta),
\end{equation}
for some $\hat c\in (0, 1/2]$ (recall that $m_{\hat\mu}(\theta)$ is the mistreatment of a student with real potential $\theta$ when no vouchers are distributed). It is easy to see that $\hat c$ is the expected proportion of disadvantaged students that will get a voucher, that is $\hat c=\int_\Theta\rho_m\,dF$.  Intuitively, $\rho_m$ assigns a larger probability of being debiased to students with a higher mistreatment.

As we show next, under broadly applicable technical hypotheses on the parameters, PropMs satisfy many of the desirable properties that deterministic voucher allocations fail to have. Moreover, we will show that they can lower the maximum expected mistreatment. To state these results formally, we first extend concepts from deterministic DDSs to RVPs. We let $\mu_\rho(\theta)$ be the expected school that a student with true potential $\theta\in\Theta$ is assigned to under $\rho$. An explicit expression for $\mu_\rho$ for arbitrary $\rho$ can be found in Lemma~\ref*{claim:mu_rho} of~\ecref{ec-sec:rvp}. Now all prior definitions and notation carries over, including the expected mistreatment $m_\rho$, and maximum expected mistreatment $mm_\rho$.

An RVP $\rho$ is \emph{incentive compatible} if $\mu_\rho(\theta')\geq\mu_\rho(\theta)$ for all $\theta'<\theta$. That is, an RVP is incentive compatible if a student with true potential $\theta$ is not better off by manipulating themselves to appear as having a true potential $\theta'<\theta$.

We define individual fairness as a Lipschitz condition on $\rho$. We say an RVP $\rho$ is \emph{$k$-individually fair} if, for each $\theta, \theta' \in [1,\infty)$, $|\rho(\theta)-\rho(\theta')|\leq k |\theta-\theta'|$ (note that under this definition, no non-trivial DDS is $k$-individually fair for any $k$). We can now state the main result from this section, whose proof is deferred to Appendix~\ref{ec-subsec:rvp-ptm-proof}. Observe that $mm^*(\hat c)$ is the maximum mistreatment achieved by the best deterministic policy minimizing this metric, as computed in Theorem~\ref{thm:optimal-range-mm}.

\begin{theorem}\label{thm:properties-of-ptm-rvps}
Let $\rho_m$ be a PropM defined as in~\eqref{eq:PTM} for some $\hat c\in(0,1/2]$ and assume $p\leq 0.5$. Let $mm^*(\hat c)=\min_{T\in\mcT(\hat c)}mm(\mu_T)$. Then:
\begin{enumerate}
\item $\rho_m$ is $\frac{2\hat c\alpha}{1-\beta^\alpha}$-individually fair.
\item $\rho_m$ is incentive compatible for
\begin{equation}\label{eq:condition-on-hat-c-for-incentive compatibility}\hat c \leq\frac{1-p}{2\sbr{p(1-\beta^\alpha)+(1-p)(\beta^{-\alpha}-1)}}.
\end{equation}
\item Suppose $p<1-\beta^\alpha$ and $\hat c\leq\frac{(1-p)(1-\beta^\alpha)}{1-p+1-\beta^\alpha}$. Then $mm_{\rho_m}\leq mm^*(\hat c)$ if
\begin{equation}\label{eq:lb-hatc}\hat c \geq 1-\frac{p+1-\beta^\alpha}{4p(1-\beta^\alpha)}.\end{equation}
\end{enumerate}
\end{theorem}

\eqref{eq:condition-on-hat-c-for-incentive compatibility} and~\eqref{eq:lb-hatc} give complementary conditions on the amount of vouchers that can be given out. On one hand,~\eqref{eq:condition-on-hat-c-for-incentive compatibility} suggests that distributing too many vouchers prevents incentive compatibility of the PropM. In fact, a $\hat c$ too large causes students performing just above the most mistreated student to be incentivized to artificially lower their score, as the absolute value of the derivative of the PropM becomes large around its maximum. 
On the other hand,~\eqref{eq:lb-hatc} suggests that we need to distribute enough vouchers to see the maximum expected mistreatment $mm_{\rho_m}$ drop below the optimal deterministic one $mm^*(\hat{c})$. This is because the optimal deterministic policy debiases the most mistreated student straight away whereas the PropM distributes vouchers more widely, and so the maximum expected mistreatment does not immediately drop as significantly. As we discuss at the end of the section, both conditions are satisfied for a large range of parameters.

PropMs represent therefore a more robust and theoretically satisfying alternative to reducing the maximum mistreatment that is at least as effective as the deterministic voucher assignments developed in Section~\ref{sec:voucher}.

\subsection{Stochastic-effect vouchers}\label{sec:stoch-effect-vouchers}

Our treatment so far assumes that a student who receives a voucher is fully debiased, i.e., their full potential is unambiguously revealed. We now discuss a model that relaxes this assumption. Proofs from this section are deferred to~\ecref{ec-subsec:proof-new}. In particular, we assume that a voucher allocated to a student may successfully debias that student with a constant (success) probability $\kappa \in [0,1]$, while with probability $(1-\kappa)$ it fails to debias the student and leaves their potential unaffected at the perceived potential. We call this model \emph{All-or-Nothing stochastic-effect Vouchers} (ANV). 
For an RVP $\rho$ and success probability $\kappa$, the concept of maximum expected mistreatment $mm_{\rho,\kappa}$ is well-defined. Observe that $mm_{\rho,1}$ coincides with $mm_{\rho}$ in the classical model. We next bound $mm_{\rho,\kappa}$ as a function of $mm_{\rho}$ and other features of the model.

\begin{lemma}\label{lem:all-or-nothing-vouchers} Consider an instance of the ANV model with success probability $\kappa$. Let $\rho$ be an RVP. Then $mm_{\rho,\kappa}\leq mm_{\rho} + (1 - \kappa)(1 - \beta^\alpha)\|\rho\|_\infty( 1 + p(1 + \kappa)\|\rho\|_\infty)$. 
\end{lemma}

Lemma~\ref{lem:all-or-nothing-vouchers} describes analytical conditions under which we can safely rely on $mm_\rho$ to upper bound $mm_{\rho,\kappa}$ for a given RVP $\rho$. Note that for $\kappa=1$ we obtain the bound $mm_\rho$, as expected. At the end of Section~\ref{sec:incentive compatible} we give explicit realizations of the bound for various choices of the parameters. As we discuss next, when $\rho$ is constant across $\Theta$, there is an exact formula for $mm_{\rho,\kappa}$ in terms of $mm_{\rho}$. 

\begin{lemma}\label{lem:rho-constant-same}
Consider an instance of the ANV model with success probability $\kappa$. Let $\rho(\theta)=r \in (0,1)$ for all $\theta \in \Theta$. Then $mm_{\rho,\kappa}=mm_{\kappa \rho}$.
\end{lemma}

The previous lemmas relate $mm_{\rho,\kappa}$ to $mm_{\rho}$ and tells us exactly how to compute or approximate the maximum mistreatment in the ANV model. It may happen however, that two RVPs, where the first has smaller maximum expected mistreatment when $\kappa=0$ may in fact have a larger one when $\kappa>0$, as shown in the next lemma. Thus, knowing $\kappa$ is important for choosing which RVP to implement.

\begin{lemma}\label{ANV:rho-flip} There exist choices of the parameters $\alpha,\beta,p,\kappa$ and RVPs $\rho, \hat \rho$ that are constant over $\Theta$ so that, for the ANV model defined by these parameters, we have $mm_{\rho}<mm_{\hat \rho}$ and $mm_{\rho,\kappa}>mm_{\hat \rho,\kappa}$. 
\end{lemma}

\subsection{Incentive compatible  allocation of vouchers when the distribution of the potential of students is unknown}\label{sec:IwP}

Let us consider again the original model where a voucher assigned to a student debias them with probability $1$. We next observe that, to design a non-trivial incentive compatible RVP, it is essential to have knowledge of the distribution of student potentials. We say an RVP $\rho$ is \emph{Increasing-with-Potential} (IwP) if $\rho(\theta)\geq\rho(\theta')$ for all $\theta>\theta'$. An IwP assigns a higher probability of being debiased to students with higher potential. It can therefore be interpreted as a randomized counterpart of the DDS from Lemma~\ref{lem:deterministic-incentive-compatible} that allocated vouchers to the top performing students (in particular, the DDS from Lemma~\ref{lem:deterministic-incentive-compatible} is IwP).

\paragraph{General distributions of potentials:} For the rest of this section, we relax the Pareto assumption and consider the more general version of the model defined in Section~\ref{sec:market} where the true potentials of students are allowed to be drawn from any continuously integrable distribution $F$. All definitions naturally extend to this setting. We first show that under mild technical conditions, IwPs are incentive compatible with respect to any $F$. We prove this fact in~\ecref{ec-subsec:iwp-rvps}.
\begin{lemma}\label{lem:well-behaved-IwP-is-ic}
Suppose $\rho$ is IwP and such that it is everywhere continuously differentiable except a countable set of isolated points where it has right-continuous jump discontinuities. Then, for any distribution of true potentials $F$, $\rho$ is incentive compatible.
\end{lemma}
The following theorem gives a converse to the previous statement and is also proved in~\ecref{ec-subsec:iwp-rvps}.
\begin{theorem}\label{thm:only-iwp-is-always-incentive compatible}
Suppose $\rho$ is an RVP. Let $\theta \in [1,\infty)$ such that $\rho$ is continuously differentiable in some neighborhood of $\theta$ but $\rho'(\theta)<0$. Then, for any $\beta\in(0,1)$ and $p\in(0,1)$, there exists a continuous distribution $F$ such that if true potentials are distributed according to $F$, $\rho$ is not incentive compatible.
\end{theorem}
Theorem~\ref{thm:only-iwp-is-always-incentive compatible} implies that without any information on the distribution of student potentials, the only voucher distribution policies that are guaranteed to be incentive compatible are those that allocate vouchers with higher probability to higher performing students. Examples of such policies are lotteries for students whose potential is above some certain threshold. Hence, if no information on the distribution of students potential can be assumed, it may be reasonable for policy-makers to stick to a more conservative distribution of vouchers which rewards top-performing students.

\paragraph{Discussion on technical assumptions:} We now discuss the technical assumptions on the parameters of the model in Section~\ref{sec:voucher} and Section~\ref{sec:incentive compatible}. In Theorem~\ref{thm:optimal-range-mm}, Theorem~\ref{thm:optimal-range-pauc}, and Theorem~\ref{thm:properties-of-ptm-rvps} we assume $p<1-\beta^\alpha$. Note that the right hand side is equal to $\mb(F_2\in[\beta,1])$, that is, the proportion of disadvantaged students whose perceived potential is less than $1$ (the smallest perceived potential of any non-disadvantaged student). The condition $p<1-\beta^\alpha$ therefore requires that the proportion of disadvantaged students out of the whole student population is no more than the proportion of disadvantaged students that are perceived as being worse than any non-disadvantaged student. In Theorem~\ref{thm:optimal-range-pauc}, we further assume $p<0.5$, and in Theorem~\ref{thm:properties-of-ptm-rvps} we assume both an upper and a lower bound on $\hat c$. The conditions need to be checked and do not always hold, but we note that all conditions hold for many reasonable choices of $(\alpha,\beta,p,\hat c)$. For instance, they hold if $\beta=0.8$, $\alpha=3$, $p<0.4$ (as in 
Figure~\ref{fig:best-range}(b)), and $\hat c \leq 1/4$; or if $\beta=0.9$, $\alpha=8.9$, $p\leq1/3$, and $\hat c\leq 1/4$ (as in our numerical experiments in Section~\ref{sec:experiments}). In Lemma~\ref{lem:all-or-nothing-vouchers}, for reasonable choices of the parameters, the increase of the maximum expected mistreatment is mild. For instance, for $\beta = 0.8$, $\alpha=4$, $p=.2$, $\kappa=0.8$, $\|\rho\|=0.1$, we obtain $m_{\rho,\kappa}\leq m_{\rho} + 0.012233088$, while for $\beta = 0.9$, $\alpha=4$, $p=.4$, $\kappa=0.5$, $\|\rho\|_\infty=0.2$ we obtain $m_{\rho,\kappa}\leq m_{\rho} + 0.0385168$.

\section{Experimental Case Study} \label{sec:experiments}

Our theoretical analysis has shown that mistreatment under various metrics can be substantially reduced via a targeted intervention tailored to the distribution of student potential. We now use data from NYC with real test scores and a student population with heterogeneous preferences over schools to compute optimal policies for reducing student mistreatment. We show that our theoretical model provides a reasonable approximation despite some deviations from the data (as discussed in the \emph{Pre-processing \& Model Fitting} paragraph below), and importantly, that our qualitative results continue to hold. Our theoretical analysis is thus instrumental in identifying effective debiasing policies for real-world applications, and can be optimized empirically for real data.  

\paragraph{Dataset:} There are eight Specialized High Schools (SHSs) in NYC which are consistently ranked among the best schools in the city. Admission to these schools is highly competitive, and is determined solely by the score a student achieves on the SHS Admissions Test (SHSAT). An intake of only about $5,000$ students gets selected every year from a pool of $29,000$ applicants who take the test. We apply our model to the dataset of 2016--17 academic year SHS admissions which include for each student their SHSAT score, their preference list over the SHSs, and whether the DOE deems them disadvantaged or not.

\paragraph{Pre-processing \& Model Fitting:} As in Section~\ref{sec:contribs}, we estimate the distributional shift in scores between the two groups (see Figure~\ref{fig:dist-shsat}), then assume that reversing this shift gives the innate ability for each student, and use this scaled score as the (in reality unobservable) true potential. In our dataset, we fit $\beta=0.882$ for a multiplicative model and $\gamma=49$ points for an additive model\footnote{In the additive model of bias, a disadvantaged student $\theta$ with real potential $Z(\theta)$ is assumed to have perceived potential $\hat Z(\theta)=Z(\theta)- \gamma$ for some universal constant $\gamma$. See~\ecref{sec:additive-bias} for details. In this section, we choose bias parameters that minimize the Wasserstein distance between the two distributions. The additive shift becomes $49/475=0.103$ after normalization.}. We take the original scores to be the perceived potentials, and the scaled scores to be the true potentials: that is, for $G_1$ students, we always use the raw SHSAT score, and for $G_2$ students with raw score $s$ in the dataset, we use $s/\beta$ (or $s+\gamma$) as their true potential if they receive a voucher (i.e.~are debiased) and $s$ otherwise. These choices are coherent with the assumptions made through most of the paper. In~\ecref{appx:misspecified} we discuss results obtained by relaxing some of these assumptions, which show that the key takeaways of our results still hold.

In this section, all matchings on real data are computed as stable matchings with the student-proposing deferred-acceptance algorithm, using the true preferences of students, and with schools choosing students based solely on their perceived score. This mimics closely the real SHS admissions process. We extend the definition of displacement in the natural way: as the difference in the ranking of the school a student is assigned to (in their own preference list) between the matching at hand and the matching that uses the estimated true potentials. Due to the heterogeneity of student preferences, the displacement for a given student may be positive, negative, or zero whether they receive a voucher or not. Due to this variability, and to the fact that our case study only has 8 schools, the maximum mistreatment measure is not a meaningful way to discriminate between different debiasing intervals. Thus, we use the positive area under the mistreatment\footnote{Recall mistreatment is the non-negative part of displacement.} curve (PAUC) measure exclusively to compare interventions as it averages out this effect.

Our theoretical model assumes a balanced market, whereas only a small number of those who apply to SHSs are admitted. We therefore discard those students whose score is lower than a cutoff of 475 points, whom we compute to not receive admissions in any case\footnote{We compute this cutoff by performing a stable matching using the true potentials and rounding down to the nearest 5 points the score of the last student that gets admitted in this matching.}. This right tail of the student scores now closely matches the Pareto distribution (after dividing by the minimum score) with $\alpha=8.856$ in the multiplicative case (see Figure~\ref{fig:pareto-fitting-shsat}) and $\alpha=9.315$ in the additive case. Of these students, we compute the proportion that are disadvantaged as $p=0.319$ for the multiplicative case and $p=0.300$ for the additive case\footnote{Because of the different distributional assumption, a different number of $G_2$ students end up above the cutoff in the two different models.}. These normalizations yield a subset of students that form a balanced market and whose distributional properties approximate our model well.

\begin{figure}[h]
    \centering
    \includegraphics[width=.75\linewidth]{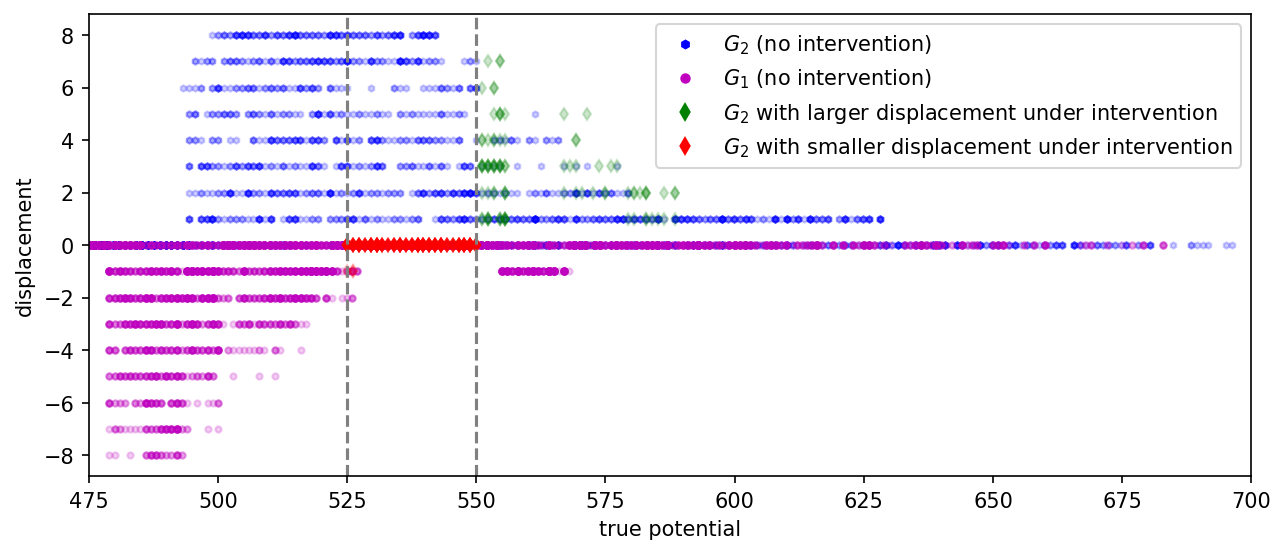}
    \includegraphics[width=.75\linewidth]{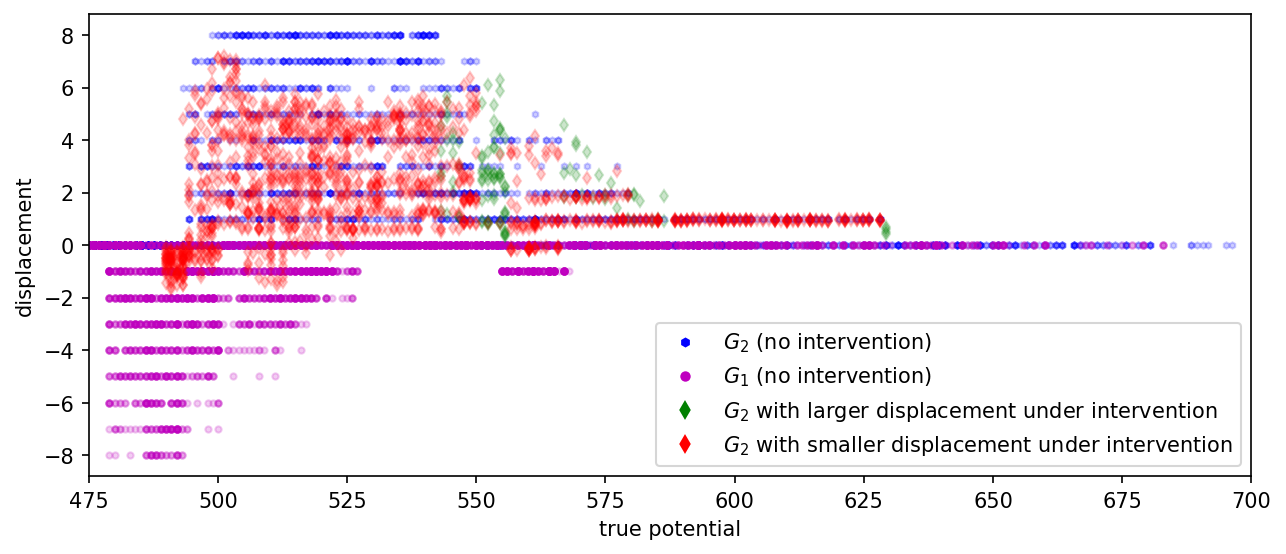}
    \caption[Displacement of $G_1$ and $G_2$ students after deterministic and randomized debiasing.]{The displacement of $G_1$ and $G_2$ students in the SHSAT dataset from NYC DOE. Blue and magenta dots respectively show the displacement of disadvantaged and non-disadvantaged students when there is no intervention. In the top figure, the students in the debiased range (dashed lines) are offered vouchers, in the bottom figure vouchers are offered with probability given in Figure~\ref{fig:empirical-ptm-rvp}. In the bottom figure with the randomized voucher program, we plot the average displacement over 100 repetitions. In both figures, we plot the displacement of disadvantaged students whose assigned schools change, with red dots representing those going to more preferred schools and green dots for those going to less preferred schools. }
    \label{fig:doe-disp-voucher}
\end{figure}

\paragraph{Model Characteristics:} We first observe empirically that without intervention, all $G_1$ students (magenta dots in Figure~\ref{fig:doe-disp-voucher}) have non-positive displacement and all $G_2$ students (blue dots) have non-negative displacement, as predicted by our analysis in Section~\ref{sec:analysis}. Furthermore, we consider deterministic and randomized interventions with $\hat c=0.17$. In Figure~\ref{fig:doe-disp-voucher} (top), deterministic vouchers are offered to students between the two dashed lines. All $G_2$ students who receive vouchers (red dots) have a displacement of at most zero, but some $G_2$ students might (green dots) fare worse, particularly the ones who are scoring slightly higher than the range to which vouchers are offered, as they are overtaken by some other $G_2$ students just below them. This highlights the non-incentive compatible nature of such deterministic policies (such students have incentive to underperform).

\begin{figure}[h]
    \centering
    \includegraphics[width=.8\linewidth]{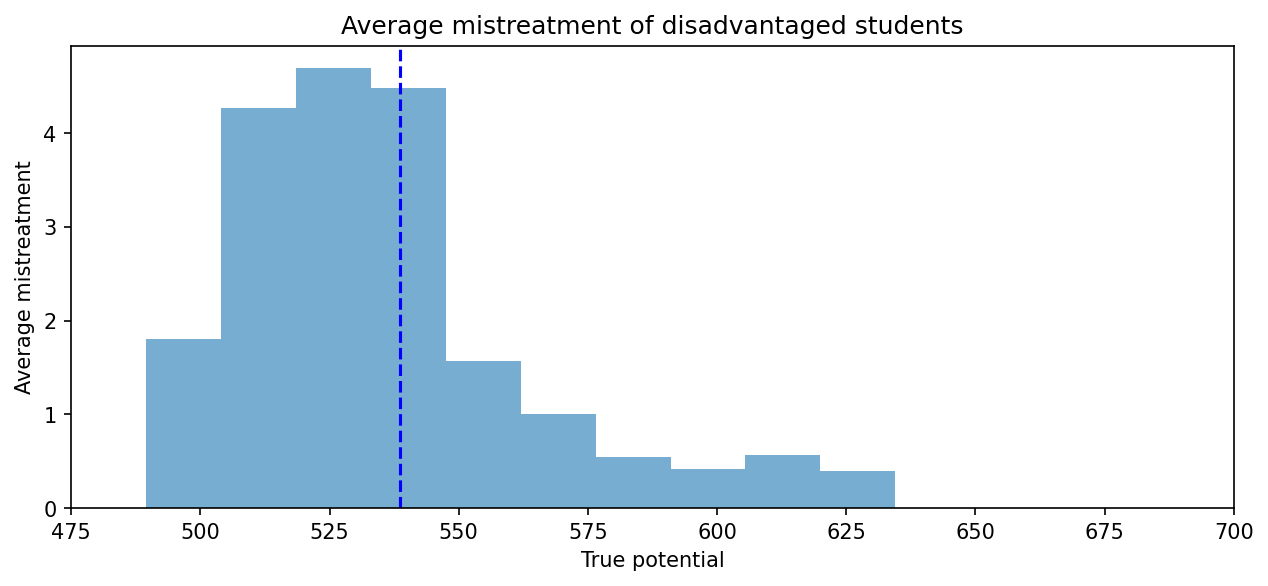}
    \caption[Average bucketed mistreatment computed from empirical data.]{The average bucketed mistreatment of students as computed from empirical data. Note the peak at 525, which represents an average student (cf.~Figure~\ref{fig:best-range}, the curve representing mistreatment before intervention). The vertical line indicates the value of $1/\beta$ where the theoretical PropM debiasing policy would have its maximum probability of assigning a voucher.} 
    \label{fig:empirical-ptm-rvp}
\end{figure}

Applying the randomized voucher program to this dataset requires further modifications. As observed earlier, since students have heterogeneous preferences, their mistreatment is also heterogeneous. In particular, case two students with the same score may have different mistreatment. To produce an empirical PropM (see Figure~\ref{fig:empirical-ptm-rvp}), we divide the potentials of admitted students into 20 equally sized buckets and compute an average mistreatment within each such bucket. The PropM is then normalized to be proportional to this average mistreatment, with magnitude determined by the budget of vouchers available. Since PropM is a randomized voucher program (such that a given student might get a voucher in one realization but not another), we run this experiment 100 times with different seeds and take the average displacement (see Figure~\ref{fig:doe-disp-voucher}, bottom). Our experiments show that the maximum mistreatment is reduced compared to the deterministic allocation
, and more generally, mistreatment improves across the board indicating a more equitable outcome. Note that due to the heterogeneity of preferences and binned averaging, PropM is not in fact incentive compatible, as indicated by the $G_2$ students with larger displacement under intervention (green diamonds). However, unlike the deterministic debiasing procedure, students with an incentive to underperform are interspersed with students with no incentive to underperform, making it harder for students to game the system ex-ante. Moreover, it is not uncommon that some theoretically incentive compatible mechanisms exhibits in practice some lack of incentive compatibility. For instance, the NYC School Match mechanism has until 2024 curtailed the preference lists of students to at most 12 schools~\citep{abdulkadirouglu2005new}, thus incentivizing students to be  strategic.

\paragraph{Theoretical vs. empirical intervals: } We next compare theoretically optimal intervals with those found to be empirically optimal. The basic primitive in our analysis is the routine that is given an interval of students to debias and computes the PAUC. By computing actual stable matchings and their PAUC, we produce a benchmark to compare both theoretical and empirical intervals. We use this routine in a grid search to find the empirically optimal debiasing intervals. We consider a low budget regime ($\hat c=0.1$) and an abundant budget regime ($\hat c= 0.4$). For both, we compute the theoretically optimal intervals by applying our theorems to the parameter values fitted earlier. 
Table~\ref{tab:optimal-range-doe} shows the optimal ranges found both via our theory and empirical grid search.
The differences between the theoretical and empirical models are minor, with the biggest difference being under the additive assumption. 
Thus, we find further evidence of reasonability of our assumptions since the empirical results on real-world data match the optimal target distribution of students predicted by our assortative model.

\begin{table}[t]
{\small{\begin{center}\begin{tabular}{l|l|c|cc}
        Model & Range & $\hat c = 0.1$  & $\hat c =0.4$ \\
        \midrule
        multiplicative & theoretical &
        $[526.73, 547.19]$ &
        $[506.84, 582.99]$ \\[2mm]
        multiplicative & empirical &
        $[527, 543]$ &
        $[505, 561]$ \\[2mm]

        additive & theoretical &
        $[516.94, 530.91]$ &
        $[499.82, 558.24]$ \\[2mm]

        additive & empirical &
        $[529, 544]$ &
        $[508, 567]$ \\[2mm]
    \end{tabular}
    \end{center}}}
    \caption[Empirically and theoretically found optimal debiasing ranges.]{Comparison of optimal ranges of students to offer vouchers to, obtained empirically and theoretically (based on our formulas), under two different budgets.}
    \label{tab:optimal-range-doe}
\end{table}

\section{Discussion}\label{sec:discussion}

The qualitative takeaways from our work speak to an ingrained systemic problem that limits access to opportunities---how can one understand the impact of these limits and systematically account for their effects? Indeed, resources available for meaningful interventions in an existing system are limited, and there is resistance to change in the form of lawsuits and pushback to changes in admissions policies. 
Thus, our focus is on understanding the impact of minimally invasive use of targeted resources at the pre-admissions stage, as opposed to changing the matching mechanism itself.

Our analysis highlights several qualitative properties using simple models of bias. There are multiple ways of designing interventions to reduce the impact of the observed performance differences. We show that the optimal debiasing interval is deterministically located around top average students. This maximally impacts measures of fairness. However, deterministic voucher distribution is not individually-fair and might create an incentive to underperform. Therefore, we propose randomized voucher distribution policies. These policies can counter some of the effects of bias, while also being incentive compatible and individually-fair (ex-ante\footnote{To make these policies fair ex-post, one would need to provide scholarships deterministically to all students above a certain threshold. Moreover, how to achieve ex-post individual-fairness in a resource-constrained system remains an open question.}). Our work therefore helps provide a mathematical basis for policymakers to make changes to allocation of public resources for the most good.

We next discuss some limitations of our model and our analysis. First, we note that randomization in voucher distribution can itself create a perception of unfairness ex-post. One way to overcome such perceptions may be to start voucher-programs in carefully selected schools with low economic index and give vouchers to all students at these schools, so that the priority overall is given to average students in the city-wide system. We leave the study of such interventions to future work. Next, we partition students using the definition of {\it disadvantaged} students as provided by the Department of Education of New York City. Should this definition change, one would need to recalibrate the ``bias" in the system, and apply our analysis accordingly. Also observe that our model assumes a multiplicative bias that is only group-dependent, and that a student to whom a voucher is allocated is afterwards perceived at their full potential. While we relax some of these assumptions in Section~\ref{sec:stoch-effect-vouchers}, a full theoretical analysis of more general model is an interesting open problem.

Finally, our analysis leads to open questions such as theoretically optimal interventions under other heterogeneous student preferences and qualitative analyses when student potentials are not Pareto-distributed.  

\paragraph{Acknowledgments.} The authors are deeply indebted to the editors and the reviewers for the many comments and suggestions on an earlier version of the manuscript.

\bibliographystyle{apalike}
\bibliography{refs}

\newpage

\appendix

\section{Discussion on discrete versus continuous models}\label{appx-sec:discrete-vs-cont}

Traditionally, matching markets are assumed to be discrete~\citep{gale1962college,roth1992two}. In recent years, however, there has been an interest in models where one or both sides of the markets are continuous~\citep{nick, azevedo2016supply}. This is justified by the fact that in many applications, markets are  large, hence predictions in continuous markets often translate with a good degree of accuracy to discrete ones. Moreover, continuous markets are often analytically more tractable than discrete ones (see, again,~\cite{nick, azevedo2016supply}). Our case is no exception: the continuous model allows us to deduce precise mathematical formulae, while we show through experiments that those formulae are a good approximation to the discrete case. 
We remark that the goal of this study is not to provide a mechanism to admits students to schools, for which the assumption of all rankings of schools as well as of students being the same would be too simplistic. On the contrary, as we want to understand the impact of bias at a macroscopic level, we believe our approximation to be meaningful and useful, since as in our model, any reasonable mechanism would output the same assignment.

\section{Impact on Schools}
\label{ec-sec:schools}

This appendix explores the school's perspective: the impact of bias on \emph{utility} (quality of accepted students) and \emph{diversity} for schools, as well as school-driven interventions such as interviews. In the notation of the two-group model in Section~\ref{sec:market}, we define the \emph{utility} $u_\gamma(s)$ of a school $s$ under matching $\gamma\in\set{\mu,\hat\mu}$, as

\begin{equation}  \label{eq:utility}
    u_\gamma(s):= \int_{\theta\in \gamma^{-1}(s)\cap G_1} Z(\theta) dF_1(\hat Z(\theta)) + \int_{\theta\in \gamma^{-1}(s)\cap G_2} Z(\theta) dF_2(\hat Z(\theta)).
\end{equation}

That is, the utility of a school is the average true potential of admitted students. We discuss first the impact of bias on the average true potential of students accepted by a school. Let $s\in [0,1]$ denote the school that is ranked in the $s\times 100\%$ position among the continuous range of schools. As the next proposition shows, the impact on the utilities of schools is negligible for all schools other than the lowest ranked schools. This is because for each school, although the average potential of assigned $G_1$ students is lower than it should be, its assigned $G_2$ students have much higher true potentials. And thus, the toll on the utility due to unqualified $G_1$ students is partially canceled out by the overqualified $G_2$ students and the net effect is minimal. On the other hand, some lower ranked schools that only admit $G_2$ students fare better in the biased setting (since they admit over-qualified $G_2$ candidates).

\begin{proposition} \label{prop:mistreat-schools}
	For school $s$, its utility under the unbiased (resp.~biased) models are respectively 
	$$u_{\mu}(s)= s^{-\frac{1}{\alpha}} \quad \hbox{ and } \quad u_{\hat\mu}(s)= \begin{cases}  
	    \displaystyle \frac{1-p+p\ba}{1-p+p\beta^{\alpha+1}} \left(\frac{s}{1-p+p\ba} \right)^{- \frac{1}{\alpha}} & \text{ if } s\le 1-p+p\ba, \\[1mm]
	    \displaystyle \left(\frac{s-(1-p)}{p} \right)^{- \frac{1}{\alpha}} & \text{ if } s> 1-p+p\ba.
	\end{cases}$$
\end{proposition}
The key idea in the proof is to first compute the \emph{cutoffs} at each school for each of the two groups, that is, the minimum \emph{perceived} potential needed for a student to be matched to a given school. Once these are known, using Bayes' rule, we deduce the minimum \emph{real} potential needed by students of each group to attend the school. From the latter, we can immediately compute the average utility of each school.

\smallskip 

\begin{mproof}{Proposition~\ref{prop:mistreat-schools}.}
In order for a student $\theta$ to be assigned to a school that is at least as good as $s$, their perceived potential $\hat Z(\theta)$ needs to be high enough to satisfy $(1-p)\bar F_1(1\vee \hat Z(\theta)) + p\bar F_2(\hat Z(\theta))\le s$. That is, we need
\begin{equation*} \label{eq:cutoff}
  \hat Z(\theta) \ge d(s):=\begin{cases}
    \displaystyle \left( \frac{s}{1-p+p\ba} \right)^{- \frac{1}{\alpha}}& \text{if } s\le 1-p+p\ba,\\[1mm]
    \displaystyle \beta \left( \frac{s-(1-p)}{p} \right)^{- \frac{1}{\alpha}}& \text{if } s> 1-p+p\ba.
  \end{cases}
\end{equation*}

We call $d(s)$ the \emph{cutoff} for school $s$. With the cutoffs, we can compute the utilities of schools. We start with the formula for $u_{\hat\mu}(s)$. First note that by Bayes rule, the probability that a given student with perceived potential $\hat Z( \theta)\ge 1$ belongs to $G_1$ is $\frac{1-p}{1-p+p\beta^{\alpha+1}}$. Using Equation~\eqref{eq:match-act}, observe that the $G_2$ student whose perceived potential is $1$ (i.e., true potential is $\invb$) is matched to school $1-p+p\ba$. Thus, if $s\ge 1-p+p\ba$, $s$ is only assigned with $G_2$ students. Therefore, when $s\le 1-p+p\ba$, 
$$u_{\hat\mu}(s)= \frac{1-p}{1-p +p\beta^{\alpha+1}} d(s)+ \frac{p\beta^{\alpha +1}}{1-p +p\beta^{\alpha+1}} \frac{d(s)}{\beta} = \frac{1-p+ p\ba}{1-p+p \beta^{\alpha+1}} \left( \frac{s}{1-p+p\ba} \right)^{- \frac{1}{\alpha}}.$$
And when $s> 1-p+p\ba$, we have
$$u_{\hat\mu}(s) = d(s)/\beta = \left( \frac{s-(1-p)}{p} \right)^{- \frac{1}{\alpha}}.$$
One the other hand, when there is no bias against $G_2$ students, we simply have $u_{\mu}(s)= s^{-\frac{1}{\alpha}}$. $\Box$ 
\end{mproof}

As one readily observes from Proposition~\ref{prop:mistreat-schools}, the negative impact of bias on schools' utility is negligible. Hence, from an operational perspective, it may be hard to convince schools to autonomously put in place mechanisms to alleviate the effect of bias given the limited impact on them.

\begin{figure}[t]
	\centering
	\includegraphics[width=.55\textwidth]{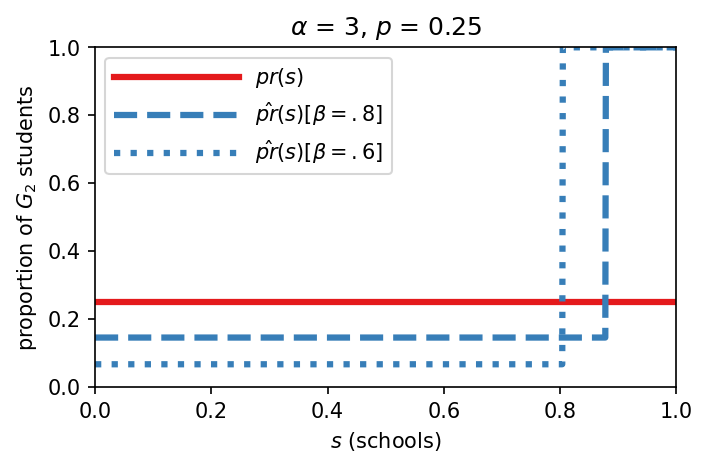}
    \caption{Proportion of $G_2$ students in higher ranked schools decreases significantly in the biased setting.} \label{ec-fig:proportion-minority}
\end{figure} 

Let $pr(s)$ (resp.~$\hat{pr}(s)$) be the proportion of $G_2$ students assigned to school $s$ when there is no bias (resp.~there is bias) against $G_2$ students. Since the distribution of potentials is the same for both $G_1$ and $G_2$ students, it is immediate that $pr(s)=p$ when there is no bias. 

\begin{proposition} \label{prop:diversity}
	Without bias, we have $pr(s)=p$. Under the biased setting, we have 
	$$\hat{pr}(s) =\begin{cases}
	    \displaystyle \frac{p\ba}{1-p+p\ba} & \text{ if } s\le 1-p+p\ba, \\
	    1 & \text{ if } s> 1-p+p\ba.
	\end{cases}$$
\end{proposition}
\proof{}
The formula for $\hat{pr}(s)$ follows from the analysis of utility of schools in Proposition~\ref{prop:mistreat-schools}.
\endproof

A visual comparison of $pr(s)$ and $\hat{pr}(s)$ can be found in Figure~\ref{ec-fig:proportion-minority} for different values of $\beta$ and $p$. In particular, we show that the proportion of $G_2$ students in higher ranked schools decreases significantly in the biased setting. 

\section{Proof of Theorem~\ref{thm:optimal-range-mm} and related facts}\label{ec-sec:proofs:or-mm}

\vspace{.8em}

\subsection{Technical discussion}

The main idea of the proof is to first assume that the set $T$ forms a connected set (i.e., a closed interval). Then, we can express $mm(\mu_T)$ as a function of the endpoints of $T$ and work out the minimizing interval. We next drop the assumption that $T$ is connected and show that the optimal set of students to debias remains the same. The analysis we give is actually more general, and presents results under which vouchers improve the mistreatment of students \emph{lexicographically}. Interestingly, it also shows that, if vouchers are not distributed carefully, one may actually {\it worsen the most mistreated students}.

\subsection{A more general approach}~\label{sec:debias}
The analysis we give is actually leads to a more general statement than Theorem~\ref{thm:optimal-range-mm}, and has the goal of investigating conditions under which giving vouchers can improve over the status quo. More formally, for bounded functions $f,g : G_2 \rightarrow \R$, we write $f\succ g$ if we can partition $G_2$ in two sets $S,S'$ (with possibly $S'=\emptyset$) so that $f(\theta)=g(\theta)$ for $\theta \in S'$ and $\sup_{\theta \in S} f(\theta) > \sup_{\theta \in S} g(\theta)$. Note that $\succ$ is transitive and antisymmetric, and can be interpreted as a continuous equivalent of the classical lexicographic ordering for discrete vectors. In particular, if we let $f=\gamma-\mu$ and $g=\gamma'-\mu$ for matchings $\gamma, \gamma'$, then $\sup_{\theta \in G_2} (\gamma-\mu)(\theta)>\sup_{\theta \in G_2} (\gamma'-\mu)(\theta)$ implies $f\succ g$ (taking $S=G_2$).
Now suppose we debias student in $T=[Z_1,Z_2]$ for some $T \in \mathcal{T}(\hat c)$, and let $f:= \hat\mu -\mu$, $g:= \mu_T -\mu$. Table~\ref{tab:condition-correction} provides conditions under which $f\succ g$ (i.e., intervention reduces the maximum mistreated experienced by $G_2$ students). In particular it shows that for certain combinations of the data and the choice of $Z_1$ and $Z_2$, giving vouchers may actually lead to worse (according to $\succ$) matchings. One can check that under assumption $p<1-\ba$, all conditions given in Table~\ref{tab:condition-correction} for different cases are satisfied.

\begin{table}[ht]
	\centering
	\begin{tabular}{l|l|l}
		\multicolumn{1}{c}{CASE} & \multicolumn{1}{c}{subcase} & \multicolumn{1}{c}{condition for $\hat\mu-\mu \succ \mu_T-\mu$} \\
		\hline
		\multirow{2}{*}{I. $\beta Z_2\ge Z_1$} & 1. $1\le \beta Z_1$ & $\displaystyle p<1- \left (\frac{Z_1}{Z_2} \right)^\alpha$ \\[.4cm]
		& 2. $\beta Z_1\le 1\le \beta Z_2$ & $\displaystyle p<1- \left( \frac{1}{\beta Z_2} \right)^\alpha$ \\[.4cm]
		\hline
		\multirow{3}{*}{II. $\beta Z_2\le Z_1$} & 1. $1\le \beta Z_1$ & $\displaystyle p<1-\ba$\\[.2cm]
		& 2. $\beta Z_1\le 1\le \beta Z_2$ & $\displaystyle p \left( \left(\frac{1}{Z_1} \right)^\alpha - \left( \frac{1}{Z_2} \right)^\alpha \right) < (1-p)  \left( \invba-1 \right) \left( \ba- \left( \frac{1}{Z_2} \right)^\alpha \right)$\\[.4cm]
		& 3. $\beta Z_2\le 1$ & Not possible: $g\succ f$ in this case. \\[.2cm]
		\hline
	\end{tabular}
	\vspace{.8em}
	\caption{Sufficient conditions for $\hat\mu-\mu \succ \mu_T-\mu$ by cases, where $T=[Z_1, Z_2]$. Each strict inequality, when replaced with its non-strict counterpart, gives instead a necessary condition.} \label{tab:condition-correction}
\end{table}

In this first part of the proof, we proceed as follows. First, we assume that $T\in \mathcal{T}^c(\hat c)$. That is, we assume $T=[Z_1, Z_2]$ with extreme points $Z_1 < Z_2$. For simplicity, we let $\tilde \mu$ denote $\mu_T$. We then compare $f:=\hat \mu- \mu$ and $g:=\tilde \mu- \mu$ using the relation $\succ$.

Note that, if we let $S$ be the set of students in $G_2$ with potential in $[Z_1, Z_2/\beta]$ and $S':=G_2\setminus S$, we have $f(\theta)=g(\theta)$ for $\theta \in S'$. That is, only $G_2$ students whose true potential lies in interval $[Z_1, Z_2/\beta]$ are affected by the intervention. Hence, $\sup_{\theta \in S} f > \sup_{\theta \in S} g $ if and only if $f\succ g$. 
We divide the analysis in the following two major cases: the first case is when $\beta Z_2\ge Z_1$ (i.e., when $[\beta Z_1, \beta Z_2]$ and $[Z_1, Z_2]$ overlap) and the second case is when $\beta Z_2\le Z_1$. For both major cases, we will consider two subcases: $\beta Z_1\ge 1$, $\beta Z_1\le 1\le \beta Z_2$. And for the second major case, we also need to consider the subcase where $\beta Z_2\le 1$. The results for all cases are summarized in the Table~\ref{tab:condition-correction}.

\begin{observation} \label{obs:case.2}
	If there is an interval $[Z_1,Z_2]$ that is of either case I.2 or case II.2 such that $\mu_{[Z_1, Z_2]}-\mu \prec \hat\mu-\mu$ with $S=G_2$, then the optimal range must be of case I.2 or case II.2. This is because for any interval $[Z_1', Z_2']$ that is not of case I.2 or case II.2, we have 
	$$\sup_{\theta \in \Theta} \{(\mu_{[Z_1', Z_2']} -\mu)\}\ge \sup_{\theta \in \Theta} \{\hat\mu- \mu\} > \sup_{\theta \in \Theta} \{\mu_{[Z_1,Z_2]} -\mu\}.$$
\end{observation}

As it turns out, indeed, the optimal range will be either case I.2 or case II.2, and exactly which one the optimal solution is depends on the amount of resources, i.e., the value of $\hat c$.

We now show the first half of Theorem~\ref{thm:optimal-range-mm}, i.e., we assume $\hat c\ge \frac{(1-p) (1 -\ba)}{1-p+1- \ba}$. The proof steps are outlined below. Each step can be shown by simple algebra and is thus omitted. 

\begin{itemize}
	\item[(1).] We first show that $[Z_1^*, Z_2^*]$ is of case I.2. That is, we  show $\beta Z_2^*\ge Z_1^*$ and $Z_1^*\le \invb \le Z_2^*$.
\end{itemize}

By writing out the formula for $\mu_{[Z_1, Z_2]}-\mu$, one can see that for an interval $[Z_1, Z_2]$ of case I.2 or case II.2, $\mu_{[Z_1, Z_2]}-\mu$ increases on $[1,Z_1]$, deceases on $[Z_2, \infty]$, and it is non-positive on $[Z_1, Z_2]$. This means $\sup_{\theta \in \Theta} \{\mu_{[Z_1, Z_2]}-\mu\}$ is achieved either at $Z_1$ or $Z_2$.

\begin{itemize}
	\item[(2).] Next, we show that $[Z_1^*, Z_2^*]$ is an \emph{exact} range, that is, $(\frac{1}{Z_1^*})^\alpha- (\frac{1}{Z_2^*})^\alpha = \hat c$. Moreover, let $\theta_1^*$ and $\theta_2^*$ be the $G_2$ students whose potentials are $Z_1^*$ and $Z_2^*$ respectively. Then, $(\mu_{[Z_1, Z_2]}-\mu)(\theta_1^*) = (\mu_{[Z_1, Z_2]}-\mu)(\theta_2^*)$ and thus, they are both equal to $\sup_{\theta \in \Theta} \{\mu_{[Z_1, Z_2]}-\mu\}$.
\end{itemize}

Together with the assumption $p<1-\ba$, we have $\sup_{\theta \in \Theta} \{\mu_{[Z_1^*, Z_2^*]}  -\mu\} \leq \sup_{\theta \in \Theta} \{\hat\mu-\mu\}$. Thus, due to Observation~\ref{obs:case.2}, it is sufficient to compare $[Z_1^*, Z_2^*]$ only with intervals $[Z_1, Z_2]$ of case I.2 and case II.2 (i.e, when $\beta Z_1\le 1\le \beta Z_2$). Since $[Z_1^*, Z_2^*]$ is exact, we must either have $Z_1> Z_1^*$ or $Z_2< Z_2^*$.

\begin{itemize}
	\item[(3).] Lastly, we show that for any other feasible range $[Z_1,Z_2]$ of case I.2 or case II.2, we must have $\sup_{\theta \in \Theta} \{\mu_{[Z_1, Z_2]} -\mu\} > \sup_{\theta \in \Theta} \{\mu_{[Z_1^*, Z_2^*]} -\mu\}$. Let $\theta_1$ and $\theta_2$ be the $G_2$ students whose potentials are $Z_1$ and $Z_2$. It suffices to show
	\begin{enumerate}
		\item[i).] if $Z_1> Z_1^*$, then $(\mu_{[Z_1, Z_2]}-\mu) (\theta_1) > (\mu_{[Z_1^*, Z_2^*]} -\mu) (\theta_1^*)$;
		\item[ii).] if $Z_2< Z_2^*$, then $(\mu_{[Z_1, Z_2]}-\mu) (\theta_2) > (\mu_{[Z_1^*, Z_2^*]} -\mu) (\theta_2^*)$.
	\end{enumerate}
\end{itemize}	

For the second half of the theorem, we will follow similar steps and reasoning, outlined below.

\begin{itemize}
	\item[(1).] We first show that $[Z_1^*, Z_2^*]$ is of case II.2. That is to show $\beta Z_2^*\le Z_1^*$ and $Z_1^*\le \invb \le Z_2^*$.
	\item[(2).] We check that $[Z_1^*, Z_2^*]$ is an exact range. And let $\theta_1^*$ and $\theta_2^*$ be the $G_2$ students whose potentials are $Z_1^*$ and $Z_2^*$ respectively, we want to show that $(\mu_{[Z_1, Z_2]}-\mu)(\theta_1^*) = (\mu_{[Z_1, Z_2]}-\mu)(\theta_2^*)$, which implies that both are $\sup_{\theta \in \Theta} \{\mu_{[Z_1^*, Z_2^*]}-\mu\}$.
	\item[(3).] We show $\mu_{[Z_1^*, Z_2^*]}-\mu \prec \hat\mu-\mu$, which, unlike in the previous case, is not immediate from the assumption $p<1-\ba$.
\end{itemize}

Again, due to Observation~\ref{obs:case.2}, it is sufficient to compare $[Z_1^*, Z_2^*]$ only with regions $[Z_1, Z_2]$ of case I.2 and case II.2 (i.e, when $\beta Z_1\le 1\le \beta Z_2$).
\begin{itemize}
	\item[(4).] As before, we will show two cases, which is enough because $[Z_1^*, Z_2^*]$ is exact and one of the two cases is bound to happen. Again, let $\theta_1$ and $\theta_2$ be the $G_2$ students whose potentials are $Z_1$ and $Z_2$ respectively. We want to show
	\begin{enumerate}
		\item[i).] if $Z_1> Z_1^*$, then $(\mu_{[Z_1, Z_2]}-\mu) (\theta_1) > (\mu_{[Z_1^*, Z_2^*]} -\mu) (\theta_1^*)$,
		\item[ii).] otherwise, we must have $Z_2< Z_2^*$, and then $(\mu_{[Z_1, Z_2]}-\mu) (\theta_2) > (\mu_{[Z_1^*, Z_2^*]} -\mu) (\theta_2^*)$.
	\end{enumerate}
\end{itemize}	

Now let $T^*\in \mathcal{T}(\hat c)$ be the optimal solution without the restriction that sets in $\mathcal{T}(\hat c)$ are connected. We will show that $T^*$ differs from $[Z_1^*, Z_2^*]$ in a set of measure zero. First, in order to have $\sup (\mu_{T^*} -\mu)\le \sup(\mu_{[Z_1^*, Z_2^*]} -\mu)=:s$, in $T^*$, we must debias all students $\theta$ whose mistreatment $(\hat\mu-\mu)(\theta)$ is greater than $s$. That is, we must have $T_1^*:= [Z_1^*, Z^{(1)}]\subseteq T^*$, where $Z^{(1)}:=Z(\theta^{(1)})\ge 1/\beta$ and $(\hat\mu-\mu) (\theta^{(1)}) =s$. There is a $G_2$ student $\theta^{(2)}$ such that $Z^{(2)}:= Z(\theta^{(2)}) >Z^{(1)}$ and $(\mu_{T_1^*}-\mu) (\theta^{(2)}) = s$. We have moreover that $(\mu_{T_1^*}-\mu)(\theta) \ge s$ for all $\theta\in G_2$ such that $Z(\theta)\in [Z^{(1)}, Z^{(2)}]$. Thus, we must also have $[Z^{(1)}, Z^{(2)}]\in T^*$. Let $T_2^*:=[Z_1^*, Z^{(2)}]$. We can repeat the argument and observe that there is a $G_2$ student $\theta^{(3)}$ such that $Z^{(3)}:= Z(\theta^{(3)}) > Z^{(2)}$ and $(\mu_{T_2^*}-\mu) (\theta)\ge s$ for $\theta\in G_2$ such that $Z(\theta)\in [Z^{(2)}, Z^{(3)}]$ and conclude that $T_3^*:=[Z_1^*, Z^{(3)}]$ must be contained in $T^*$. Continuously applying the same argument, we have  $\lim_{n\rightarrow\infty} Z(\theta^{(n)}) = Z_2^*$ and thus the claim follows.

\section{Proof of Theorem~\ref{thm:optimal-range-pauc} and related facts} \label{ec-sec:proofs:or-pauc}

Assume $T=[Z_1, Z_2]$ is the range of true potentials of $G_2$ students we want to debias. For simplicity, as in previous sections, let $\tilde\mu$ denote $\mu_T$. In order to obtain the minimizer of $\sigma(\tilde\mu-\mu)$, first, we want to compute $\sigma(\tilde\mu-\mu)$ for each of the cases in Table~\ref{tab:condition-correction}. 

For $1\leq t_1\leq t_2 \in \R\cup\{+\infty\}$, let $\sigma_{t_1}^{t_2}(f):= \int_{t_1}^{t_2} \max(f(t),0) dF_1(t)$ for any function $f:[1,\infty]\rightarrow [0,1]$. When $t_1=1$ and $t_2=\infty$, we simply write $\sigma(f)$, which is consistent with previous notations. Note that with $\sigma(\hat\mu-\mu)$ as a reference, it actually suffices to compute only $\sigma_{Z_1}^{Z_2/\beta} (\tilde\mu-\mu)$, because minimizing $\sigma(\tilde\mu-\mu)$ is equivalent to maximizing $\sigma_{Z_1}^{Z_2/\beta} (\hat\mu-\mu)- \sigma_{Z_1}^{Z_2/\beta} (\tilde\mu-\mu)$ since $(\hat\mu-\mu)(\theta)= (\tilde\mu-\mu)(\theta)$ for all $\theta\in G_2$ with $Z(\theta)\notin [Z_1, Z_2/\beta]$. 

For each case, we give an explicit formula for $\sigma_{Z_1}^{Z_2/\beta} (\hat\mu-\mu)- \sigma_{Z_1}^{Z_2/\beta} (\tilde\mu-\mu)$. These formulae can be computed via simply integration, and are thus omitted. In addition, we analyze how this value changes (increase or decrease) with respect to $Z_1$ and $Z_2$.

\noindent{\bf CASE I -- Subcase 1.} After integrating, we have
\begin{align*}
	\sigma_{Z_1}^{Z_2/\beta} (\hat\mu-\mu)- \sigma_{Z_1}^{Z_2/\beta} (\tilde\mu-\mu) &= \frac{(1-p) \left( \invba-1 \right)}{2} \left( \frac{1}{Z_1} \right)^{2\alpha} +\frac{p-p\ba-\invba+1}{2} \left( \frac{1}{Z_2} \right)^{2\alpha}.
\end{align*}

Now, to analyze how this quantity changes with $Z_1$ and $Z_2$, we first simplify some of the terms, which will also be used in later sections. Let $x=( \frac{1}{Z_2} )^\alpha\in [0,1]$ and let $( \frac{1}{Z_1} )^\alpha = c+x\in [0,1]$ for some $c\le \hat c$. Also, let $g(x,c):= \sigma_{Z_1}^{Z_2/\beta} (\hat\mu-\mu)- \sigma_{Z_1}^{Z_2/\beta} (\tilde\mu-\mu)$. Then,
$$g(x,c)= \frac{(1-p) \left( \invba-1 \right)}{2}(c+x)^2 +\frac{p-p\ba-\invba+1}{2} x^2.$$
First order conditions (FOC) show that $g(x,c)$ increases as $x$ increases (or equivalently, as $Z_2$ decreases) and as $c$ increases (meaning that the constraint $( \frac{1}{Z_1})^\alpha -( \frac{1}{Z_2} )^\alpha \le\hat c$ is effectively $( \frac{1}{Z_1} )^\alpha -( \frac{1}{Z_2} )^\alpha =\hat c$). 

\noindent{\bf CASE I -- Subcase 2.} In this case, we have
\begin{align*}
	\sigma_{Z_1}^{Z_2/\beta} (\hat\mu-\mu)- \sigma_{Z_1}^{Z_2/\beta} (\tilde\mu-\mu) &= -\frac{1}{2}(1-p)\ba +(1-p) \left( \left( \frac{1}{Z_1} \right)^{\alpha} -\frac{1}{2} \left( \frac{1}{Z_1} \right)^{2\alpha} \right) \\
	&+ \frac{p-p\ba-\invba+1}{2} \left( \frac{1}{Z_2} \right)^{2\alpha}.
\end{align*}

Now, for the analysis, similarly, write 
$$g(x,c) = \text{const} + (1-p)\left( (c+x)- \frac{1}{2} (c+x)^2 \right) + \frac{p-p\ba-\invba+1}{2} x^2.$$
Then, FOC shows that  $g(x,c)$ is an increasing function w.r.t. $c$, and it is an increasing function w.r.t. $x$ on $[0,h_\text{I}(c)]$ and is a decreasing function on $[h_\text{I}(c),1]$, where
$$h_\text{I}(c) = \frac{(1-p)(1-c)}{p\ba+\invba-2p}.$$

\noindent{\bf CASE II -- Subcase 1.} In this case, $\sigma_{Z_1}^{Z_2/\beta} (\hat\mu-\mu)- \sigma_{Z_1}^{Z_2/\beta} (\tilde\mu-\mu)$ equals to
{\small\begin{align*}
    \left( \frac{1}{Z_1} \right)^{2\alpha} \left( \frac{(1-p) \left( \invba-1 \right) +p\ba}{2} \right) - \left( \frac{1}{Z_2} \right)^{2\alpha} \left( \frac{(1-p) \left( \invba-1 \right) +p\ba}{2} \right) + p\left( \frac{1}{Z_2} \right)^{2\alpha} -p \left( \frac{1}{Z_1} \right)^{\alpha} \left( \frac{1}{Z_2} \right)^{\alpha}.
\end{align*}
}
Now, for the analysis, let $ A= [(1-p) (\invba-1) +p\ba]/2 \ge 0$. Then,
$$g(x,c) = A(c+x)^2 - Ax^2 + px^2 -p(c+x)(x).$$
Checking the FOCs, we have that $g(x,c)$ is an increasing function w.r.t. $c$ and w.r.t. $x$.

\noindent{\bf CASE II -- Subcase 2.} We have
{\small\begin{align*}
	\sigma_{Z_1}^{Z_2/\beta} (\hat\mu-\mu)- \sigma_{Z_1}^{Z_2/\beta} (\tilde\mu-\mu) = &-\frac{1}{2}(1-p)\ba + \left( \frac{1}{Z_1} \right)^{2\alpha} \bigg( \frac{-(1-p)+p\ba}{2} \bigg) +(1-p) \left( \frac{1}{Z_1} \right)^\alpha \\
	&- \left( \frac{1}{Z_2} \right)^{2\alpha} \bigg( \frac{(1-p) \left( \invba-1 \right) +p\ba}{2} \bigg)+ p\left( \frac{1}{Z_2} \right)^{2\alpha} -p \left( \frac{1}{Z_1} \right)^{\alpha} \left( \frac{1}{Z_2} \right)^{\alpha}.
\end{align*}
}

For the analysis, again let $A=[(1-p) (\invba-1) +p\ba]/2\ge 0$ and $B=[-(1-p)+p\ba]/2<0$. Then,
$$g(x,c) = \text{const} + B(c+x)^2 +(1-p)(c+x) - Ax^2 + px^2 -p(c+x)(x),$$
and for $c$, it is an increasing function; and for $x$, it is an increasing function on $[0,h_{\text{II}}(c)]$ and is a decreasing function on $[h_{\text{II}}(c),1]$, where
$$h_{\text{II}}(c)=\frac{(p\ba-1)c+(1-p)}{(1-p)\invba}.$$

\noindent{\bf CASE II -- Subcase 3.} Lastly, we have that $\sigma_{Z_1}^{Z_2/\beta} (\hat\mu-\mu)- \sigma_{Z_1}^{Z_2/\beta} (\tilde\mu-\mu)$ equals to
{\small\begin{align*}
    \left( \frac{1}{Z_1} \right)^{2\alpha} B +(1-p) \left( \frac{1}{Z_1} \right)^\alpha - \left( \frac{1}{Z_2} \right)^{2\alpha} B -(1-p) \left( \frac{1}{Z_2} \right)^\alpha + p\left( \frac{1}{Z_2} \right)^{2\alpha} -p \left( \frac{1}{Z_1} \right)^{\alpha} \left( \frac{1}{Z_2} \right)^{\alpha}.
\end{align*}
}

For the analysis, write
$$g(x,c)=B(c+x)^2+(1-p)(c+x) -Bx^2-(1-p)x +px^2-p(c+x)x.$$
$g(x,c)$ is a decreasing function in $x$. The sign of $\frac{\partial g(x,c)}{\partial c}$ is actually not clear in this subcase. But for the purpose of finding the minimizer of $\sigma(\tilde\mu-\mu)$, this is not important because for a fixed value of $c$, $g(x,c)$ achieves its maximum when $x$ is of the value such that $[Z_1, Z_2]$ is of subcase 2, of either case I or case II. 

Not that for a fixed value of $\hat c$, as $Z_1$ gets larger (or equivalently as $Z_2$ gets larger, or as $x:=( \frac{1}{Z_2} )^\alpha$ gets smaller), the range $[Z_1, Z_2]$ goes from case II to case I. In particular, for each value of $\hat c$, such transition happens exactly when $\beta Z_2=Z_1$. That is, when
$$\hat c=  \left( \invba-1 \right) \left( \frac{1}{Z_2} \right)^\alpha \quad \Leftrightarrow\quad \left( \frac{1}{Z_2} \right)^\alpha =\frac{\hat c\ba}{1-\ba}.$$

\begin{figure}[t]
	\centering
	\includegraphics[width=\textwidth]{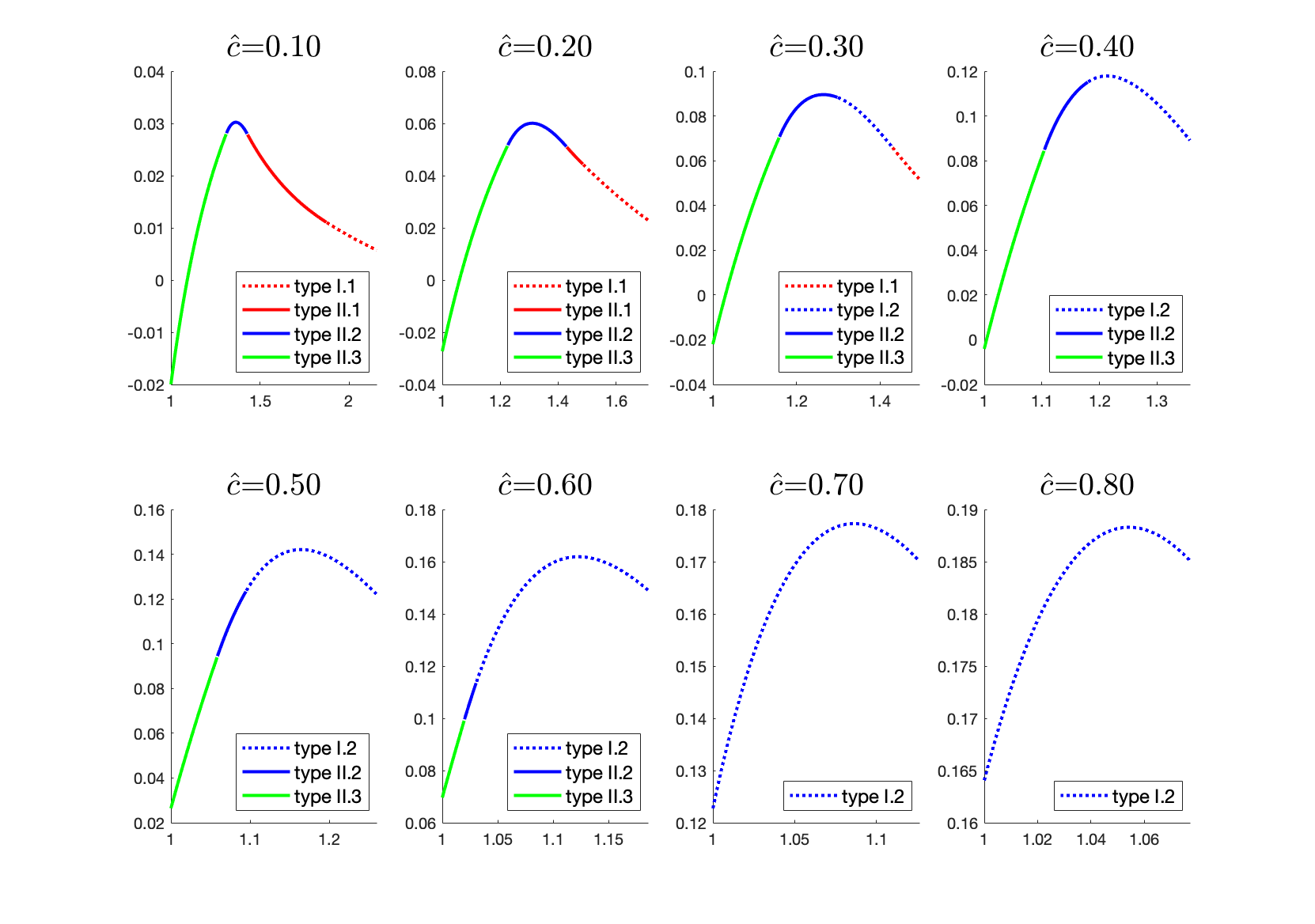}
	\caption{For fixed values of $\hat c$, the function of $\sigma(\hat\mu-\mu)-\sigma(\tilde\mu-\mu)$ on $Z_1$} \label{fig:auc-t1-hcs}
\end{figure}

Now, for each fixed value of $\hat c$, Figure~\ref{fig:auc-t1-hcs} plots $\sigma(\hat\mu-\mu)-\sigma(\tilde\mu-\mu)$ against $Z_1$. It also shows that as $Z_1$ increases, how the interval $[Z_1, Z_2]$ changes by cases.

With simple algebra, one can easily check that 
$$\hat c=\frac{(1-p) (1-\ba)}{2-p-\ba-p\ba +p\beta^{2\alpha}} \quad \Rightarrow \quad h_\text{I}(\hat c) = h_{\text{II}}(\hat c) =\frac{\hat c\ba}{1-\ba}.$$ 
Therefore, when $$\hat c\ge\frac{(1-p) (1-\ba)}{2-p-\ba-p\ba +p\beta^{2\alpha}},$$ both $h_\text{I}(\hat c)$ and $h_{\text{II}(\hat c)}$ are no more than $\displaystyle \frac{\hat c\ba}{1-\ba}$. Thus, the maximum value of $\sigma(\hat\mu-\mu)-\sigma(\tilde\mu-\mu)$ is achieved when $x=h_\text{I}(\hat c)$; and when $$\hat c\le\frac{(1-p) (1-\ba)}{2-p-\ba-p\ba +p\beta^{2\alpha}},$$ both $h_\text{I}(\hat c)$ and $h_{\text{II}(\hat c)}$ are no less than $\displaystyle \frac{\hat c\ba}{1-\ba}$. Thus, the maximum value of $\sigma(\hat\mu-\mu)-\sigma(\tilde\mu-\mu)$ is achieved when $x=h_{\text{II}}(\hat c)$.

\begin{table}[t]
	\centering
	\begin{tabular}{c|ccccc}
	\footnotesize
		$\hat c$ & $\mathcal{T}_{mm}(\hat c) = [Z_1, Z_2]$ & $\mathcal{T}_{auc}(\hat c) = [Z_1', Z_2']$ & $Z_1-Z_1'$ \\[2mm]
		\hline
		0.10 & [1.2252, 1.3111] & [1.2187, 1.3026] & 0.0065 \\[2mm] 
		0.20 & [1.2022, 1.3861] & [1.1903, 1.3653] & 0.0119 \\[2mm]
		0.30 & [1.1802, 1.4803] & [1.1644, 1.4421] & 0.0158 \\[2mm] 
		0.40 & [1.1461, 1.5584] & [1.1346, 1.5203] & 0.0115 \\[2mm]
		0.50 & [1.1156, 1.6560] & [1.1070, 1.6155] & 0.0086 \\[2mm]
		0.60 & [1.0881, 1.7839] & [1.0819, 1.7403] & 0.0063 \\[2mm]
		0.70 & [1.0632, 1.9635] & [1.0589, 1.9154] & 0.0043 \\[2mm]
		0.80 & [1.0404, 2.2476] & [1.0377, 2.1926] & 0.0026 
	\end{tabular}
	\vspace{.8em}
	\caption{\raggedright Compare the optimal ranges of $G_2$ students to debias under two measures of unfairness, under parameters $\alpha=3$, $\beta=.8$, and $p=.25$. We check the optimal intervals under both measures of unfairness, and find on an average 95\% overlap of the optimal intervals.
	}
	\label{ec-tab:compare-two-measures}
\end{table}

\section{Proof of Lemma~\ref{lem:deterministic-incentive-compatible}}\label{appx-sec:proof-deterministic-incentive-compatible}
Assume $T$ is not of the form $\{\theta \in \Theta : Z(\theta)\geq \delta\}$ or $\{\theta \in \Theta : Z(\theta)> \delta\}$ for some $\delta$ and let $U$ be a connected and inclusionwise maximal subset of $T$ that is bounded. Take the smallest number $\delta_1 \in [1,\infty)$ so that $Z(\theta)\leq \delta_1$ for all $\theta \in U$. Let $\overline \theta \in \theta$ such that $Z(\overline \theta)=\delta_1$.

Assume first that $\overline \theta \in U$. Then, for each $\epsilon_1 >0$, there exists $\epsilon \in
(0,\epsilon_1]$ such that $Z^{-1}(\delta_1 + \epsilon)\notin T$. Since $\beta<1$ and by continuity of $Z(\cdot)$, there exists $\epsilon_2>0$ such that $\beta (\delta_1+\epsilon) < \delta_1$ for all $\epsilon<\epsilon_2$. We can then take an appropriate $x\in[\delta_1,\delta_1+\epsilon_2]\setminus T$ and $x'=\delta_1$ to show that $T$ is not incentive compatible. Next assume  $\overline \theta \notin U$. In particular, we have $\overline \theta\notin T$. Similarly to the case above, we can find $\epsilon>0$ such that $x'=\delta_1-\epsilon$ satisfies $Z^{-1}(x') \in U\subseteq T$ and $\beta \delta_1<x'$. Setting  $x=\delta_1$, we deduce that $T$ is not incentive compatible.

\section{Proofs from Section~\ref{sec:randomized-vouchers}} \label{ec-sec:rvp}

\subsection{Auxiliary results for Section~\ref{sec:randomized-vouchers}}

Recall that, in this section, we consider the generalization of the model from Section~\ref{sec:market} where students' true potential follow a generic continuous, integrable cdf $F$. Moreover, we write $\maxz{x}:=\max(0,x)$ for a number or a function $x$. Recall that, similarly to Section~\ref{sec:randomized-vouchers}, we abuse notation and identify a student $\theta$ with their potential $Z(\theta)$.

\begin{lemma}\label{claim:mu_rho}
Let $\rho$ be an RVP. Under any continuous distribution of potentials $F$, we have
\begin{align}
\mu_\rho(\theta)&=\rho(\theta)\br{(1-p)\int_\theta^\infty dF+p\sbr{\int_\theta^{\theta/\beta}\rho\,dF+\int_{\theta/\beta}^\infty dF}} +(1-\rho(\theta)) \cdot \nonumber \\ & \br{(1-p)\int_{\beta\theta}^\infty dF+p\sbr{\int_{\beta\theta}^\theta\rho\,dF+\int_\theta^\infty dF}},\label{eq:mu_rho} \\
m_\rho(\theta)&=[0,(1-\rho(\theta))(1-p)\int_{\beta\theta}^\theta\,dF+p\sbr{(1-\rho(\theta))\int_{\beta\theta}^\theta\rho\,dF-\rho(\theta)\int_\theta^{\theta/\beta}(1-\rho)\,dF}]^+.\label{eq:m_rho}
\end{align}
\end{lemma}
\proof{}
Suppose a student appears to have potential $\tau$, possibly after having been debiased. Then under $\mu_\rho$, they will be matched to school $s(\tau)$ given by
$$s(\tau)=(1-p)\int_\tau^\infty\,dF+p\sbr{\int_\tau^{\tau/\beta}\rho\,dF+\int_{\tau/\beta}^\infty\,dF},\label{eq:s-z}
$$
that is, they will appear after all non-disadvantaged students with true potential exceeding $\tau$; those disadvantaged students with potential exceeding $\tau/\beta$; and those disadvantaged students who receive a voucher and have potential in the interval $(\tau,\tau/\beta)$.

A student with true potential $\theta$ now receives a voucher with probability $\rho(\theta)$, so by the law of total expectation, have
$$\mu_\rho(\theta)=\rho(\theta)s(\theta)+(1-\rho(\theta))s(\beta\theta),
$$
which is exactly \eqref{eq:mu_rho}. \eqref{eq:m_rho} follows from~\eqref{eq:mu_rho} and the definitions of displacement and $\mu(\theta)$. 
\endproof

We next report more expressions for $\mu_\rho$ and $\mu_\rho'$, as they will be used in the upcoming proofs.

\begin{proposition}\label{prop:mu_rho_expr}
Let $\rho$ be an RVP. For all $\theta \in \Theta$, we have
\begin{align}
\mu_\rho(\theta)&=-\rho(\theta)\br{(1-p)\int_{\beta\theta}^\theta dF+p\sbr{\int_\theta^{\theta/\beta}(1-\rho)\,dF+\int_{\beta\theta}^\theta\rho\,dF}} \nonumber \\& +\br{(1-p)\int_{\beta\theta}^\infty dF+p\sbr{\int_{\beta\theta}^\theta\rho\,dF+\int_\theta^\infty dF}}.\label{eq:mu_rho2} \end{align}
Moreover, if $\mu_\rho$ is differentiable at $\theta$, we have \begin{align}\mu_\rho'(\theta)&=-f(\theta)-\rho'(\theta)\sbr{p\int_\theta^{\theta/\beta}(1-\rho)\,dF+(1-p)\int_{\beta\theta}^\theta\,dF+p\int_{\beta\theta}^\theta\rho\,dF} \nonumber \\
&\qquad-p\rho(\theta)\sbr{\frac{1}{\beta}(1-\rho(\theta/\beta))f(\theta/\beta)-(1-\rho(\theta))f(\theta)} \nonumber \\
&\qquad+(1-\rho(\theta))\sbr{(1-p)(f(\theta)-\beta f(\beta\theta))+p(f(\theta)\rho(\theta)-\beta f(\beta\theta)\rho(\beta\theta))}.\label{eq:mu_rho_deriv}
\end{align}
\end{proposition}
\proof{}
\eqref{eq:mu_rho2} follows by simple rearrangement of \eqref{eq:mu_rho}, and \eqref{eq:mu_rho_deriv} follows by standard mechanics of derivative computation. 
\endproof

\begin{definition}
The RVP that assigns no vouchers, denoted $\rho_0$ is defined by $\rho_0(\theta):=0$ for all $\theta \in [1,\infty)$. Note that $\mu_{\rho_0}(\theta)=(1-p)\int_{\beta\theta}^\infty dF+p\int_\theta^\infty dF$ and $m_{\rho_0}(\theta)=m(\theta)=(1-p)\int_{\beta\theta}^\theta\,dF$.
\end{definition}

\subsection{Necessary and sufficient conditions for incentive compatibility}

In this section we develop necessary and sufficient conditions for incentive compatibility through the concept of well-behavedness and prove an important technical lemma.

\begin{definition}[Well-behaved RVP]
We call an RVP $\rho$ \emph{well-behaved} if it is everywhere continuously differentiable except for a set of isolated points where it has non-negative, right-continuous jump discontinuities.
\end{definition}

\begin{lemma}[Necessary and sufficient conditions for incentive compatibility]\label{lemma:incentive compatible}
Let $\rho$ be a well-behaved RVP and $F$ be an arbitrary continuous distribution of potentials. $\rho$ is incentive compatible with respect to $F$ if and only if, for all $\theta$ such that $\rho$ is continuously differentiable at $\theta$, we have $\rho'(\theta)\geq 0$ or $\mu_\rho'(\theta)\leq 0$.
\end{lemma}
\proof{}
Recall that $\rho$ is incentive compatible if $\mu_\rho$ is everywhere non-increasing. Observe from \eqref{eq:mu_rho2} in Proposition~\ref{prop:mu_rho_expr} that $\mu_\rho$ is continuous at $\theta$ if and only if $\rho$ is continuous at $\theta$. On the other hand, if $\mu_\rho$ is not continuous at $\theta$ then it must have a negative jump-discontinuity caused by a positive jump-discontinuity of $\rho$ (since all other terms of \eqref{eq:mu_rho2} are positive). Further note that if $\mu_\rho$ is not continuously differentiable at $\theta\in\Theta$, then $\rho$ is not continuously differentiable at $\theta$, $\beta\theta$ or $\theta/\beta$; so the set of points where $\mu_\rho$ is not continuously differentiable also forms an isolated set.

Consider any $\theta\in\Theta$ where $\mu_\rho$ is continuously differentiable, then $\mu_\rho$ is non-increasing iff $\mu_\rho'(\theta)\leq 0$. By collecting the terms for $f(\theta)$, $f(\beta\theta)$ and $f(\theta/\beta)$ in \eqref{eq:mu_rho_deriv}, one can see that $\mu_\rho'(\theta)\leq 0$ if $\rho'(\theta)\geq 0$.

We have established that $\mu_\rho$ is continuous at all but an isolated set of negative jump-discontinuities, and that $\mu_\rho$ is continuously differentiable and non-increasing at all but an isolated set of points. $\mu_\rho$ is therefore everywhere non-increasing, as required. 
\endproof

\begin{lemma}\label{claim:technical-incentive compatible}
Suppose $\rho$ is a well-behaved RVP such that for all $\theta$ where $\rho$ is continuously differentiable, we have $\rho'(\theta)\geq-\phi(\theta)$, with
\begin{align*}
\phi(\theta):=\frac{\alpha(1-p)}{\theta\sbr{p(1-\beta^\alpha)+(1-p)(\beta^{-\alpha}-1)}}.
\end{align*}
Then, $\rho$ is incentive compatible.
\end{lemma}
\proof{}
By Lemma~\ref{lemma:incentive compatible}, it suffices to show that $\mu_\rho'(\theta)\leq 0$ for $\theta$ such that $\rho'(\theta)$ exists and is continuous, and $\rho'(\theta)<0$.
Now define
$$
\mcL=p\int_\theta^{\theta/\beta}(1-\rho)\,dF+(1-p)\int_{\beta\theta}^\theta\,dF+p\int_{\beta\theta}^\theta\rho\,dF \hbox{ and }
W=-\rho(\theta)(1-p)f(\theta)-(1-\rho(\theta))(1-p)\beta f(\beta\theta),
$$
and note that $\mcL\geq 0$, and $W\leq 0$. Simple calculations based on~\eqref{eq:mu_rho_deriv} in Proposition~\ref{prop:mu_rho_expr} shows that $\mu_\rho'(\theta)\leq -\rho'(\theta)\mcL+W$. It is therefore enough to prove $-\rho'(\theta)\leq\frac{-W}{\mcL}$. Compute next
\begin{align*}
\mcL&\leq p\int_\theta^{\theta/\beta}\,dF+(1-p)\int_{\beta\theta}^\theta\,dF \leq \theta^{-\alpha}\sbr{p(1-\beta^\alpha)+(1-p)(\beta^{-\alpha}-1)} \hbox{ and }\\
-W & \geq(1-p)\min\set{f(\theta),\beta f(\beta \theta)}=\frac{\alpha(1-p)}{\theta^{1+\alpha}}.
\end{align*}
This yields
\begin{align*}
\frac{-W}{\mcL}&\geq
\frac{\alpha(1-p)}{\theta\sbr{p(1-\beta^\alpha)+(1-p)(\beta^{-\alpha}-1)}} 
=\phi(\theta).
\end{align*}
We have shown $\frac{-W}{\mcL}\geq\phi(\theta)$, which combined with the assumption that $\phi(\theta)\geq-\rho'(\theta)$ completes the proof.
\endproof

\subsection{Proof of Theorem~\ref{thm:properties-of-ptm-rvps}: properties of PropMs} \label{ec-subsec:rvp-ptm-proof}

We next prove Lemmas~\ref{lemma:ptm-p1}, \ref{lemma:ptm-p2}, and~\ref{lemma:ptm-p3}, which together constitute Theorem~\ref{thm:properties-of-ptm-rvps}.

\begin{lemma}\label{lemma:ptm-p1}
The proportional-to-mistreatment RVP $\rho_m$ is $\frac{2\hat c\alpha}{1-\beta^\alpha}$-individually fair.
\end{lemma}
\proof{}
$\rho_m$ is everywhere continuous and continuously differentiable on $\Theta$, except at $\theta=1/\beta$. $\rho_m$ is therefore Lipschitz for a constant given by the supremum of the absolute value of the derivative, which occurs at $\theta=1$ where $\rho_m'(\theta)=\frac{2\hat c\alpha}{1-\beta^\alpha}$.
\endproof

\begin{lemma}\label{lemma:ptm-p2}
The proportional-to-mistreatment RVP $\rho_m$ is incentive compatible for
\begin{align*}
\hat c\leq\frac{1-p}{2\sbr{p(1-\beta^\alpha)+(1-p)(\beta^{-\alpha}-1)}}.
\end{align*}
\end{lemma}
\proof{}
Applying Lemma~\ref{claim:technical-incentive compatible}, it suffices to show
\begin{align}
-\rho'_m(\theta)=2\alpha \hat c\beta^{-\alpha}\theta^{-\alpha-1}&\leq
\phi(\theta)=\frac{\alpha(1-p)}{\theta\sbr{p(1-\beta^\alpha)+(1-p)(\beta^{-\alpha}-1)}}.
\end{align}
for $\theta\geq 1/\beta$ (since $\rho_m'(\theta)\geq0$ for $\theta<1/\beta$). Note that this is tightest when $\theta=1/\beta$, which gives the condition for $\hat c$. 
\endproof

\begin{lemma}\label{prop:sup-of-1-minus-rho}
For the proportional-to-mistreatment RVP $\rho_m$, we have
$$ \sup_{\theta\in\Theta}\,(1-\rho_m(\theta))\int_{\beta\theta}^\theta\,dF=(1-\beta^\alpha)\xi(\hat c), \quad \hbox{ where} \quad 
\xi(\hat c):=\piecewise{1-2\hat c,&\hat c\leq 1/4,\\ \frac{1}{8\hat c},& \hat c>1/4.}
$$\end{lemma}
\proof{}
With $F(a,b):=\int_a^b\,dF$, define $q:=\frac{2\hat c}{1-\beta^\alpha}$ and $y(\theta):=F(\beta\theta,\theta)$. Now write
\begin{align*}
(1-\rho_m(\theta))\int_{\beta\theta}^\theta\,dF&=\br{1-\frac{2\hat c}{1-\beta^\alpha}F(\beta\theta,\theta)}F(\beta\theta,\theta) =(1-qy(\theta))y(\theta).
\end{align*}

This is a quadratic in $y$ that increases from $y=0$ to its maximum at $y=1/(2q)$. Observe that $\max_\theta F(\beta\theta,\theta)=1-\beta^\alpha$ which is attained at $\theta=1/\beta$. This means that if $\hat c\leq 1/4$, then
\begin{align*}
y(\theta)=F(\beta\theta,\theta)\leq 1-\beta^\alpha\leq\frac{1-\beta^\alpha}{4\hat c}=\frac{1}{2q},
\end{align*}
and the maximum of the quadratic over $y$ is realized at the maximum value of $y$. Thus,
\begin{align*}
\sup_{\theta\in\Theta,\hat c\leq 1/4}\,(1-qy(\theta))y(\theta) =\br{1-\frac{2\hat c}{1-\beta^\alpha}(1-\beta^\alpha)}(1-\beta^\alpha) =(1-\beta^\alpha)(1-2\hat c).
\end{align*}
On the other hand if $\hat c>1/4$, then this expression reaches its maximum when the quadratic does, at $y=1/(2q)$, giving
\begin{align*}
\sup_{\theta\in\Theta,\hat c>1/4}\,(1-qy(\theta))y(\theta)&=\br{1-q\frac{1}{2q}}\frac{1}{2q} =\br{1-\frac{1}{2}}\frac{1}{2q} =\frac{1-\beta^\alpha}{8\hat c}.
\end{align*}

Combining the two completes the proof.
\endproof

\begin{lemma}\label{lem:mm}
For the proportional-to-mistreament RVP, $\rho_m$, the maximum mistreament $mm_{\rho_m}$ satisfies
\begin{align*}
mm_{\rho_m}\leq(1-p(1-2\hat c))(1-\beta^\alpha)\xi(\hat c),
\end{align*}
where $\xi(\cdot)$ is defined as in Lemma~\ref{prop:sup-of-1-minus-rho}.
\end{lemma}
\proof{}
Abbreviate $\rho=\rho_m$. Apply $\rho(\theta)\leq\rho(1/\beta)$ to \eqref{eq:m_rho} and simplify to get
\begin{align*}
m_{\rho_m}(\theta)
&\leq\maxz{(1-p)(1-\rho(\theta))\int_{\beta\theta}^\theta\,dF+p(1-\rho(\theta))\rho(1/\beta)\int_{\beta\theta}^\theta\,dF} \\
&=\maxz{(1+p(\rho(1/\beta)-1))(1-\rho(\theta))\int_{\beta\theta}^\theta\,dF}.
\end{align*}
The thesis follows by taking the supremum over $\theta\in\Theta$, applying Lemma~\ref{prop:sup-of-1-minus-rho}, and substituting $\rho(1/\beta)=2\hat c$.
\endproof

Recall that we let $mm^*(\hat c)$ be the maximum mistreatment achieved by the optimal policy from Theorem~\ref{thm:optimal-range-mm} with the amount of resources being $\hat c$. We have
\begin{align}
mm^*(\hat c)&=\begin{cases}(1-p-\hat c)(1-\beta^\alpha)+\hat cp& \hbox{if } \hat c\leq\frac{(1-p)(1-\beta^\alpha)}{1-p+1-\beta^\alpha}, \\
(1-p)(1-\beta^\alpha)\frac{1-\hat c}{1-p\beta^\alpha}& \text{otherwise.}\end{cases}
\end{align}

\begin{lemma}\label{lemma:ptm-p3}
Suppose $p<1-\beta^\alpha$, $p\leq 1/2$, and $\hat c\leq\frac{(1-p)(1-\beta^\alpha)}{1-p+1-\beta^\alpha}$. If $\hat c \geq 1-\frac{p+1-\beta^\alpha}{4p(1-\beta^\alpha)}$, then $mm_{\rho_m}\leq mm^*(\hat c)$.
\end{lemma}
\proof{}
Let $Q:=mm^*(\hat c)-mm_{\rho_m}$, we need to show $Q\geq 0$. Using Theorem~\ref{thm:optimal-range-mm} and Lemma~\ref{lem:mm}, compute
\begin{align*}
Q&=mm^*(\hat c)-mm_{\rho_m}
\geq(1-p)(1-\beta^\alpha)-\hat c(1-\beta^\alpha-p)-(1+p(2\hat c-1))(1-\beta^\alpha)\xi(\hat c).
\end{align*}
For $\hat c\leq 1/4$, we now have
\begin{align}\label{eq:Q-geq-0}
Q\geq \hat c\sbr{(1-4p(1-\hat c))(1-\beta^\alpha)+p}.
\end{align}
If $p\leq \frac{1}{4}$, then the right-hand side of \eqref{eq:Q-geq-0} is nonnegative, concluding the proof. Thus, assume $p > \frac{1}{4}$. Since $\hat c>0$, we can drop the leading $\hat c$, so for $Q\geq 0$, we need
\begin{align*}
p(1-4(1-\beta^\alpha)(1-\hat c))\geq-(1-\beta^\alpha).
\end{align*}
Rearranging leads to the thesis. 

Consider next the case where $\hat c\geq1/4$. In this case we want to show the inequality
\begin{align*}
(1-p)(1-\beta^\alpha)-\hat c(1-\beta^\alpha-p)-(1+p(2\hat c-1))(1-\beta^\alpha)\frac{1}{8\hat c}\geq 0.
\end{align*}
Again since $\hat c>0$, we can multiply by $\hat c$ to get a quadratic in $\hat c$; call the resulting expression $W(\hat c)$:
\begin{align*}
W(\hat c)&=\hat c(1-p)(1-\beta^\alpha)-\hat c^2(1-\beta^\alpha-p)-(1+p(2\hat c-1))(1-\beta^\alpha)\frac{1}{8}.
\end{align*}
Since $p<1-\beta^\alpha$, $W''(\hat c)\leq 0$, hence this is a concave quadratic. One can verify that if $p\leq1/2$ then $W(1/4)\geq0$ and $W(1/2)\geq 0$, which means that $W$ must also be non-negative for $\hat c\in[1/4,1/2]$, as required.
\endproof

\section{Proofs from Section~\ref{sec:stoch-effect-vouchers}}\label{ec-subsec:proof-new}

\begin{mproof}{Lemma~\ref{lem:all-or-nothing-vouchers}.} Define $m_{\rho,\kappa}(\theta)$ as the expected mistreatment experienced by a disadvantaged student $\theta \in \Theta\cap G_2$. Recall from Lemma~\ref{claim:mu_rho} that $m_\rho(\theta) =[M_\rho]^+(\theta)$, where 
$$M_\rho(\theta)=(1-\rho(\theta))(1-p)\int_{\beta\theta}^\theta\,dF+p\sbr{(1-\rho(\theta))\int_{\beta\theta}^\theta\rho\,dF-\rho(\theta)\int_\theta^{\theta/\beta}(1-\rho)\,dF}.$$ Similarly, we have $m_{\rho,\kappa}(\theta) = [M_{\rho,\kappa}(\theta)]^+$, where 
{$$M_{\rho,\kappa}(\theta)=(1-\kappa \rho(\theta))(1-p)\int_{\beta\theta}^\theta\,dF+p\sbr{(\kappa-\kappa^2 \rho(\theta))\int_{\beta\theta}^\theta\rho\,dF-\kappa \rho(\theta)\int_\theta^{\theta/\beta}(1-\kappa\rho)\,dF}.$$}
We can then write,
{\small \begin{align*}
m_{\rho,\kappa}(\theta) &= \left[ M_{\rho,\kappa}(\theta) \right]^+ \\
&= \left[ M_\rho(\theta) + (1 - \kappa)\rho(\theta)(1 - p)\int_{\beta\theta}^{\theta} dF - p(1 - \kappa)\int_{\beta\theta}^{\theta} \rho \, dF + p(1 - \kappa^2)\rho(\theta)\int_{\beta\theta}^{\theta} \rho \, dF \right. \\
&\quad \left. + p(1 - \kappa)\rho(\theta)\int_{\theta}^{\theta/\beta} dF - p(1 - \kappa^2)\rho(\theta)\int_{\theta}^{\theta/\beta} \rho \, dF \right]^+ \\
&\le \left[ M_\rho(\theta) + (1 - \kappa)\rho(\theta)(1 - p)\int_{\beta\theta}^{\theta} dF + p(1 - \kappa^2)\rho(\theta)\int_{\beta\theta}^{\theta} \rho \, dF + p(1 - \kappa)\rho(\theta)\int_{\theta}^{\theta/\beta} dF \right]^+ \\
&= \left[ M_\rho(\theta) + (1 - \kappa)\rho(\theta) \left( (1 - p)\int_{\beta\theta}^{\theta} dF + p(1 + \kappa)\int_{\beta\theta}^{\theta} \rho \, dF + p\int_{\theta}^{\theta/\beta} dF \right) \right]^+ \\
&\le \left[ M_\rho(\theta) + (1 - \kappa)\|\rho\|_\infty \left( (1 - p)\int_{\beta\theta}^{\theta} dF + p(1 + \kappa)\|\rho\|_\infty \int_{\beta\theta}^{\theta} dF + p\int_{\theta}^{\theta/\beta} dF \right) \right]^+ \\
&\le \left[ M_\rho(\theta) \right]^+ + (1 - \kappa)(1 - \beta^\alpha)\|\rho\|_\infty \left( 1 + p(1 + \kappa)\|\rho\|_\infty \right)\\
& = m_{\rho}(\theta) + (1 - \kappa)(1 - \beta^\alpha)\|\rho\|_\infty \left( 1 + p(1 + \kappa)\|\rho\|_\infty \right),
\end{align*}}
where the second relation is obtained by the using the formula for $M_\rho(\theta)$, the third by dropping nonnegative terms, the fourth by rearranging terms, the fifth by replacing $\rho(\theta)$ by its maximum, the sixth by upper bounding both $\int_{\beta \theta}^\theta dF$ and $\int_{\theta}^{\theta/\beta} dF$ by $(1-\beta^\alpha)$, and moving nonnegative terms out of the $[ \cdot ]^+$ function, and the last by definition of $m_{\rho}(\theta)$. By taking the supremum of $m_{\rho,\kappa}(\theta)$  for $\theta \in \Theta \cap G_2$, the upper bound follows.  $\Box$
\end{mproof}

\begin{mproof}{Lemma~\ref{lem:rho-constant-same}.}
From the proof of Lemma~\ref{lem:all-or-nothing-vouchers}, we have that, for $\theta \in \Theta$,
 \begin{align*}M_{\rho}(\theta)& = (1-r)(1-p)\int_{\beta\theta}^\theta\,dF + p\sbr{(1-r)r\int_{\beta\theta}^\theta\,dF -r(1-r)\int_{\theta}^{\theta/\beta}\,dF}\\
&  = (1-r)\sbr{(1-p+pr) \int_{\beta\theta}^\theta\,dF -pr \int_{\theta}^{\theta/\beta}\,dF}, \end{align*}
and that
\begin{align*}M_{\rho,\kappa}(\theta)& = (1-\kappa r)(1-p)\int_{\beta\theta}^\theta\,dF + p\sbr{(\kappa -\kappa^2 r)r\int_{\beta\theta}^\theta\,dF -\kappa r(1-\kappa r) \int_{\theta}^{\theta/\beta}\,dF}\\
& = (1-\kappa r)\sbr{(1-p+pkr)\int_{\beta\theta}^\theta\,dF -pkr \int_{\theta}^{\theta/\beta}\,dF} = M_{\kappa\rho}.
 \end{align*}
 $\Box$ \end{mproof}

\begin{mproof}{Lemma~\ref{ANV:rho-flip}.}
Assume that the potential of students is described by a Pareto distribution with $\alpha=1$. Let $p=0.8, \beta=0.2, \kappa =0.4; \rho(\theta)=0.8$ and $\hat \rho(\theta)=0.5$ for $\theta \in \Theta$. Let $\rho^*(\theta)=r$ for $\theta \in \Theta$. Note that, for every $\hat \kappa$, $mm_{\rho^*,\hat \kappa}$, is achieved at the value of $\theta$ that maximizes the mistreatment under $\hat \mu$, that is at $\hat \theta=1/\beta=5$. Compute
$\int_{\beta\hat \theta}^{\hat \theta}\,dF=0.8$, $\int_{\hat \theta}^{\hat \theta/\beta}\,dF=0.16$ 
and substitute them into the formula for $M_{\rho^*}(5)$ computed in the proof of Lemma~\ref{lem:rho-constant-same}:
$$M_{\rho^*}(\hat \theta) = (1-r)[(0.2+0.8r)(0.8)-0.8r(0.16)] =  (1 - r) (0.16 + 0.512r).$$
Thus we obtain
$
mm_{\rho}=0.11392, \, mm_{\hat \rho}=0.208, \, mm_{\rho,\kappa}=mm_{0.4\rho}> 0.22, \, mm_{\hat \rho,\kappa}=mm_{0.4\hat\rho}<0.21,
$
which verifies the thesis.
 $\Box$ \end{mproof}
\section{Increasing-with-Potential RVPs} \label{ec-subsec:iwp-rvps}

\begin{mproof}{Lemma~\ref{lem:well-behaved-IwP-is-ic}.}
Directly from Lemma~\ref{lemma:incentive compatible}. $\Box$
\end{mproof}

\begin{mproof}{Theorem~\ref{thm:only-iwp-is-always-incentive compatible}.}
We claim that there exists $\dta>0$ with $\dta<\theta(1-\beta)$ such that on $I:=(\theta-\dta,\theta+\dta)$, the following properties hold for all $t\in I$: $\rho$ is continuous and differentiable at $t$; $\rho$ is monotonically decreasing at $t$; and $0<\rho(t)\leq (1+\rho(\theta))/2$. The existence of an interval that satisfies the first and second properties follows since $\rho$ is continuously differentiable in some neighborhood of $\theta$ and has strictly negative derivative. The third follows since $\rho$ has a strictly negative derivative at $\theta$, so it must be strictly bounded away from $0$ and $1$ itself, and then one can restrict $\dta$ to guarantee the same for $t$ close to $\theta$. Note also that $I\subset(\beta\theta,\theta/\beta)$.

Next, fix $\eps>0$, then one can construct a distribution $f$ that satisfies the following conditions: $f$ is continuous and differentiable everywhere; $f(\theta)=\eps$;
$f(t)=0$ for $t\not\in I$; and
$\int_\theta^{\theta+\dta}f(t)\,dt\geq\frac{1}{2}$. This can be done for instance by constructing a piece-wise constant function that satisfies all but the first condition, then smoothing it out with an appropriate bump function via standard techniques.

From \eqref{eq:mu_rho_deriv} of Proposition~\ref{prop:mu_rho_expr} we can compute
\begin{align}
\mu_\rho'(\theta)
&=-\rho'(\theta)\br{
  p\int_\theta^{\theta/\beta}(1-\rho)\,dF+(1-p)\int_{\beta\theta}^\theta\,dF+p\int_{\beta\theta}^\theta\rho\,dF
} \nonumber \\
&\qquad
  -\frac{1}{\beta}p\rho(\theta)f\br{\frac{\theta}{\beta}}\br{1-\rho\br{\frac{\theta}{\beta}}}
  -\beta(1-\rho(\theta))f(\beta \theta)(1-p(1-\rho(\beta \theta))) \nonumber \\
&\qquad
  -f(\theta)(p(1-\rho(\theta))(1-2\rho(\theta))+\rho(\theta)) \nonumber \\
&=(-\rho'(\theta))\br{
  p\int_\theta^{\theta+\dta}(1-\rho)\,dF+(1-p)\int_{\theta-\dta}^\theta\,dF+p\int_{\theta-\dta}^\theta\rho\,dF
}   -\eps(p(1-\rho(\theta))(1-2\rho(\theta))+\rho(\theta)) \nonumber \\
&\geq(-\rho'(\theta))p\int_\theta^{\theta+\dta}(1-\rho)\,dF-\eps \nonumber \\
&\geq\frac{1}{2}(-\rho'(\theta))p(1-\rho(\theta))\int_\theta^{\theta+\dta}\,dF-\eps \nonumber \\
&\geq\frac{1}{4}(-\rho'(\theta))p(1-\rho(\theta))-\eps \label{eq:last}.
\end{align}
Here we used the fact that $p(1-\rho(\theta))(1-2\rho(\theta))+\rho(\theta)\in[0,1]$.

Now the first term in~\eqref{eq:last} is strictly positive, and we can freely choose $\eps$ strictly smaller in magnitude to get $\mu_\rho'(\theta)>0$. Note that although $\rho$ might not be well-behaved everywhere, it is well behaved on $I$, and we can apply Lemma~\ref{lemma:incentive compatible} to this point to get that $\rho$ is not incentive compatible for $\theta$, completing the proof. $\Box$
\end{mproof}

\section{Alternate Models of Bias}\label{sec:additive-bias}

We now investigate deviations from the model studied so far, which assumes that disadvantaged students suffer from a uniform multiplicative bias.

\begin{table}[h]
    \centering
{\small{    \begin{tabular}{c|ccc|ccc|cc}
    \toprule
    Budget & \multicolumn{3}{c}{Multiplicative} & \multicolumn{3}{c}{Additive} & \multicolumn{2}{c}{Difference} \\
    $\hat c$ & PAUC & MM & Difference & PAUC & MM & Difference & PAUC & MM \\
    \midrule
    0.1 & 44.4\% & 45.2\% & 0.9\% & 45.4\% & 46.1\% & 0.7\% & 1.1\% & 0.9\% \\
    0.2 & 37.6\% & 39.3\% & 1.7\% & 39.7\% & 41.1\% & 1.4\% & 2.1\% & 1.8\% \\
    0.3 & 30.8\% & 33.3\% & 2.5\% & 34.1\% & 35.9\% & 1.8\% & 3.3\% & 2.6\% \\
    0.4 & 26.4\% & 28.5\% & 2.0\% & 30.2\% & 31.8\% & 1.6\% & 3.7\% & 3.3\% \\
    0.5 & 22.0\% & 23.7\% & 1.7\% & 26.0\% & 27.4\% & 1.4\% & 4.0\% & 3.7\% \\
    0.6 & 17.6\% & 19.0\% & 1.4\% & 21.6\% & 22.8\% & 1.2\% & 4.0\% & 3.9\% \\
    0.7 & 13.2\% & 14.2\% & 1.0\% & 17.0\% & 18.0\% & 0.9\% & 3.8\% & 3.7\% \\
    0.8 &  8.8\% &  9.5\% & 0.7\% & 12.1\% & 12.7\% & 0.7\% & 3.3\% & 3.3\% \\
    \bottomrule
    \end{tabular}}}
    \smallskip
    \caption[Proportion of disadvantaged students above theoretically optimal debiasing ranges under various scenarios.]{Proportion of disadvantaged students above theoretically optimal debiasing ranges under multiplicative/additive models and PAUC/maximum mistreatment aggregate mistreatment measures for $\alpha=3$ and $p=1/4$, with $\beta=0.8$ for multiplicative models and $\gamma=0.252$ for additive. For instance, the first value, 44.4\% indicates that under budget $\hat c=0.1$ and the multiplicative model, the top-44.4\% to 54.4\% of students were debiased.}

    \label{ec-tab:compare-two-measures2}
\end{table}

\subsection{An additive bias model} We first extend some of the theoretical results in Section~\ref{sec:voucher} to the case of additive (in place of multiplicative) bias. 
We consider the same setting as Section~\ref{sec:market} where the multiplicative bias model was introduced, but now the perceived potential of a student is given by $\hat Z(\theta)=Z(\theta)-\gamma$ if $\theta\in G_2$ (and $\theta$ does not receive a voucher), with $\hat Z(\theta)=Z(\theta)$ otherwise. 
This then gives $F_1(t)=1-t^{-\alpha}$ and $F_2(t)=1-(t+\gamma)^{-\alpha}$,
with domains $[1,\infty)$ and $[1-\gamma,\infty)$, respectively. The expression given in~\eqref{eq:match-act} for the biased matching $\hat\mu(\theta)$ continues to hold under the additive definition of $\hat Z$, and
this fact yields the following result, which is an analogue of Proposition~\ref{prop:mistreat-students}.
\begin{proposition} \label{prop:mistreat-students-additive}
Under additive bias of $\gamma\geq0$, for any student $\theta \in G_2$, the displacement under $\hat\mu$ is given by:
\begin{align}
\disp_{\hat\mu}(\theta)
&=\piecewise{
(1-p)\br{(Z(\theta)-\gamma)^{-\alpha}-(Z(\theta))^{-\alpha}}, & \text{ if } Z(\theta)\geq 1+\gamma, \\
(1-p)\br{1-(Z(\theta))^{-\alpha}}, & \text{ if } Z(\theta)<1+\gamma.}
\end{align}
For any student $\theta\in G_1$, we have $\disp_{\hat\mu}(\theta)=-p((Z(\theta))^{-\alpha}-(Z(\theta)+\gamma)^{-\alpha})$. Thus, the maximum displacement of $(1-p)(1-(1+\gamma)^{-\alpha})$ is experienced by a $G_2$ student with potential $1+\gamma$; and the most significant negative displacement of $-p(1-(1+\gamma)^{-\alpha})$ is experienced by a $G_1$ student with potential $1$.
\end{proposition}
 We next extend Theorem~\ref{thm:optimal-range-mm} which establishes the optimal debiasing interval under maximum mistreatment to the additive model. Let $\mcS^c(\hat c)$ be the set of closed and connected subsets of $[1,\infty)$ such that for $S\in\mcS^c(\hat c)$, $\int_S\,dF_1\leq\hat c$. Similarly to the multiplicative case, we define $\mu_S$ to be the matching under the additive model when students in $S$ receive vouchers. Let then $\mcS^c_{mm}(\hat c)=\argmin_{S\in\mcS^c(\hat c)}mm(\mu_S)$ be the collection of sets in $\mcS^c(\hat c)$ that minimize maximum mistreatment. The following result is analogous to Theorem~\ref{thm:optimal-range-mm}.
\begin{theorem}\label{thm:optimal-range-mm-additive}
The set $\mcS^c_{mm}(\hat c)$ consists of a unique set $S=[Y_1^*,Y_2^*]$ where $Y_1^*$ and $Y_2^*$ are computed as follow: $Y_1^*=(\hat c+(Y_2^*)^{-\alpha})^{-1/\alpha}$ and  $Y_2^*=\min\set{U_1,U_2}$, where $U_1$ is the positive solution to $(1-p)(1-\hat c)=\mbb{F\geq U_1-\gamma}-p\mbb{F\geq U_1}$
and $U_2$ is the positive solution to
$(1-p)(1-\hat c)=(1-p)\mbb{F\geq U_2-\gamma}+p\hat c.
$
\end{theorem}

\proof
To simplify notation, we define the following. Recall $F\distributed\Pareto(1,\alpha)$ is the distribution of true potentials of all students, and let $\phi$ be a \emph{bias map} that gives the perceived potential of a disadvantaged student given their true potential. $\phi$ therefore encodes information both about the nature of bias, as well as any debiasing steps taken. In this case we consider an additive bias model where we debias an interval $S=[Y_1,Y_2]$, so that $\phi(x)=x$ if $x\in S$ and $\phi(x)=x-\gamma$ otherwise. Let $H\distributed \phi(F)$ be the distribution of perceived potentials of disadvantaged students, and then note that $(1-p)F+pH$ is the distribution of perceived potentials of all students in aggregate in the presence of bias and any debiasing. In the fair matching (where $\phi$ is the identity map), a disadvantaged student is matched to school with rank $\mu(x)=\mbb{F\geq x}$, and in the ranking with bias and debiasing of $S$, they are matched to $\mu_\phi(x)=\mbb{(1-p)F+pH\geq\phi(x)}$. We therefore have under $S$ the displacement
\normalsize{\begin{align*}
\disp_S(x)&=\mu_\phi(x)-\mu(x)= \piecewise{
p\sbr{\mbb{F\leq x}-\mbb{\phi(F)\leq x}}, & x\in S, \\
\mbb{F\leq x}-\sbr{p\mbb{\phi(F)\leq x-\gamma}+(1-p)\mbb{F\leq x-\gamma}}, & x\not\in S.}
\end{align*}
}
 
By carefully dividing into the cases where $Y_2-Y_1<\gamma$ and $Y_2-Y_1\geq\gamma$, one can show

{\small{\begin{align}
    \disp_S(x)
&=(1-p)\mbb{F\in[x-\gamma,x]}\indc_{\set{x\not\in S}}+\piecewise{
-p\mbb{F\in[Y_2,\max\set{x+\gamma,Y_2}]}, & x\in(Y_1,Y_2), \\
p\mbb{F\in[\max\set{x-\gamma,Y_1},Y_2]}, & x\in[Y_2,Y_2+\gamma). \\
}\label{eq:disp_S}
\end{align}

}}
 
We are now ready to complete the proof. Let $S=[Y_1,Y_2]$ be some interval to debias, then~\eqref{eq:disp_S} implies that $m_\theemptyset(x)=m_S(x)$ for $x\not\in(Y_1,Y_2+\gamma)$,
$m_S(x)=0$ for $x\in(Y_1,Y_2)$,
$m_S(x)\geq m_\theemptyset(x)$ for $x\in[Y_2,Y_2+\gamma)$, and
that $m_S(x)$ is decreasing for $x\geq\max\set{1+\gamma,Y_2}$.

Further, if $1+\gamma\not\in S$ then $\max_x m_S(x)\geq \max_x m_\theemptyset(x)$. To see this, suppose $Y_1>1+\gamma$, then the result follows because $\max_x m_\theemptyset(x)=m_\theemptyset(1+\gamma)$. Otherwise if $Y_2\in[1,1+\gamma]$, we must have $1+\gamma\in[Y_2,Y_2+\gamma]$, so $\max_x m_S(x)\geq m_S(1+\gamma)\geq m_\theemptyset(1+\gamma)=\max_x m_\theemptyset(x)$.

We therefore know that the optimal $S\in\mcS^c(\hat c)$ contains $1+\gamma$ and that it minimizes $\max\set{m_S(Y_1),m_S(Y_2)}$. Now write
\begin{align*}
m_S(Y_1)&=(1-p)\mbb{F\in[Y_1-\gamma,Y_1]}=(1-p)\mbb{F\leq Y_1}, \\
m_S(Y_2)&=(1-p)\mbb{F\in[Y_2-\gamma,Y_2]}+p\mbb{F\in[\max\set{Y_2-\gamma,Y_1},Y_2]}.
\end{align*}
Recall that $\mbb{F\in[Y_1,Y_2]}=\hat c$. We need $m_S(Y_1)=m_S(Y_2)$, so write
\begin{align*}
(1-p)\mbb{F\leq Y_1}&=(1-p)\mbb{F\in[Y_2-\gamma,Y_2]}+p\mbb{F\in[\max\set{Y_2-\gamma,Y_1},Y_2]} \\
\iff (1-p)(1-\hat c)&=(1-p)\mbb{F\geq Y_2-\gamma}+p\mbb{F\in[\max\set{Y_2-\gamma,Y_1},Y_2]} \\
&=(1-p)\mbb{F\geq Y_2-\gamma}+p(\mbb{F\geq\max\set{Y_2-\gamma,Y_1}}-\mbb{F\geq Y_2}) \\
&=(1-p)\mbb{F\geq Y_2-\gamma}+p(\min\set{\mbb{F\geq Y_2-\gamma},\mbb{F\geq Y_1}}-\mbb{F\geq Y_2}) \\
&=\min\set{\mbb{F\geq Y_2-\gamma}-p\mbb{F\geq Y_2},(1-p)\mbb{F\geq Y_2-\gamma}+p\hat c}.
\end{align*}
Since both terms inside the minimum are decreasing in $Y_2$, the $Y_2$ that solves this equation is given by $\min\set{U_1,U_2}$ where $U_1$ solves $(1-p)(1-\hat c)=\mbb{F\geq U_1-\gamma}-p\mbb{F\geq U_1}$, and $U_2$ solves $(1-p)(1-\hat c)=(1-p)\mbb{F\geq U_2-\gamma}+p\hat c$. These are exactly the expressions sought for in the theorem.
 \endproof
 
\begin{table}[ht]
\centering
{\small{\begin{tabular}{l|c|c|c}
\textbf{Model} & \textbf{Metric} & \textbf{Value of Metric} & \textbf{Beta fit $\beta$} \\
\hline
Additive & KL-divergence & 0.0321 & -36.9 \\
Multiplicative & KL-divergence & \textbf{0.0293} & 0.902 \\ \hline
Additive & Wasserstein distance & 9.52 & -49.0 \\
Multiplicative & Wasserstein distance & \textbf{5.36} & 0.882 \\
\end{tabular}}}
\caption{Comparison of Additive and Multiplicative Models using KL-divergence and Wasserstein Distance Metrics}\label{tab:comparison}
\end{table}
We remark that it is possible to similarly compute the optimal DDS under the PAUC measure. In the proof of Theorem~\ref{thm:optimal-range-mm-additive} we gave an expression for the mistreatment of a given student, and one can then integrate this against $F_1$ to compute PAUC. One easily argues that this expression has a global minimum, but the resultant expressions do not seem to be amenable to analytical computations with a clean final statement.

 We close by noting that we fitted both the multiplicative as well as the additive model to our real data, and as measured by both Wasserstein distance and KL-divergence, the multiplicative model finds a slightly better fit, leading to us focusing on this model (see Table \ref{tab:comparison}).

\subsection{Impact of model misspecification}\label{appx:misspecified}

We next study the robustness of our framework under model misspecification. That is, we investigate the impact of applying our simple model of constant multiplicative bias when the true process by which bias arises is more complicated. In particular, we study additive models and models where idiosyncratic randomness exists within the bias factor or potentials of disadvantaged students. Based on computational experiments we show that our main takeaway holds, and that applying our results would lead to little efficiency loss except in the case of very high randomness.

\paragraph{Setup:} We generate simulated data with parameters chosen to match those we fit to our real data. In the language of Section~\ref{sec:market}, all students $\theta\in\Theta$ have a true potential $Z(\theta)$ sampled i.i.d. from a $\Pareto(1,\alpha)$ distribution with $\alpha=9$, and we identify students with their potential, so we write $\theta$ for their true potential. A proportion $p=0.3$ of students is disadvantaged and they appear at a perceived potential $\hat Z(\theta)$, where $\hat Z$ is some random variable with $\hat Z(\theta)\leq\theta$. We study various models for $\hat Z(\cdot)$.

The central planner has a budget $\hat c$ and applies our model with $\hat Z(\theta)=\beta\theta$ to choose an interval\footnote{We assume the central planner cannot observe the true potentials and must therefore debias disadvantaged students chosen based on perceived potentials. For the case of uniform multiplicative bias, this makes no difference, but when the bias process has randomness, this is an important detail.} of disadvantaged students to debias as instructed by Theorem~\ref{thm:optimal-range-pauc}. We call this the \emph{theoretical} debias interval. Because of randomness, it does not make sense to measure the maximum mistreatment, so we concentrate solely on the positive area under the mistreatment curve (PAUC). We then compare the theoretical interval to the optimal empirical interval if the full bias process were known to the central planner a priori. We compute such an interval using grid search, which we call the \emph{empirical} debias interval.

\paragraph{Models:} For each model, we let $\eta$ be some fixed parameter. We report results for the cases where the true bias process takes each of the following forms:
\begin{enumerate}
\item $\hat Z(\theta)=\theta-\eta$, a deterministic additive model;
\item $\hat Z(\theta)=(\eta+\eps)\theta$ for $\eps\sim\mbox{Normal}(0,.02)$, minor Gaussian noise in bias factor;
\item $\hat Z(\theta)=(\eta+\eps)\theta$ for $\eps\sim\mbox{Uniform}(-.05,.05)$, minor uniform noise in bias factor;
\item $\hat Z(\theta)=\theta-\eta+\eps$ for $\eps\sim\mbox{Normal}(0,.1)$, additive, medium Gaussian noise in bias factor;
\item $\hat Z(\theta)=\eta\theta+\eps$ for $\eps\sim\mbox{Normal}(0,.1)$, medium Gaussian noise in potential;
\item $\hat Z(\theta)=(\eta+\eps)\theta$ for $\eps\sim\mbox{Uniform}(-.15,.15)$, medium uniform noise in bias factor;
\item $\hat Z(\theta)=\eta\theta+\eps$ for $\eps\sim\mbox{Uniform}(-.3,.3)$, large uniform noise in potential; and
\item $\hat Z(\theta)=(\eta+\eps)\theta$ for $\eps\sim\mbox{Uniform}(-.3,.3)$, large uniform noise in bias factor.
\end{enumerate}
The first model in particular is the additive model studied in EC~\ref{sec:additive-bias}, and the fourth model is exactly the statistical discrimination model of~\cite{phelps1972statistical}. All models except the first can be interpreted as adaptations of a statistical discrimination model, as they contain the key feature of increased variance of the disadvantaged group. We choose for each model the fixed parameter $\eta$ in such a way that if the central planner would apply our theoretical model of constant multiplicative bias, they would fit exactly $\beta=0.88$ as the Wasserstein metric minimizing parameter. This yields a set of experiments that can be readily compared. The best fit for the $\eta$ parameter for each model is shown in Table~\ref{table:statdisc-best-fit}.

\begin{table}[!h]
\centering
\setlength{\tabcolsep}{8pt}
{\small{\begin{tabular}{l c c c c c c c c}
Model & 1 & 2 & 3 & 4 & 5 & 6 & 7 & 8 \\
\midrule
Parameter & 0.130 & 0.878 & 0.876 & 0.151 & 0.860 & 0.857 & 0.843 & 0.848 \\
Error & 0.011 & 0.002 & 0.004 & 0.046 & 0.037 & 0.035 & 0.092 & 0.106 \\ 
\end{tabular}}}
\smallskip
\caption{Best-fits for the model parameter $\eta$. 
The error row shows the minimum achieved Wasserstein distance between the data under the assumed true bias distribution, and the simplified theoretical model (with Pareto parameters $\alpha=9$ and $\beta=0.88$).
}\label{table:statdisc-best-fit}
\end{table}

\paragraph{Simulations:} We perform experiments with two budgets, $\hat c=0.1$ and $\hat c=0.4$, and run simulations with 1 million students in order to adequately approximate the continuous market. 
Many of these models increase the variance of the distribution of disadvantaged student scores, so the theoretical interval would often debias a significantly smaller proportion of students than allowed by the budget, naturally leading to a lower PAUC reduction and making comparison difficult. Because of this phenomenon, we fix the upper endpoint of the theoretical interval, but choose the lower endpoint such that it fills up the budget.

Table~\ref{table:statdisc-c1} and Table~\ref{table:statdisc-c4} summarize the results for $\hat c=0.1$ and $\hat c=0.4$ respectively. For each model, we report the aggregate mistreatment as measured by the PAUC metric (introduced in Section~\ref{sec:voucher}) for three cases: no debiasing, under the empirically optimal debiasing, and under the theoretically optimal debiasing. We report the reduction in PAUC given by both debiasing methods as well as their difference. Based on our computations, we conclude that our results are highly robust to model misspecification.

\normalsize{\begin{table}[!h]
\centering
{\small{\begin{tabular}{l ccc cc c}
\toprule
Model & \multicolumn{3}{c}{PAUC} & \multicolumn{2}{c}{PAUC Reduction} & Difference \\
\cmidrule(lr){2-4} \cmidrule(lr){5-6} & No debiasing & Empirical & Theoretical & Empirical & Theoretical & \\
\midrule
1 & 0.2261 & 0.1870 & 0.1870 & 17.28\% & 17.28\% & 0.00\% \\
2 & 0.2398 & 0.2009 & 0.2009 & 16.21\% & 16.20\% & 0.00\% \\
3 & 0.2406 & 0.2029 & 0.2029 & 15.67\% & 15.65\% & 0.01\% \\
4 & 0.2276 & 0.1978 & 0.1980 & 13.12\% & 13.00\% & 0.12\% \\
5 & 0.2406 & 0.2105 & 0.2108 & 12.50\% & 12.39\% & 0.11\% \\
6 & 0.2395 & 0.2138 & 0.2149 & 10.72\% & 10.27\% & 0.45\% \\
7 & 0.2360 & 0.2048 & 0.2057 & 13.23\% & 12.82\% & 0.41\% \\
8 & 0.2337 & 0.2033 & 0.2042 & 13.03\% & 12.64\% & 0.38\% \\
\bottomrule
\end{tabular}}}
\smallskip
\caption{Comparison of PAUC reductions between theoretically and empirically optimal intervals for $\hat c=0.1$.}
\label{table:statdisc-c1}
\end{table}}

\paragraph{Low Budget:} In the small budget case ($\hat c=0.1$, see Table~\ref{table:statdisc-c1}), the difference in using the empirically optimal and the theoretically optimal debiasing intervals is minuscule in every case. In the case of largest difference in PAUC reduction, the empirical interval is able to reduce PAUC by 10.72\%, whereas the theoretically optimal interval would have reduced it by 10.27\%, a difference of only 0.45\%.

\normalsize{\begin{table}[!h]
\centering
{\small{\begin{tabular}{l ccc cc c}
\toprule
Model & \multicolumn{3}{c}{PAUC} & \multicolumn{2}{c}{PAUC Reduction} & Difference \\
\cmidrule(lr){2-4} \cmidrule(lr){5-6} & No debiasing & Empirical & Theoretical & Empirical & Theoretical & \\
\midrule
1 & 0.2261 & 0.0850 & 0.0850 & 62.39\% & 62.38\% & 0.00\% \\
2 & 0.2398 & 0.0994 & 0.0994 & 58.56\% & 58.55\% & 0.01\% \\
3 & 0.2406 & 0.1062 & 0.1064 & 55.86\% & 55.80\% & 0.06\% \\
4 & 0.2276 & 0.1146 & 0.1194 & 49.66\% & 47.54\% & 2.12\% \\
5 & 0.2406 & 0.1254 & 0.1311 & 47.87\% & 45.51\% & 2.36\% \\
6 & 0.2395 & 0.1284 & 0.1402 & 46.37\% & 41.47\% & 4.90\% \\
7 & 0.2360 & 0.1015 & 0.1202 & 56.99\% & 49.08\% & 7.92\% \\
8 & 0.2337 & 0.0991 & 0.1187 & 57.62\% & 49.22\% & 8.40\% \\
\bottomrule
\end{tabular}}}
\smallskip
\caption{Comparison of PAUC reductions between theoretically and empirically optimal intervals for $\hat c=0.4$.}
\label{table:statdisc-c4}
\end{table}}

\paragraph{High Budget:} For the large budget case ($\hat c=0.4$, see Table~\ref{table:statdisc-c4}), the difference is more pronounced, yet still limited. 
The difference is negligible for models 1--3. Medium Gaussian noise in either potential or bias factor causes a small difference in PAUC reduction, in the order of 2--2.5\%. In the case of medium uniform noise in bias factor, the difference starts being more noticeable at 4.9\%, and with the cases of large uniform noise in potential or bias, the difference goes to 7.92\% and 8.4\% respectively.

\begin{figure}[h]
    \centering
    \includegraphics[width=.55\textwidth]{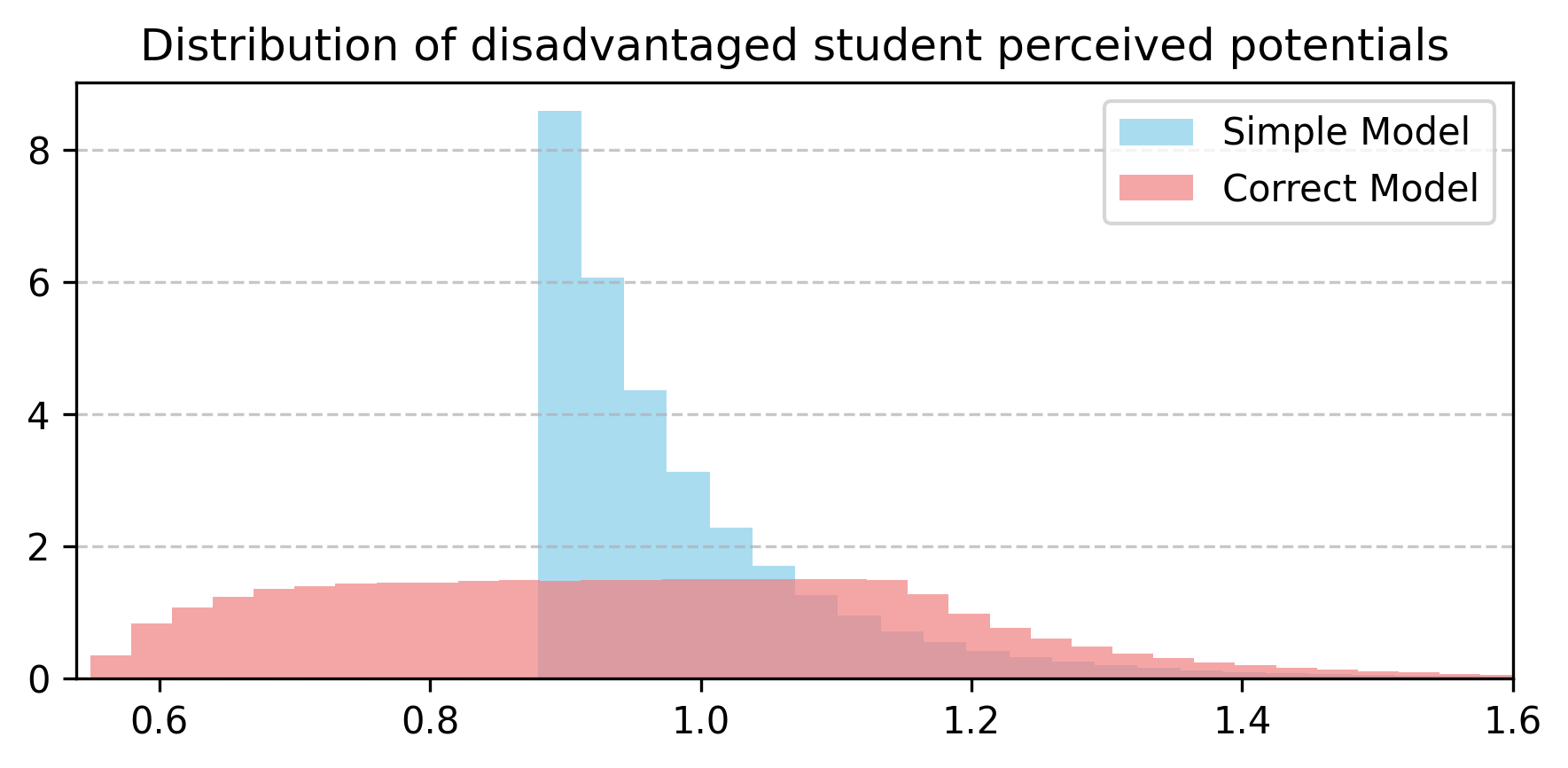}
    \caption{Difference in disadvantaged student perceived potentials for model 8, comparing the naive model and the true model.}\label{fig:pp-diff-appx}
\end{figure}

\paragraph{Most Misspecified Case:} In the worst case, the absolute difference in PAUC is 0.0196, meaning that on average, disadvantaged students in the correctly specified model achieve a ranking of approximately 2 percentage points higher than if the central planner were to apply the misspecified simpler model. {We note that to achieve such a difference, the model has to be highly misspecified: for example, our model would predict the perceived potentials of disadvantaged students to have mean 0.99 and standard deviation 0.125, whereas the correct model in this case has mean 0.95 and standard deviation 0.231. One can see this difference in Figure~\ref{fig:pp-diff-appx}, one would clearly observe that the model is not appropriate.}

\paragraph{Additive Case:} Our simple model has excellent performance under the case that the true model is any of models 1--3. In particular, in the case of the true model being additive, or the bias factor being slightly idiosyncratic, our model predictions apply virtually unchanged.

\paragraph{Incentive Compatibility:} 
We assume in the section that the participants in the process, such as the central planner, observe only perceived potentials and make decisions based on them. If the student knows only their perceived potential\footnote{One could also consider the question of whether disadvantaged students would benefit from being able to manipulate their true potentials before the randomized bias process. This is certainly an interesting mathematical question and could easily be computed, but it is nonsensical in the present setup to assume that students are able to manipulate their true potential before the bias process takes place. One could imagine that the bias occurs in the testing process, but then randomness should be present for both disadvantages as well as non-disadvantaged students.} and is able to manipulate it, then this case is no different to the treatment of incentive-compatibility in the main body. Voucher allocations that use a strict cutoff are always non-incentive-compatible.

\paragraph{Fit to real data:} To complement the synthetic comparison above, we also fit each model family directly to the real data from the case study in Section~\ref{sec:experiments}. Models 3, 6 and 8 above share the functional form $(\eta+\eps)\theta$ with $\eps\sim\mbox{Uniform}(-w,w)$,  differing only in the value of $w$. We therefore group the eight models into six families, and report them alongside the deterministic multiplicative model $\hat Z(\theta)=\eta\theta$ for reference. For each family we fit both the location parameter $\eta$ and, where applicable, the noise width. As in the main body, we perform this fitting by assuming that the distribution of the true potentials of non-disadvantaged and disadvantaged students ought to match. We then apply the bias function on the scores of the non-disadvantaged student scores, and find the parameters that most closely match this distribution with that of the unmodified disadvantaged student scores\footnote{This is therefore the reverse to how we fit in the main body, but is equivalent under the assumption. We have to perform it this way as the bias functions in this section include a noise component, which cannot be ``undone''.}. Table~\ref{table:statdisc-doe-fit} summarizes the best fits on both the full sample and on the truncated sample used in the main body\footnote{In the main body, we truncate the distributions to only students whose true potential under the $\beta=0.88$ case would exceed the cutoff of 475, in order to produce a balanced market.}. Across all multiplicative families, the best-fit location parameter is essentially indistinguishable from the value of $0.88$ assumed throughout the paper, with Wasserstein error within $0.001$ of the best overall fit. Moreover, in the families where the noise term is added to the output, the best-fit noise width collapses to zero, so the data is not well explained by homoskedastic noise. This lends further empirical support to our choice of the simple multiplicative bias model in the main analysis.

\begin{table}[!h]
\centering
\small
\setlength{\tabcolsep}{5pt}
\begin{tabular}{l c rrr rrr}
\toprule
Family & Model & \multicolumn{3}{c}{Full sample} & \multicolumn{3}{c}{Truncated sample} \\
\cmidrule(lr){3-5} \cmidrule(lr){6-8}
 & & $\eta$ & $w$ & Wass. & $\eta$ & $w$ & Wass. \\
\midrule
$\hat Z(\theta)=\eta\theta$                                                    & ---   & 0.882 & ---   & 0.010 & 0.895 & ---   & 0.013 \\
$\hat Z(\theta)=\theta-\eta$                                                   & 1     & 0.103 & ---   & 0.020 & 0.116 & ---   & 0.005 \\
$\hat Z(\theta)=(\eta+\eps)\theta,\ \eps\sim\mbox{Normal}(0,w)$            & 2     & 0.882 & 0.053 & 0.009 & 0.896 & 0.046 & 0.009 \\
$\hat Z(\theta)=(\eta+\eps)\theta,\ \eps\sim\mbox{Uniform}(-w,w)$               & 3,6,8 & 0.882 & 0.092 & 0.009 & 0.897 & 0.075 & 0.009 \\
$\hat Z(\theta)=\theta-\eta+\eps,\ \eps\sim\mbox{Normal}(0,w)$             & 4     & 0.104 & 0.001 & 0.020 & 0.116 & 0.009 & 0.005 \\
$\hat Z(\theta)=\eta\theta+\eps,\ \eps\sim\mbox{Normal}(0,w)$              & 5     & 0.882 & 0.001 & 0.010 & 0.895 & 0.046 & 0.010 \\
$\hat Z(\theta)=\eta\theta+\eps,\ \eps\sim\mbox{Uniform}(-w,w)$                 & 7     & 0.882 & 0.001 & 0.010 & 0.896 & 0.071 & 0.010 \\
\bottomrule
\end{tabular}
\smallskip
\caption{Best-fits of each family's location parameter ($\eta$) and noise width ($\sigma$ or $w$) against the real data, for the full sample and for the truncated sample used in the main body, together with the corresponding Wasserstein error.}\label{table:statdisc-doe-fit}
\end{table}

\end{document}

%% file: macros.tex
\usepackage[utf8]{inputenc}
\usepackage{amsmath,amsthm,amssymb}
\usepackage{thmtools, thm-restate}

\usepackage{graphicx}
\usepackage{subcaption}

\usepackage{xcolor}

\usepackage{hyperref}
\hypersetup{colorlinks=true,linkcolor=blue,citecolor=blue,urlcolor=blue}

\usepackage{soul}

\usepackage{color-edits}
\addauthor{eb}{red}
\addauthor{ce}{blue}
\addauthor{yf}{purple}

\usepackage[normalem]{ulem}

\usepackage{algorithm}
\usepackage{algpseudocode}
\algdef{SE}[DOWHILE]{Do}{doWhile}{\algorithmicdo}[1]{\algorithmicwhile\ #1}%

\usepackage{enumerate}

\newcommand{\abs}[1]{|{#1}|}

\newcommand{\mb}{\mathbf}

\theoremstyle{plain}
\newtheorem{theorem}{Theorem}

\newtheorem{lemma}[theorem]{Lemma}
\newtheorem{proposition}[theorem]{Proposition}

\theoremstyle{definition}
\newtheorem{definition}[theorem]{Definition}
\newtheorem{example}[theorem]{Example}

\theoremstyle{remark}

\renewcommand{\[}{\begin{equation}}
\renewcommand{\]}{\end{equation}}

\newcommand{\bea}{\begin{eqnarray}}
\newcommand{\eea}{\end{eqnarray}}

\newcommand{\R}{\mathbb{R}}

\newcommand{\br}{\mathbf{r}}

\newcommand{\eps}{\varepsilon}

\newcommand{\argmin}{\mathop{\mathrm{argmin}}}

%% file: refs.bib
@unpublished{nick,
    author  = "Nick Arnosti",
    title   = "A continuum model of stable matchings with finite capacities",
    year    = "2019",
    note    = "Talk at Simons Institute for the Theory of Computing",
}

@article{roth1992two,
  title={Two-sided matching},
  author={Roth, Alvin E and Sotomayor, Marilda},
  journal={Handbook of game theory with economic applications},
  volume={1},
  pages={485--541},
  year={1992},
  publisher={Elsevier}
}

@inproceedings{emelianov2020fair,
  title={On fair selection in the presence of implicit variance},
  author={Emelianov, Vitalii and Gast, Nicolas and Gummadi, Krishna P and Loiseau, Patrick},
  booktitle={Proceedings of the 21st ACM Conference on Economics and Computation},
  pages={649--675},
  year={2020}
}

@inproceedings{garg2021standardized,
  title={Standardized tests and affirmative action: The role of bias and variance},
  author={Garg, Nikhil and Li, Hannah and Monachou, Faidra},
  booktitle={Proceedings of the 2021 ACM Conference on Fairness, Accountability, and Transparency},
  pages={261--261},
  year={2021}
}

@article{quinn2017implicit,
  title={Implicit racial bias in medical school admissions},
  author={Quinn Capers, IV and Clinchot, Daniel and McDougle, Leon and Greenwald, Anthony G},
  journal={Academic Medicine},
  volume={92},
  number={3},
  pages={365--369},
  year={2017},
  publisher={LWW}
}

@inproceedings{celis2021effect,
  title={The Effect of the {R}ooney {R}ule on Implicit Bias in the Long Term},
  author={Celis, L Elisa and Hays, Chris and Mehrotra, Anay and Vishnoi, Nisheeth K},
  booktitle={Proc.~of the 2021 ACM Conference on Fairness, Accountability, and Transparency},
  pages={678--689},
  year={2021}
}

@inproceedings{celis2020interventions,
  title={Interventions for ranking in the presence of implicit bias},
  author={Celis, L Elisa and Mehrotra, Anay and Vishnoi, Nisheeth K},
  booktitle={Proceedings of the 2020 Conference on Fairness, Accountability, and Transparency},
  pages={369--380},
  year={2020}
}

@article{shapirowang,
  title={Amid Racial Divisions, Mayor’s Plan to Scrap Elite School Exam Fails},
  author={Shapiro, Eliza and Wang, Vivian},
  journal={{N}ew {Y}ork {T}imes},
 year={2019}
}

@article{boschma2016concentration,
  title={The concentration of poverty in {A}merican schools},
  author={Boschma, Janie and Brownstein, Ronald},
  journal={The Atlantic},
  volume={29},
  year={2016}
}

@article{gratz, 
title={{Gratz v. Bollinger, 539 U.S. 244 (2003).}},
author={{Gratz v. Bollinger}},
year={2003}
}

@misc{DOE,
    author    = "{NYC DOE}",
    title     = "2019 {NYC} {High} {S}chool {D}irectory",
    year     = "2019",
    howpublished = "\url{https://bigappleacademy.com/wp-content/uploads/2018/06/HSD_2019_ENGLISH_Web.pdf}"
}

@misc{doe-disadvantaged,
    author = "{NYC DOE}",
    title = "Specialized High Schools Proposal",
    year = "2018",
    howpublished = "\url{https://www.schools.nyc.gov/docs/default-source/default-document-library/specialized-high-schools-proposal}"
}

@article{abdulkadirouglu2005new,
  title={The {N}ew {Y}ork {C}ity high school match},
  author={Abdulkadiro{\u{g}}lu, Atila and Pathak, Parag A and Roth, Alvin E},
  journal={{American Economic Review}},
  volume={95},
  number={2},
  pages={364--367},
  year={2005}
}

@article{abdulkadirouglu2005college,
  title={College admissions with affirmative action},
  author={Abdulkadiro{\u{g}}lu, Atila},
  journal={International Journal of Game Theory},
  volume={33},
  pages={535--549},
  year={2005},
  publisher={Springer}
}

@article{chade2014student,
  title={Student portfolios and the college admissions problem},
  author={Chade, Hector and Lewis, Gregory and Smith, Lones},
  journal={Review of Economic Studies},
  volume={81},
  number={3},
  pages={971--1002},
  year={2014},
  publisher={Oxford University Press}
}

@article{chan2003does,
  title={Does banning affirmative action lower college student quality?},
  author={Chan, Jimmy and Eyster, Erik},
  journal={American Economic Review},
  volume={93},
  number={3},
  pages={858--872},
  year={2003},
  publisher={American Economic Association}
}

@article{fershtman2021soft,
  title={``Soft''' Affirmative Action and Minority Recruitment},
  author={Fershtman, Daniel and Pavan, Alessandro},
  journal={American Economic Review: Insights},
  volume={3},
  number={1},
  pages={1--18},
  year={2021},
  publisher={American Economic Association}
}

@article{gale1962college,
  title={College admissions and the stability of marriage},
  author={Gale, David and Shapley, Lloyd S},
  journal={{The American Mathematical Monthly}},
  volume={69},
  number={1},
  pages={9--15},
  year={1962},
  publisher={Taylor \& Francis}
}

@article{azevedo2016supply,
  title={A supply and demand framework for two-sided matching markets},
  author={Azevedo, Eduardo M and Leshno, Jacob D},
  journal={{Journal of Political Economy}},
  volume={124},
  number={5},
  pages={1235--1268},
  year={2016},
  publisher={University of Chicago Press Chicago, IL}
}

@article{calsamiglia2010constrained,
  title={Constrained school choice: An experimental study},
  author={Calsamiglia, Caterina and Haeringer, Guillaume and Klijn, Flip},
  journal={American Economic Review},
  volume={100},
  number={4},
  pages={1860--74},
  year={2010}
}

@article{burgess2015parents,
  title={What parents want: School preferences and school choice},
  author={Burgess, Simon and Greaves, Ellen and Vignoles, Anna and Wilson, Deborah},
  journal={The Economic Journal},
  volume={125},
  number={587},
  pages={1262--1289},
  year={2015},
  publisher={Wiley Online Library}
}

@article{laverde2020unequal,
  title={Unequal Assignments to Public Schools and the Limits of School Choice},
  author={Laverde, Mariana},
  journal={Unpublished},
  year={2020}
}

@article{hastings2009heterogeneous,
  title={Heterogeneous preferences and the efficacy of public school choice},
  author={Hastings, Justine and Kane, Thomas J and Staiger, Douglas O},
  journal={NBER Working Paper},
  volume={2145},
  pages={1--46},
  year={2009}
}

@article{Eliza2,
  title={Racist? Fair? Biased? {A}sian-{A}merican Alumni Debate Elite High School Admissions},
  author={Shapiro, Eliza},
  journal={The New York Times Magazine},
  year={2019}
}

@article{Eliza,
  title={Segregation Has Been the Story of {N}ew {Y}ork {C}ity's Schools for 50 Years},
  author={Shapiro, Eliza},
  journal={The New York Times Magazine},
year={2019}
}

@article{kucsera2014new,
  title={{N}ew {Y}ork {S}tate's extreme school segregation: {I}nequality, inaction and a damaged future},
  author={Kucsera, John and Orfield, Gary},
  year={2014}
}

@article{nguyen2019stable,
  title={Stable matching with proportionality constraints},
  author={Nguyen, Th{\`a}nh and Vohra, Rakesh},
  journal={Operations Research},
  year={2019},
  publisher={INFORMS}
}

@article{biro2010college,
  title={The college admissions problem with lower and common quotas},
  author={Bir{\'o}, P{\'e}ter and Fleiner, Tam{\'a}s and Irving, Robert W and Manlove, David F},
  journal={Theoretical Computer Science},
  volume={411},
  number={34-36},
  pages={3136--3153},
  year={2010},
  publisher={Elsevier}
}

@article{clauset2009power,
  title={Power-law distributions in empirical data},
  author={Clauset, Aaron and Shalizi, Cosma Rohilla and Newman, Mark EJ},
  journal={SIAM review},
  volume={51},
  number={4},
  pages={661--703},
  year={2009},
  publisher={SIAM}
}

@article{corcoran2018pathways,
  title={Pathways to an elite education: Application, admission, and matriculation to {New York City}'s specialized high schools},
  author={Corcoran, Sean Patrick and Baker-Smith, E Christine},
  journal={Education Finance and Policy},
  volume={13},
  number={2},
  pages={256--279},
  year={2018},
  publisher={MIT Press}
}

@inproceedings{kleinberg2018selection,
  title={Selection Problems in the Presence of Implicit Bias},
  author={Kleinberg, Jon and Raghavan, Manish},
  booktitle={Proc.~of the 9th Innovations in Theoretical Computer Science Conference (ITCS 2018)},
  year={2018},
}

@article{hafalir2013effective,
  title={Effective affirmative action in school choice},
  author={Hafalir, Isa E and Yenmez, M Bumin and Yildirim, Muhammed A},
  journal={Theoretical Economics},
  volume={8},
  number={2},
  pages={325--363},
  year={2013},
  publisher={Wiley Online Library}
}

@article{tomoeda2018finding,
  title={Finding a stable matching under type-specific minimum quotas},
  author={Tomoeda, Kentaro},
  journal={Journal of Economic Theory},
  volume={176},
  pages={81--117},
  year={2018},
  publisher={Elsevier}
}

@inproceedings{Kumar2000,
  title={Fairness measures for resource allocation},
  author={Kumar, Amit and Kleinberg, Jon},
  booktitle={Proceedings 41st Annual Symposium on Foundations of Computer Science},
  pages={75--85},
  year={2000},
  organization={IEEE}
}

@inproceedings{Conitzer2019,
  title={Group fairness for the allocation of indivisible goods},
  author={Conitzer, Vincent and Freeman, Rupert and Shah, Nisarg and Vaughan, Jennifer Wortman},
  booktitle={Proceedings of the 33rd AAAI Conference on Artificial Intelligence (AAAI)},
  year={2019}
}

@article{salem2023secretary,
  title={Secretary problems with biased evaluations using partial ordinal information},
  author={Salem, Jad and Gupta, Swati},
  journal={Management Science},
  year={2023},
  publisher={INFORMS}
}

@misc{Dwork2018,
  title={Group fairness under composition},
  author={Dwork, Cynthia and Ilvento, Christina},
  year={2018},
  publisher={{FATML}}
}

@article{roughgarden2010algorithmic,
  title={Algorithmic game theory},
  author={Roughgarden, Tim},
  journal={Communications of the ACM},
  volume={53},
  number={7},
  pages={78--86},
  year={2010},
  publisher={ACM New York, NY, USA}
}

@inproceedings{Dwork2012,
  title={Fairness through awareness},
  author={Dwork, Cynthia and Hardt, Moritz and Pitassi, Toniann and Reingold, Omer and Zemel, Richard},
  booktitle={{Proceedings of the 3rd Innovations in Theoretical Computer Science conference}},
  pages={214--226},
  year={2012},
  organization={ACM}
}

@article{Sen1979,
  title={Equality of what?},
  author={Sen, Amartya},
  journal={The Tanner lecture on human values},
  volume={1},
  year={1979},
  publisher={Stanford University}
}

@article{kamada2024fair,
  title={Fair matching under constraints: Theory and applications},
  author={Kamada, Yuichiro and Kojima, Fuhito},
  journal={Review of Economic Studies},
  volume={91},
  number={2},
  pages={1162--1199},
  year={2024},
  publisher={Oxford University Press US}
}

@article{backes2012affirmative,
  title={Do affirmative action bans lower minority college enrollment and attainment?: Evidence from statewide bans},
  author={Backes, Ben},
  journal={Journal of Human resources},
  volume={47},
  number={2},
  pages={435--455},
  year={2012},
  publisher={University of Wisconsin Press}
}

@article{dee2011impact,
  title={The impact of {No Child Left Behind} on student achievement},
  author={Dee, Thomas S and Jacob, Brian},
  journal={Journal of Policy Analysis and management},
  volume={30},
  number={3},
  pages={418--446},
  year={2011},
  publisher={Wiley Online Library}
}

@article{greenwald1996effect,
  title={The effect of school resources on student achievement},
  author={Greenwald, Rob and Hedges, Larry V and Laine, Richard D},
  journal={Review of educational research},
  volume={66},
  number={3},
  pages={361--396},
  year={1996},
  publisher={Sage Publications Sage CA: Thousand Oaks, CA}
}

@book{folland1999real,
  title={Real analysis: modern techniques and their applications},
  author={Folland, Gerald B},
  volume={40},
  year={1999},
  publisher={John Wiley \& Sons}
}

@article{phelps1972statistical,
  title={The statistical theory of racism and sexism},
  author={Phelps, Edmund S},
  journal={The american economic review},
  volume={62},
  number={4},
  pages={659--661},
  year={1972},
  publisher={JSTOR}
}

@inproceedings{niu2022best,
  title={Best vs. all: Equity and accuracy of standardized test score reporting},
  author={Niu, Mingzi and Kannan, Sampath and Roth, Aaron and Vohra, Rakesh},
  booktitle={Proceedings of ACM FAccT 2022},
  pages={574--586},
  year={2022}
}

@article{long2004race,
  title={Race and college admissions: An alternative to affirmative action?},
  author={Long, Mark C},
  journal={review of Economics and Statistics},
  volume={86},
  number={4},
  pages={1020--1033},
  year={2004},
  publisher={MIT Press}
}

@inproceedings{kannan2019downstream,
  title={Downstream effects of affirmative action},
  author={Kannan, Sampath and Roth, Aaron and Ziani, Juba},
  booktitle={Proceedings of the Conference on Fairness, Accountability, and Transparency},
  pages={240--248},
  year={2019}
}

@article{coate1993will,
  title={Will affirmative-action policies eliminate negative stereotypes?},
  author={Coate, Stephen and Loury, Glenn C},
  journal={The American Economic Review},
  pages={1220--1240},
  year={1993},
  publisher={JSTOR}
}

@article{dessein2023test,
  title={Test-Optional Admissions},
  author={Dessein, Wouter and Frankel, Alex and Kartik, Navin},
  journal={arXiv preprint arXiv:2304.07551},
  year={2023}
}

@misc{TexasTop10_2024,
  author = {{Texas Comptroller of Public Accounts}},
  title = {Top 10\% Rule},
  year = {2024},
  url = {https://comptroller.texas.gov/programs/education/msp/funding/aid/state-programs/txttp.php}}

@inproceedings{hu2019disparate,
  title={The disparate effects of strategic manipulation},
  author={Hu, Lily and Immorlica, Nicole and Vaughan, Jennifer Wortman},
  booktitle={Proceedings of the Conference on Fairness, Accountability, and Transparency},
  pages={259--268},
  year={2019}
}

@inproceedings{liu2021test,
  title={Test-optional policies: Overcoming strategic behavior and informational gaps},
  author={Liu, Zhi and Garg, Nikhil},
  booktitle={Proc.~of the 1st ACM Conference on Equity and Access in Algorithms, Mechanisms, and Optimization},
  pages={1--13},
  year={2021}
}

@article{CityJournal2022,
  author = {Wai Wah Chin},
  title = {Equity and Excellence, Four Years Later},
  journal = {City Journal},
  year = {2022},
  month = {12},
  url = {https://www.city-journal.org/article/equity-and-excellence-four-years-later}}

@article{considine2002influence,
  title={The influence of social and economic disadvantage in the academic performance of school students in {A}ustralia},
  author={Considine, Gillian and Zappal{\`a}, Gianni},
  journal={Journal of sociology},
  volume={38},
  number={2},
  pages={129--148},
  year={2002},
  publisher={Sage Publications Sage CA: Thousand Oaks, CA}
}

@article{yue2018supplemental,
  title={Supplemental instruction: Helping disadvantaged students reduce performance gap},
  author={Yue, Hongtao and Rico, Ruby Sangha and Vang, Mai Kou and Giuffrida, Tosha Aquino},
  journal={Journal of Developmental Education},
  pages={18--25},
  year={2018},
  publisher={JSTOR}
}

@misc{shsatpipeline,
  author = {{New York City Independent Budget Office}},
  title = {{The Specialized High School Admissions pipeline}},
  year = {2024},
  url = {https://ibo.nyc.ny.us/iboreports/print-the-specialized-high-school-admissions-pipeline-june-2024.pdf}
}

@article{arrow1973theory,
  title={The Theory of Discrimination},
  author={Arrow, KJ},
  journal={Discrimination in Labor Markets},
  year={1973}
}

@article{fang2011theories,
  title={Theories of statistical discrimination and affirmative action: A survey},
  author={Fang, Hanming and Moro, Andrea},
  journal={Handbook of social economics},
  volume={1},
  pages={133--200},
  year={2011},
  publisher={Elsevier}
}

@article{lang2012racial,
  title={Racial discrimination in the labor market: Theory and empirics},
  author={Lang, Kevin and Lehmann, Jee-Yeon K},
  journal={Journal of Economic Literature},
  volume={50},
  number={4},
  pages={959--1006},
  year={2012},
  publisher={American Economic Association}
}

@article{aigner1977statistical,
  title={Statistical theories of discrimination in labor markets},
  author={Aigner, Dennis J and Cain, Glen G},
  journal={Ilr Review},
  volume={30},
  number={2},
  pages={175--187},
  year={1977},
  publisher={SAGE Publications Sage CA: Los Angeles, CA}
}

@incollection{fryer2011importance,
  title={The importance of segregation, discrimination, peer dynamics, and identity in explaining trends in the racial achievement gap},
  author={Fryer Jr, Roland G},
  booktitle={Handbook of Social Economics},
  volume={1},
  pages={1165--1191},
  year={2011},
  publisher={Elsevier}
}
